\documentclass[reqno]{amsproc}
\usepackage{amsmath,amsthm,amsrefs,verbatim}
\usepackage{amssymb}
\usepackage{graphicx}
\usepackage{bbm}
\usepackage{subfigure}
\usepackage[rgb]{xcolor}
\usepackage{mathrsfs}
\usepackage{comment}
\usepackage{stackrel} 
\usepackage{geometry}
\usepackage[all]{xy}
\usepackage{tikz}
\usepackage{wasysym} 
\usepackage{stmaryrd} 
\usepackage{cancel} 
\usepackage{float} 
\usepackage{numprint} 
\usepackage{enumitem} 
\usepackage{afterpage} 

\let\origbar\bar 
\let\origdot\dot
\usepackage{allrunes} 
\let\bar\origbar
\let\dot\origdot

\usetikzlibrary{calc,decorations.pathreplacing}
\usetikzlibrary{arrows,shapes}

\newcommand{\beq}{\begin{equation}}
\newcommand{\eeq}{\end{equation}}

\newcommand{\C}{\mathbb{C}}
\renewcommand{\N}{\mathbb{N}}
\renewcommand{\Q}{\mathbb{Q}}
\renewcommand{\R}{\mathbb{R}}

\renewcommand{\H}{\mathcal{H}}
\newcommand{\B}{\mathcal{B}}
\newcommand{\tr}{\mathrm{tr}}
\newcommand{\CE}{\mathcal{CE}}
\newcommand{\eps}{\varepsilon}
\renewcommand{\labelenumi}{(\alph{enumi})} 
\renewcommand{\theenumi}{(\alph{enumi})} 

\newcommand{\defin}{:=}
\newcommand{\disjcup}{\mathbin{\dot{\cup}}} 

\swapnumbers 

\definecolor{myurlcolor}{rgb}{0,0,0.4}
\definecolor{mycitecolor}{rgb}{0,0.5,0}
\definecolor{myrefcolor}{rgb}{0.5,0,0}
\usepackage[pagebackref]{hyperref}
\hypersetup{colorlinks,
linkcolor=myrefcolor,
citecolor=mycitecolor,
urlcolor=myurlcolor}

\newtheorem{theo}{Theorem}[subsection]
\newtheorem{thm}[theo]{Theorem}
\newtheorem{prop}[theo]{Proposition} 
\newtheorem{lemma}[theo]{Lemma}
\newtheorem{lem}[theo]{Lemma}
\newtheorem{cor}[theo]{Corollary}
\newtheorem{conj}[theo]{Conjecture}
\newtheorem{defn}[theo]{Definition}
\theoremstyle{definition} 
\newtheorem{rem}[theo]{Remark}
\newtheorem{ex}[theo]{Example}
\newtheorem{prob}[theo]{Problem} 

\numberwithin{equation}{section}
\renewcommand{\emph}[1]{\textbf{#1}}    
\newcommand{\emphalt}[1]{\textit{#1}}   

\newcommand{\decprob}[3]{\begin{quote}
\textbf{Problem name:} \texttt{#1}\\
\textbf{Input data:} #2,\\
\textbf{To be decided:} #3?
\end{quote}}

\begin{document}

\title{A Combinatorial Approach to Nonlocality and Contextuality}

\author{Antonio Ac\'in}
\address{ICFO--Institut de Ci\`encies Fot\`oniques, E--08860 Castelldefels, Barcelona, Spain and ICREA--Instituci\'o Catalana de Recerca i Estudis Avan{\c{c}}ats, 08010 Barcelona, Spain}
\email{antonio.acin@icfo.es}

\author{Tobias Fritz}
\address{Perimeter Institute for Theoretical Physics, Waterloo, Ontario, Canada}
\email{tfritz@perimeterinstitute.ca}

\author{Anthony Leverrier}
\address{INRIA Rocquencourt, Domaine de Voluceau, B.P. 105, 78153 Le Chesnay Cedex, France}
\email{anthony.leverrier@inria.fr}

\author{Ana Bel{\'e}n Sainz}
\address{ICFO--Institut de Ci\`encies Fot\`oniques, E--08860 Castelldefels, Barcelona, Spain}
\email{belen.sainz@icfo.es}

\date{\today}	

\thanks{We thank Mateus Ara{\'u}jo, Ad{\'a}n Cabello, Ravi Kunjwal, Simone Severini, Alexander Wilce, Andreas Winter, Elie Wolfe and Gilles Z\'emor for comments and discussion, Andr{\'a}s Salamon for help with a reference, and Will Traves for help on \texttt{MathOverflow}. Research at Perimeter Institute is supported by the Government of Canada through Industry Canada and by the Province of Ontario through the Ministry of Economic Development and Innovation. A.A. was supported by the ERC CoG QITBOX, the Spanish projects FOQUS and DIQIP and the John Templeton Foundation. T.F. was supported by the John Templeton foundation. Part of this work was done while A.L. was at the Institute for Theoretical Physics, ETH Z{\"u}rich. A.L. was supported by the Swiss National Science Foundation though the National Centre of Competence in Research ``Quantum Science and Technology'', by the CNRS through the PEPS ICQ2013 TOCQ, and through the European Research Council (grant No.~258932). A.B.S. was supported by the ERC SG PERCENT and by the Spanish projects FIS2010-14830 and FPU:AP2009-1174 PhD grant.}

\begin{abstract}
So far, most of the literature on (quantum) contextuality and the Kochen--Specker theorem seems either to concern particular examples of contextuality, or be considered as quantum logic. Here, we develop a general formalism for contextuality scenarios based on the combinatorics of hypergraphs which significantly refines a similar recent approach by Cabello, Severini and Winter (CSW). In contrast to CSW, we explicitly include the normalization of probabilities, which gives us a much finer control over the various sets of probabilistic models like classical, quantum and generalized probabilistic. In particular, our framework specializes to (quantum) nonlocality in the case of Bell scenarios, which arise very naturally from a certain product of contextuality scenarios due to Foulis and Randall. In the spirit of CSW, we find close relationships to several graph invariants. The recently proposed Local Orthogonality principle turns out to be a special case of a general principle for contextuality scenarios related to the Shannon capacity of graphs. Our results imply that it is strictly dominated by a low level of the Navascu{\'e}s--Pironio--Ac{\'i}n hierarchy of semidefinite programs, which we also apply to contextuality scenarios.

We derive a wealth of results in our framework, many of these relating to quantum and supraquantum contextuality and nonlocality, and state numerous open problems. For example, we show that the set of quantum models on a contextuality scenario can in general not be characterized in terms of a graph invariant.

In terms of graph theory, our main result is this: there exist two graphs $G_1$ and $G_2$ with the properties
\begin{align*}
\alpha(G_1) &= \Theta(G_1), & \alpha(G_2) &= \vartheta(G_2), \\[6pt]
\Theta(G_1\boxtimes G_2) & > \Theta(G_1)\cdot \Theta(G_2),& \Theta(G_1 + G_2) & > \Theta(G_1) + \Theta(G_2).
\end{align*}
\end{abstract}

\maketitle

\clearpage
\thispagestyle{empty}
\newgeometry{top=1.5cm,bottom=1.5cm,left=2.5cm,right=2.5cm}   
\tableofcontents
\restoregeometry

\newpage
\section{\textbf{Introduction}}

Much effort has been devoted to understanding the manifold counterintuitive aspects of quantum theory. In particular, this applies to the phenomena known as quantum \emph{nonlocality} and quantum \emph{contextuality}. Bell's theorem~\cite{Bell} shows that no theory can make the same predictions as quantum theory, while jointly satisfying the properties of \emphalt{realism}, \emphalt{locality} and \emphalt{free will}. This is often abbreviated to the statement that quantum theory displays \emph{nonlocality}\footnote{This terminology can be confusing, since all known fundamental interactions are of a local nature~\cite{Haag}, in a different sense of the term~\cite{Zeh}.}. Similarly, the Kochen--Specker theorem~\cite{KS} states that quantum theory is at variance with any attempt at assigning deterministic values to all observables in a way which would be consistent with the functional relationships between these observables predicted by quantum theory. This impossibility is generally known as \emph{contextuality}, since it means that any potential `hidden' predetermined value of an observable will necessarily have to depend on the \emphalt{context} in which it is probed.

It is often stated that nonlocality is, at the mathematical level, a particular case of contextuality. However, it is rarely made explicit what exactly this means. Moreover, the study of contextuality so far often seems to have been concerned with particular examples of contextuality and `small' proofs of the Kochen--Specker theorem~\cite{simplest}, while a general theory has hardly been developed. Some notable exceptions are the following:

\begin{enumerate}
\item\label{ts} The study of \emph{test spaces} in quantum logic~\cite{CMW,WilceHB},
\item Spekkens' work on \emph{measurement and preparation contextuality}~\cite{Spek,LSW},
\item\label{gt} The \emph{graph-theoretic} approach of Cabello, Severini and Winter (CSW)~\cite{CSW},
\item\label{st} The \emph{sheaf-theoretic} approach pioneered by Abramsky and Brandenburger~\cite{AB}.
\end{enumerate}

What we set out to do here is to develop a hypergraph-theoretic approach in the spirit of~\ref{ts} which comprises~\ref{gt} and~\ref{st} as special cases (see Section~\ref{CSWtransfer} and Appendix~\ref{reltosheaf}).

Although the test spaces from~\ref{ts} are usually considered in the context of quantum logic and state spaces, they serve equally well for the study of contextuality, which is intimately related. This is our first main theme: a test space can be considered as a \emph{contextuality scenario}, and this is the term we will use. As in~\ref{gt} and similar to~\ref{st}, we take a contextuality scenario to be a specification of a collection of measurements which says how many outcomes each measurement has and which measurements have which outcomes in common. We show how the \emph{Foulis--Randall product} of test spaces is the `correct' product of two or more contextuality scenarios, in the sense that it describes parallel execution of these scenarios and naturally incorporates the no-signaling condition. In particular, the \emph{Bell scenarios} which describe nonlocality turn out to be given by Foulis--Randall products. We define a \emph{probabilistic model} as an assignment of a probability to each outcome in such a way that the probabilities for the outcomes of any given measurement sum up to $1$. In the Bell scenario case, these are precisely the well-known \emph{no-signaling boxes}. One of our main results is a combinatorial characterization of those probabilistic models that are extreme points of the convex polytope of all probabilistic models in a given scenario; see Theorem~\ref{extchar}.

Our second main theme is to relate, again inspired by~\ref{gt}, contextuality scenarios and their probabilistic models to graph theory and invariants of graphs like the \emph{independence number}, \emph{Lov{\'a}sz number} and \emph{fractional packing number}. Our approach differs significantly from CSW's in two important respects. First and most importantly, we explicitly take into account the normalization of probabilities from the very beginning. In contrast to this, CSW were working with \emphalt{subnormalized} probabilities, which seems necessary in order to derive their relations to graph-theoretic invariants, but leads to shortcomings such as upper bounds on quantum nonlocality which are not always tight, exemplified by a higher-than-quantum violation of the $I_{3322}$ inequality~\cite{CSW}. We show that such relations still exist, even if one retains the normalization of probability. This gives us much finer quantitative information and control about contextuality. Second, while CSW study the maximal values of contextuality \emph{inequalities} for classical, quantum, and general probabilistic models, we consider the \emph{sets} of classical, quantum, and general probabilistic models themselves as the primary objects. While these two points of view are equivalent by duality, we believe that the latter is a more natural thing to do, since the actual quantities gathered e.g.~from an experiment are outcome probabilities rather than coefficients of some inequality; satisfaction of a predetermined inequality is sufficient, but not necessary, for the measured statistics to arise from a classical or quantum model. Also, taking this dual approach based on sets of models rather than inequalities is exactly what enables us to derive our relations to graph invariants while retaining the normalization of probabilities---doing this on the level of inequalities does not seem possible. Our dual approach also results in the relations to graph invariants being opposite: classical models are characterized in terms of the fractional packing number, while probabilistic models satisfying consistent exclusivity are characterized by the independence number; see Figure~\ref{chain-inc}. The relations obtained in~\cite{CSW} on the level of inequalities are exactly opposite.

\begin{figure}[t!!!] 
\centerline{\hspace{0cm}$$
\xymatrix{\overset{\substack{\text{classical}\\ \text{models}}}{\mathcal{C}(H)} \ar@{}[r]|(.37){\subseteq}_(.37){\ref{Qprop}\ref{CsubQ}} \ar@{<->}[d]^{\ref{Cvsfpn}} & \overset{\substack{\text{quantum}\\ \text{models}}}{\mathcal{Q}(H)\stackrel[\ref{convNPA}]{=}{}\mathcal{Q}_\infty(H)} \ar@{}[r]|(.63){\subseteq\subseteq}_(.63){\ref{QinQ1}} \ar@{<->}[d]^{\ref{Qinv}}  & \mathcal{Q}_{1}(H) \ar@{}[r]|{\subseteq}_{\ref{miguelito}} \ar@{<->}[d]^{\ref{Q1vsLov}} & \CE^\infty(H) \ar@{}[r]|{\subseteq \subseteq}_{\ref{LOincs}} \ar@{<->}[d]^{\ref{LOchar}} & \CE^1(H) \ar@{}[r]|{\subseteq} \ar@{<->}[d]^{\ref{LOchar}} & \overset{\substack{\text{probabilistic}\\ \text{models}}}{\mathcal{G}(H)} \\
\underset{\substack{\text{fractional}\\ \text{packing number}}}{\alpha^*} \ar@{}[r]|(.37){\geq} & \underset{\substack{\text{no graph}\\ \text{invariant}}}{-}\ar@{}[r]|(.63){\geq} & \underset{\substack{\text{Lov\'asz}\\ \text{number}}}{\vartheta} \ar@{}[r]|{\geq} & \underset{\substack{\text{Shannon}\\ \text{capacity}}}{\Theta} \ar@{}[r]|{\geq} & \underset{\substack{\text{independence}\\ \text{number}}}{\alpha}}
$$}
\caption{Chain of inclusions between our various sets of probabilistic models on a contextuality scenario $H$ (first row) and the corresponding inequalities between graph invariants (second row). The two inclusions marked `$\subseteq\subseteq$' actually each contain an infinite hierarchy of sets $\mathcal{Q}_n(H)$ and $\CE^n(H)$ indexed by $n\in\N$. We suspect that all inclusions in the first row are strict for some $H$, including the ones in both infinite hierarchies. All theorem numbers reference the corresponding statements and proofs in the main text.}\label{chain-inc}
\end{figure}
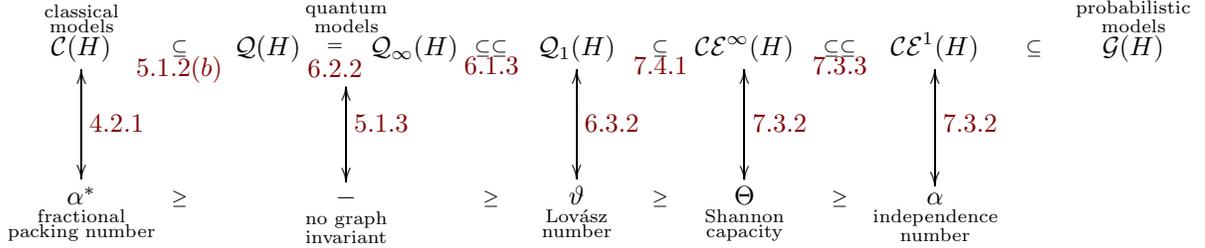


Figure~\ref{chain-inc} summarizes the sets of probabilistic models that we consider together with their relations to invariants of graphs. The classical $\mathcal{C}(H)$ corresponds to models which can be described in terms of noncontextual (deterministic) hidden variables on a contextuality scenario $H$. Similarly, $\mathcal{Q}(H)$ is the quantum set, defined in terms of quantum states on a Hilbert space and projective measurements. The $\mathcal{Q}_n(H)$ family comes from a hierarchy of semidefinite programs; this hierarchy characterizes $\mathcal{Q}$ in the sense that in the limit $n\to\infty$, we have $\mathcal{Q}_\infty(H)=\mathcal{Q}(H)$. The general probabilistic set $\mathcal{G}$ contains all models that are probabilistic models, that is probability assignments satisfying the normalization of probability for all measurements. Finally, the family of sets $\CE^n(H)$ arises from our third main theme.

This third main theme is the concept of \emph{Local Orthogonality (LO)} which was recently introduced in~\cite{LOfp} as an information-theoretic principle delimiting the set of quantum models $\mathcal{Q}(H)$ in Bell scenarios. We show how LO naturally arises in our formalism as a special case of a previously studied concept called \emph{Consistent Exclusivity (CE)}~\cite{Henson} or \emph{Specker's principle}~\cite{CabelloSP}. CE builds on the observation that compatibility of quantum observables is a binary property determined by \emphalt{pairwise} commutativity. It can be applied both on the single-copy level of a scenario, in which case we denote the principle as CE$^1$, and on the many-copy level when the same system is distributed among any number of parties, for which we write CE$^\infty$. This parallels the distinction between LO$^1$ and LO$^\infty$ that we made in~\cite{LOfp,LO2}. While CE$^1$ relates to the \emph{independence number} of a graph, CE$^\infty$ corresponds to the \emph{Shannon capacity} (in the sense of graph theory). This connection allows us to answer some open questions about LO$^\infty$. 
In particular, we show that LO$^\infty$, and more generally CE$^\infty$, does not characterize quantum models. In fact, CE$^\infty$ is satisfied for every probabilistic model which lies in $\mathcal{Q}_1$, a set of probabilistic models that contains the quantum set (usually strictly) and can be decided by means of a semidefinite program. Moreover, at least on some scenarios, there are probabilistic models which satisfy CE$^\infty$, but do not even lie in $\mathcal{Q}_1$. We also prove that the set of probabilistic models satisfying CE$^\infty$ is \emph{not convex}, which also implies that CE$^\infty$ can be \emph{activated}: there are pairs of probabilistic models both of which satisfy CE$^\infty$, although their product does not. These results relate to theorems on the Shannon capacity of graphs, some of which are new to this paper.

\subsection{Structure and contents of this paper}

We begin in Section~\ref{scenarios} by introducing test spaces as our notion of contextuality scenario. Later (in Section~\ref{products}), we will see that every Bell scenario is a contextuality scenario. We continue in Section~\ref{scenarios} by defining probabilistic models on a contextuality scenario; e.g.~for a Bell scenario, these are precisely the no-signaling boxes. Furthermore, we give an abstract characterization of extremal probabilistic models. We also introduce the \emph{non-orthogonality graph} of a contextuality scenario, whose graph-theoretical invariants are related to different sets of probabilistic models studied in the following sections.

In Section~\ref{products}, we consider products of contextuality scenarios corresponding to simultaneous measurements on spatially separated systems. We find the relevant product operation to be the Foulis--Randall product of test spaces. This product guarantees the no-signaling property for probabilistic models on the product scenario by, seemingly paradoxically, incorporating measurements \emphalt{with} communication. Figure~\ref{CHSH_sce} displays the CHSH scenario~\cite{CHSH} as a contextuality scenario.

In Section~\ref{cm}, we study classical models on contextuality scenarios. These are precisely those probabilistic models that can occur in a world described by noncontextual hidden variables. We further show how the weighted fractional packing number of the non-orthogonality graph detects the (non-)classicality of a probabilistic model.

In Section~\ref{qm}, we consider quantum models. We show that these cannot be characterized by a graph invariant of the non-orthogonality graph. On a product of two contextuality scenarios, the typical quantum models are those which arise from commuting observables for each component scenario. We show that all quantum models on a product scenario are in fact of this form, although the definition of product scenario does not directly impose this.

In Section~\ref{npahierarchy}, we show how to formulate a \emph{hierarchy of semidefinite programs characterizing quantum models} for contextuality scenarios. This can be regarded either as a generalization of the original hierarchy for quantum correlations in Bell scenarios~\cite{NPA0,NPA} or as a special case of the general hierarchy for noncommutative polynomial optimization~\cite{NPA2}. We characterize the first level of this hierarchy by the weighted Lov{\'a}sz number of the non-orthogonality graph and find a long list of equivalent reformulations.

In Section~\ref{LOsec}, we consider the principle of Local Orthogonality (LO) introduced in~\cite{LOfp} and show in which sense it arises from Consistent Exclusivity (CE)~\cite{Henson}. We explain how CE can be characterized in terms of the weighted independence number and the weighted Shannon capacity of the non-orthogonality graph. It turns out that the principle, even when applied on the level of distributed copies as CE$^\infty$, is weaker than $\mathcal{Q}_1$, the first level of the semidefinite hierarchy. We show that $\CE^\infty(H)$, the set of probabilistic models on a scenario $H$ satisfying CE$^\infty$, is in general not convex, and also that activation is possible: there are scenarios $H_A$ and $H_B$ and probabilistic models $p_A\in\CE^\infty(H_A)$ and $p_B\in\CE^\infty(H_B)$ such that $p_A\otimes p_B\not\in\CE^\infty(H_A\otimes H_B)$. We also discuss a proposal for an extension of the consistent exclusivity principle based on the ideas of~\cite{Yan} and show that it recovers exactly $\mathcal{Q}_1$. Finally, we observe that if the non-orthogonality graph is a \emph{perfect graph}, which frequently happens, then every probabilistic model satisfying CE$^1$ is classical and no interesting contextuality is possible in the given scenario. The strong perfect graph theorem then implies that a scenario can display (quantum) contextuality only if it has a certain odd cycle or odd anti-cycle structure.

In Section~\ref{complexity}, we study the complexity of various decision problems on contextuality scenarios. Our `inverse sandwich conjecture'~\ref{invsandwich} is an undecidability statement whose proof would have significant repercussions in $C^*$-algebra theory and quantum logic.

In Section~\ref{examples}, we discuss some further examples of contextuality scenarios and the various sets of probabilistic models associated to them, including a prescription for translating graph-based scenarios with subnormalization of probabilities (as in the CSW approach) into our framework.

In Appendix~\ref{appcap}, we discuss the graph theory relevant for the main text. In particular, we introduce graph-theoretical invariants for both unweighted and  weighted graphs and discuss their properties.

In Appendix~\ref{uwmainconjs}, we reformulate the examples of activation of CE$^{\infty}$ in terms of graph theory and show that there is a pair of graphs $G_1$ and $G_2$ with $\alpha(G_1) = \Theta(G_1)$ and $\alpha(G_2) = \vartheta(G_2)$, but $\Theta(G_1\boxtimes G_2) > \Theta(G_1)\cdot \Theta(G_2)$ and $\Theta(G_1 + G_2) > \Theta(G_1) + \Theta(G_2)$.

In Appendix~\ref{multiproducts}, we introduce the notions of virtual edge and completion of a contextuality scenario, and turn to discuss the Foulis--Randall products of three or more contextuality scenarios. 

Finally, in Appendix~\ref{reltosheaf}, we discuss how our approach, based on hypergraphs in which the vertices represent measurement outcomes, relates to the one of Abramsky and Brandenburger~\cite{AB}, which is based on hypergraphs in which vertices represent observables. We explain in which sense the two approaches are equivalent.

\newpage
\section{\textbf{Contextuality scenarios and their probabilistic models}}
\label{scenarios}

\subsection{Motivation: the Kochen--Specker theorem}

Cabello \textit{et al.}~\cite{Cab1,Cab2} showed that one can find $18$ vectors in $\C^4$ labeling the vertices of Figure~\ref{CKSfig} such that the four vectors associated to each one of the $9$ edges form an orthonormal basis. Together with the observation that there is no consistent way to label the vertices by $\{0,1\}$ such that every edge contains exactly one vertex labeled by $1$, this is a proof of the Kochen--Specker theorem for the four-dimensional Hilbert space $\C^4$. 

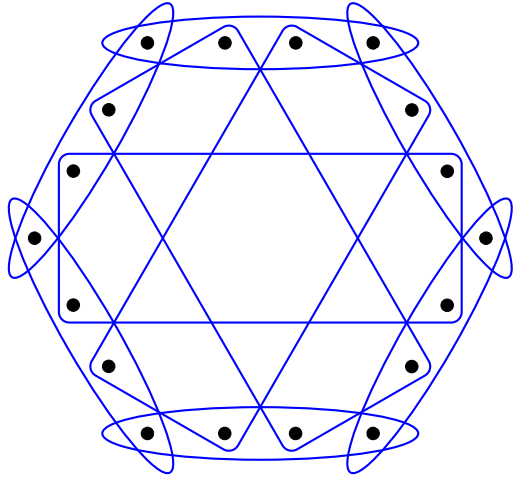
\begin{figure}
\begin{center}
\begin{tikzpicture}
\node[draw,shape=circle,fill,scale=.5] (a) at (0:3) {} ;
\node[draw,shape=circle,fill,scale=.5] (b) at (60:3) {} ;
\node[draw,shape=circle,fill,scale=.5] (c) at (120:3) {} ;
\node[draw,shape=circle,fill,scale=.5] (d) at (180:3) {} ;
\node[draw,shape=circle,fill,scale=.5] (e) at (240:3) {} ;
\node[draw,shape=circle,fill,scale=.5] (f) at (300:3) {} ;
\draw[white] (a) -- (b) node[pos=.333,draw,shape=circle,fill,scale=.5,black] (a1) {} node[pos=.5] (am) {} node[pos=.666,draw,shape=circle,fill,scale=.5,black] (a2) {} ;
\draw[white] (b) -- (c) node[pos=.333,draw,shape=circle,fill,scale=.5,black] (b1) {} node[pos=.5] (bm) {} node[pos=.666,draw,shape=circle,fill,scale=.5,black] (b2) {} ;
\draw[white] (c) -- (d) node[pos=.333,draw,shape=circle,fill,scale=.5,black] (c1) {} node[pos=.5] (cm) {} node[pos=.666,draw,shape=circle,fill,scale=.5,black] (c2) {} ;
\draw[white] (d) -- (e) node[pos=.333,draw,shape=circle,fill,scale=.5,black] (d1) {} node[pos=.5] (dm) {} node[pos=.666,draw,shape=circle,fill,scale=.5,black] (d2) {} ;
\draw[white] (e) -- (f) node[pos=.333,draw,shape=circle,fill,scale=.5,black] (e1) {} node[pos=.5] (em) {} node[pos=.666,draw,shape=circle,fill,scale=.5,black] (e2) {} ;
\draw[white] (f) -- (a) node[pos=.333,draw,shape=circle,fill,scale=.5,black] (f1) {} node[pos=.5] (fm) {} node[pos=.666,draw,shape=circle,fill,scale=.5,black] (f2) {} ;
\draw[thick,blue,rotate=120] (am) ellipse (2.1cm and .35cm) ;
\draw[thick,blue,rotate=180] (bm) ellipse (2.1cm and .35cm) ;
\draw[thick,blue,rotate=240] (cm) ellipse (2.1cm and .35cm) ;
\draw[thick,blue,rotate=300] (dm) ellipse (2.1cm and .35cm) ;
\draw[thick,blue,rotate=0] (em) ellipse (2.1cm and .35cm) ;
\draw[thick,blue,rotate=60] (fm) ellipse (2.1cm and .35cm) ;
\node (a3) at ($ (a2) + (10:.3)$) {} ;
\node (b3) at ($ (b2) + (70:.3)$) {} ;
\node (c3) at ($ (c2) + (130:.3)$) {} ;
\node (d3) at ($ (d2) + (190:.3)$) {} ;
\node (e3) at ($ (e2) + (250:.3)$) {} ;
\node (f3) at ($ (f2) + (310:.3)$) {} ;
\draw[thick,blue,rounded corners,rotate=150] (a3) rectangle (d3) ;
\draw[thick,blue,rounded corners,rotate=210] (b3) rectangle (e3) ;
\draw[thick,blue,rounded corners,rotate=270] (c3) rectangle (f3) ;
\end{tikzpicture}
\end{center}
\caption{The contextuality scenario $H_{\mathrm{KS}}$ proving the Kochen--Specker theorem \cites{Cab1,Cab2}. Each vertex represents a vector in $\C^4$, while each closed curve delimits a set of $4$ vertices. These sets are what we call `edges'.}
\label{CKSfig}
\end{figure}

Now what does the hypergraph of Figure~\ref{CKSfig} represent, operationally? This is what we would like to consider next.

\subsection{General definition}

Since each edge of Figure~\ref{CKSfig} stands for a basis in $\C^4$, we may think of an edge as representing a $4$-outcome measurement. Now every vertex occurs in two different such edges; in other words, some of the measurements share outcomes. The assumption of \emph{measurement noncontextuality}~\cite{Spek} means that any reasonable theory should represent such a shared outcome as a function from states to probabilities which does not depend on the particular measurement in which the outcome occurs.

Abstracting from this particular example to a general definition of \emph{contextuality scenario} means that we need to consider a mathematical structure containing a set of vertices, representing outcomes, and a collection of subsets of the vertices, representing measurements. Mathematically this is a \emph{hypergraph} $H$ with vertices $V(H)$ and edges $E(H)$. We therefore arrive at:

\begin{defn}
\label{defncs}
A \emph{contextuality scenario} is a hypergraph $H$ with set of vertices $V(H)$ and set of edges $E(H)\subseteq 2^{V(H)}$ such that $\bigcup_{e\in E(H)} e = V(H)$.
\end{defn}

In the following, we will use the terms \emph{vertex} and \emph{outcome} interchangeably with \emph{edge} over \emph{measurement}, respectively, while keeping in mind that the latter is the physical interpretation of the former, respectively.

The condition $\bigcup_{e\in E(H)} e=V(H)$ simply states that each outcome should occur in at least one measurement.

Typically, a contextuality scenario will satisfy the condition that if $e_1,e_2\in E(H)$ with $e_1\subseteq e_2$, then $e_1=e_2$, i.e.~there are no different edges one of which is contained in the others. All of our examples will have this property. The reason for this is that if we have different measurements $e_1$ and $e_2$ such that every outcome of $e_1$ is also an outcome of $e_2$, then the additional outcomes of $e_2$ necessarily have probability $0$ and can therefore be disregarded. In the literature on hypergraph theory, hypergraphs satisfying this condition that no edge is a subset of another are known as \emph{clutters}~\cite{EF70} or \emph{Sperner families}~\cite{Engel}. However, notably, the scenarios constructed in the proof of Theorem~\ref{extchar} do not have this property. So although one could include the condition that no edge should be contained in another one as an additional requirement for a hypergraph to be a contextuality scenario, this would complicate the proof of Theorem~\ref{extchar} without yielding simplifications of theorems or their proofs anywhere else. For this reason, we abstain from including this condition in Definition~\ref{defncs}.

In a typical scenario, the hypergraph $H$ is finite, meaning that $V(H)$ is a finite set, and this is the only case that we want to consider. It implies that $E(H)$ is finite as well.

Definition~\ref{defncs} or variants thereof have been considered before in the literature on contextuality and the Kochen--Specker theorem, e.g.~in~\cites{Tkadlec,PMMF}, and coincides with the notion of \emph{test space}~\cite{WilceHB} which had been introduced in~\cites{FR,RF} as \emph{(generalized) sample space}. In particular, the \emph{Greechie diagrams}~\cite{Greechie,ST} of quantum logic can all be regarded as contextuality scenarios.

On the other hand, Definition~\ref{defncs} differs from the formalisms proposed in~\cite{AB} and~\cites{CF,FC}. These works also provide a formalization of contextuality phenomena in terms of hypergraphs, but the vertices of the hypergraph represent observables rather than outcomes, while the edges stand for (maximal) jointly measurable sets of observables. See Appendix~\ref{reltosheaf} for a more detailed discussion.

\subsection{Non-orthogonality graphs} 

One of our main themes will be to relate properties of contextuality scenarios to some graph invariants. For this, we associate a specific graph to each scenario.

In the (hyper-)graph theory literature, one frequently considers the \emph{orthogonality graph} of a hypergraph (also referred to as its \emph{primal} or \emph{Gaifman graph}~\cite{gottlob}). The orthogonality graph of a hypergraph $H$ is obtained by declaring two vertices to be adjacent if and only if there exists an edge containing both. This coincides with the orthogonality relation present in the \emph{generalized sample spaces} of~\cite{FR}. Alternatively speaking, upon thinking of $H$ as an abstract simplicial complex with the edges as its facets, its orthogonality graph is the $1$-skeleton of this simplicial complex.

For the purpose of relating contextuality scenarios to the standard invariants of graph theory discussed in Appendix~\ref{appcap}, it turns out to be more convenient to consider the complement of the orthogonality graph. The drawback of this is that it may make some of our considerations sound more confusing, e.g.~the proof of Lemma~\ref{ortproduct}.

\begin{defn}[Non-orthogonality graph]
\label{defnno}
Let $H$ be a contextuality scenario. The \emph{non-orthogonality graph} $\mathrm{NO}(H)$ is the undirected graph with the same vertices as $H$ and adjacency relation
$$
u\sim v \:\Longleftrightarrow\: \not\exists e\in E(H) \textrm{ with } \{u,v\} \subseteq e .
$$
\end{defn}

We say that two different vertices $u$ and $v$ of $H$ are orthogonal, which we denote by $u \perp v$, if they are not adjacent in $\mathrm{NO}(H)$, i.e.~if they do belong to a common edge in $H$.

A possible interpretation of the non-orthogonality graph of a contextuality scenario is as a \emph{confusability graph} whose vertices correspond to events which can be confused with each other whenever they share an edge~\cite{Shannon}, that is because there is no measurement for which they appear as distinct outputs.

\subsection{Probabilistic models}

The definition of contextuality scenario is inherently operational: we think of the edges as all the possible measurements which can be conducted on a system. A consistent measurement statistic will assign a probability to each outcome, in such a way that the total probability for each measurement is $1$:

\begin{defn}
\label{defnpm}
Let $H$ be a contextuality scenario. A \emph{probabilistic model} on $H$ is an assignment $p:V(H)\to [0,1]$ of a probability $p(v)$ to each vertex $v\in V(H)$ such that
\beq
\label{probnorm}
\sum_{v\in e} p(v) = 1 \quad \forall e\in E(H).
\eeq
\end{defn}

It is important to keep in mind that each $p(v)$ is actually a \emphalt{conditional} probability: it stands for the probability of getting the outcome $v$ \emphalt{given that} a measurement $e$ containing $v$ is being conducted.

The set of all probabilistic models on $H$ is a convex subset of $\R^{V(H)}$, possibly empty, which we denote by $\mathcal{G}(H)$. This notation is supposed to suggest the reading `general probabilistic' in the sense of \emph{general probabilistic theories}~\cite{Barrett}. In the terminology of test spaces~\cite{Wilce2}, $\mathcal{G}(H)$ is the set of states over $H$; unfortunately, the term `probabilistic model' also exists in the test space formalism, but refers to a different concept.

In a concrete experiment, one will want to know whether the given measurement statistic, described by a probabilistic model $p$, is consistent with a certain theoretical framework. This is the main idea behind the various subsets of $\mathcal{G}(H)$ that we will define in the upcoming sections: a set $\mathcal{C}(H)$ of probabilistic models which can be explained in terms of an underlying classical state space, a set $\mathcal{Q}(H)$ of quantum models which can be explained using the mathematical formalism of quantum theory, and so on. Note that we are only concerned with the possibility of such a description: $p$ lies e.g.~in $\mathcal{Q}(H)$ as soon as there is \emphalt{some} way of explaining it using quantum theory. Whether the given measurement statistic is consistent with the \emphalt{particular} quantum-theoretic description predicted by a certain concrete theoretical model is an entirely different matter on which our formalism has no bearing.

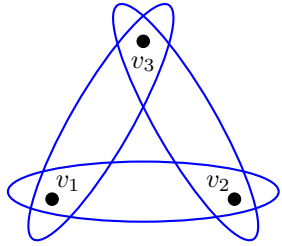
\begin{figure}
\begin{center}
\begin{tikzpicture}
\node[draw,shape=circle,fill,scale=.5] (a) at (90:1.4) {} ;
\node[draw,shape=circle,fill,scale=.5] (c) at (210:1.4) {} ;
\node[draw,shape=circle,fill,scale=.5] (e) at (330:1.4) {} ;
\node[below of=a,node distance=3mm] {$v_3$};
\node[above right of=c,node distance=3mm] {$v_1$};
\node[above left of=e,node distance=3mm] {$v_2$};
\draw[thick,blue,rotate=270] (0:.6) ellipse (.4cm and 1.8cm) ;
\draw[thick,blue,rotate=150] (0:.65) ellipse (.4cm and 1.8cm) ;
\draw[thick,blue,rotate=30] (0:.65) ellipse (.4cm and 1.8cm) ;
\end{tikzpicture}
\end{center}
\caption{The triangle scenario $\Delta$.}
\label{triscen}
\end{figure}

We now turn to some basic examples other than Figure~\ref{CKSfig}. Those mainly interested in nonlocality and Bell scenarios will become satisfied in Section~\ref{products}.

\begin{ex}
Figure~\ref{triscen} displays the triangle scenario $\Delta$. Its only probabilistic model is $p(v_1)=p(v_2)=p(v_3)=\tfrac{1}{2}$, since this is the only solution to the system of normalization equations
$$
p(v_1) + p(v_2) = 1,\qquad p(v_2) + p(v_3) = 1,\qquad p(v_1) + p(v_3) = 1.
$$
See~\cite{LSW} for more on this scenario and its unique probabilistic model.
\end{ex}

Contextuality scenarios having a unique probabilistic model, like $\Delta$ does, will be of particular importance in Theorem~\ref{extchar}.

\begin{ex}
Figure~\ref{empty} displays a contextuality scenario $H_0$ with $\mathcal{G}(H_0)=\emptyset$. Indeed, each of the outer triangles corresponds to a copy of the scenario $\Delta$ of Figure~\ref{triscen} and admits a unique probabilistic model where each vertex is assigned a probability $1/2$. This is incompatible with the central three-outcome measurement depicted in orange, since the existence of this measurement imposes that the probabilities associated with the three corresponding vertices should sum to $1$. 
\end{ex}

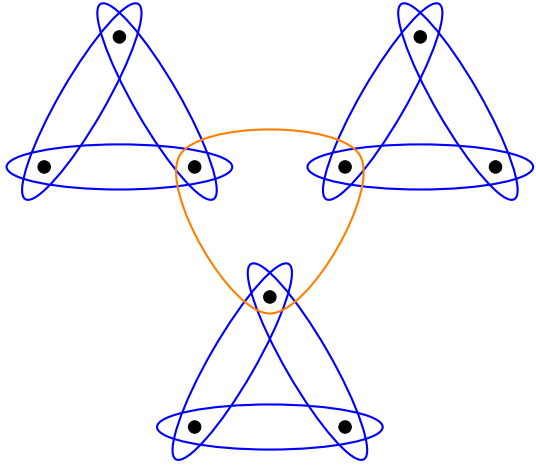
\begin{figure}
\begin{center}
\begin{tikzpicture}
\node[draw,shape=circle,fill,scale=.5] at (0,1.73) {} ;
\node[draw,shape=circle,fill,scale=.5] at (1,0) {} ;
\node[draw,shape=circle,fill,scale=.5] at (-1,0) {} ;
\draw[thick,blue,rotate around={90:(0,0)}] (0,0) ellipse (.3cm and 1.5cm) ;
\draw[thick,blue,rotate around={210:(.5,.87)}] (.5,.87) ellipse (.3cm and 1.5cm) ;
\draw[thick,blue,rotate around={330:(-.5,.87)}] (-.5,.87) ellipse (.3cm and 1.5cm) ;
\begin{scope}[shift={(4,0)}]
\node[draw,shape=circle,fill,scale=.5] at (0,1.73) {} ;
\node[draw,shape=circle,fill,scale=.5] at (1,0) {} ;
\node[draw,shape=circle,fill,scale=.5] at (-1,0) {} ;
\draw[thick,blue,rotate around={90:(0,0)}] (0,0) ellipse (.3cm and 1.5cm) ;
\draw[thick,blue,rotate around={210:(.5,.87)}] (.5,.87) ellipse (.3cm and 1.5cm) ;
\draw[thick,blue,rotate around={330:(-.5,.87)}] (-.5,.87) ellipse (.3cm and 1.5cm) ;
\end{scope}
\begin{scope}[shift={(2,-3.46)}]
\node[draw,shape=circle,fill,scale=.5] at (0,1.73) {} ;
\node[draw,shape=circle,fill,scale=.5] at (1,0) {} ;
\node[draw,shape=circle,fill,scale=.5] at (-1,0) {} ;
\draw[thick,blue,rotate around={90:(0,0)}] (0,0) ellipse (.3cm and 1.5cm) ;
\draw[thick,blue,rotate around={210:(.5,.87)}] (.5,.87) ellipse (.3cm and 1.5cm) ;
\draw[thick,blue,rotate around={330:(-.5,.87)}] (-.5,.87) ellipse (.3cm and 1.5cm) ;
\end{scope}
\draw[thick,orange] plot [smooth cycle,tension=.8] coordinates { (0.8,0.15) (3.2,0.15) (2,-1.95) } ;
\end{tikzpicture}
\end{center}
\caption{Example of a scenario $H_0$ without any probabilistic model: $\mathcal{G}(H_0)=\emptyset$.}
\label{empty}
\end{figure}

\begin{ex}
Figure~\ref{singleparty} displays the contextuality scenario defined by $k$ measurements with $m$ outcomes each, such that no two measurements share any outcome. Such scenarios are particularly relevant for describing `box' experiments where an observer can press one of $k$ buttons and record the corresponding measurement outcome. Since there is only one party, calling this a `Bell scenario' is a bit of a stretch, but it indeed arises as a degenerate example of a Bell scenario in which the number of parties is one.
\end{ex}

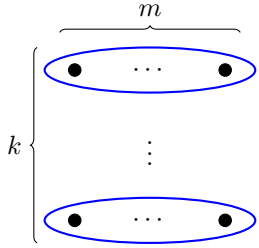
\begin{figure}
\begin{center}
\begin{tikzpicture}
\node[draw,shape=circle,fill,scale=.5] at (0,0) {} ;
\node[draw,shape=circle,fill,scale=.5] at (2,0) {} ;
\node[draw,shape=circle,fill,scale=.5] at (0,2) {} ;
\node[draw,shape=circle,fill,scale=.5] at (2,2) {} ;
\node at (1,0) {$\cdots$} ;
\node at (1,2) {$\cdots$} ;
\draw[thick,blue] (1,0) ellipse (1.4cm and .3cm) ;
\draw[thick,blue] (1,2) ellipse (1.4cm and .3cm) ;
\draw[decoration={brace},decorate] (-.2,2.5) -- (2.2,2.5) ;
\node at (1,2.8) {$m$} ;
\draw[decoration={brace},decorate] (-.5,-.3) -- (-.5,2.3) ;
\node at (-.8,1) {$k$} ;
\node at (1,1) {$\vdots$} ;
\end{tikzpicture}
\end{center}
\caption{The contextuality scenario $B_{1,k,m}$, a `Bell scenario' with only one party.}
\label{singleparty}
\end{figure}

Further examples will be discussed in Section~\ref{examples}.

For fixed $H$, the set $\mathcal{G}(H)\subseteq\R^{V(H)}$ is defined in terms of finitely many linear inequalities with rational coefficients. Therefore, it is a convex polytope with rational vertices. A natural question now is, which polytopes with rational vertices can arise in this way? This has been answered by Shultz:

\begin{thm}[Shultz~\cite{Shultz}]
\label{shultzthm}
Let $P\subseteq\R^d$ be a polytope with vertices in $\Q^d$. Then there exists a contextuality scenario $H_P$ such that $\mathcal{G}(H_P)$ is affinely isomorphic to $P$.
\end{thm}

Surprisingly, the combinatorial structure of some polytopes (i.e.~their face lattices) is such that they cannot be represented with rational coordinates only~\cite[Ex.~6.21]{Ziegler}. Hence, the requirement `with rational vertices' is a significant restriction on the combinatorial and geometric structure of those polytopes which arise as $\mathcal{G}(H)$ for some $H$.

\subsection{Characterizing extremal probabilistic models}

Since $\mathcal{G}(H)$ is a convex polytope, a natural question is: what are its extreme points? 

For instance, as we will discuss in Section~\ref{products}, for a Bell scenario $(n,k,m)$, the polytope $\mathcal{G}(B_{n,k,m})$ is the standard no-signaling polytope and hence its extreme points are the extremal no-signaling boxes. In the particular case of the CHSH scenario $B_{2,2,2}$, these extreme points are the $16$ deterministic boxes together with the $8$ variants of the PR-box~\cite{Barrett05}.

In this subsection, we would like to give an abstract characterization of these extremal models which applies to every contextuality scenario whatsoever, including all Bell scenarios.

\begin{defn}
Let $H$ be a contextuality scenario and $W\subseteq V(H)$. The \emph{subscenario induced by W} is the hypergraph $H_W$ with
$$
V(H_W) \defin W ,\qquad E(H_W) \defin \left\{\, e\cap W \::\: e\in E(H) \, \right\}.
$$
\end{defn} 

In words: $H_W$ is constructed by dropping all vertices which do not belong to $W$ and restricting all edges accordingly. If there are two edges $e_1,e_2\in E(H)$ with $e_1\cap W = e_2\cap W$, then these define the same edge in $E(H_W)$. For instance, the triangle scenario $\Delta$ of Figure~\ref{triscen} is the subscenario of the 3-circular hypergraph $\Delta_3$ of Figure~\ref{3meas6out} induced by $\{v_1, v_2, v_3\}$.

Intuitively, $H_W$ is the same scenario as $H$, except that all outcomes not in $W$ have been forbidden, i.e.~declared to have probability zero. In particular, every probabilistic model $p_W$ on $H_W$ extends to $H$ by setting
$$
p(v) \defin \begin{cases} p_W(v) & \textrm{if } v\in W, \\ 0 & \textrm{if } v\not\in W. \end{cases} 
$$
We say that $p$ is the \emph{extension} of $p_W$ from $H_W$ to $H$.

We have implicitly used induced subscenarios in~\cite{LOfp,LO2} when considering the (non-)orthogonality graphs of `possible events'.

\begin{lem}
\label{inducedtrans}
If $H_W$ is an induced subscenario of $H$ and $H_{W,W'}$ is an induced subscenario of $H_W$, then $H_{W,W'}$ is also an induced subscenario of $H$.
\end{lem}

\begin{proof}
Clear.
\end{proof}

Our main result in this section is this:

\begin{thm}
\label{extchar}
$p\in\mathcal{G}(H)$ is extremal if and only if it is the extension of $p_W\in\mathcal{G}(H_W)$ from some induced subscenario $H_W$ which has $p_W$ as its unique probabilistic model.
\end{thm}

\begin{proof}
If $H$ has a unique probabilistic model, i.e.~if $\mathcal{G}(H)=\{p\}$, then there is nothing to prove, since $H$ is an induced subscenario of itself.

Otherwise, the extreme points of $\mathcal{G}(H)$ are precisely the extreme points of the facets of $\mathcal{G}(H)$. Since $\mathcal{G}(H)$ is defined by
$$
p(v) \geq 0 \;\;\forall v\in V(H) ,\qquad \sum_{v\in e} p(v) = 1\;\;\forall e\in E(H),
$$
for every facet of $\mathcal{G}(H)$ there exists some $v\in V$ such that the facet contains exactly those $p\in\mathcal{G}(H)$ with $p(v)=0$. We fix such a $v$ and set $W \defin V(H)\setminus\{v\}$, obtaining an induced subscenario $H_W$ containing all vertices but $v$. By construction, the extensions of all probabilistic models in $\mathcal{G}(H_W)$ constitute the facet of $\mathcal{G}(H)$ defined by the equation $p(v)=0$.

The assertion then follows by repeatedly applying the same construction to the induced subscenarios obtained in this way; Lemma~\ref{inducedtrans} guarantees that one has an induced subscenario of the original $H$ at each step. This recursion necessarily ends with a scenario which admits a unique probabilistic model, since the dimension of $\mathcal{G}(H)$ decreases by $1$ in each step.
\end{proof}

As the proof shows, a similar statement also holds for all faces of $\mathcal{G}(H)$: they all are of the form $\mathcal{G}(H_W)$ for some induced subscenario $H_W$.

The proof also shows that an extreme point $p\in\mathcal{G}(H)$ is uniquely determined by the set of vertices $W_p \defin \{ v\in V(H) \:|\: p(v)\neq 0\}$, which induces a subscenario $H_W$ with a unique probabilistic model corresponding to forgetting those vertices on which $p$ has zero probability. So it is an important problem to understand which contextuality scenarios, besides Figure~\ref{triscen}, have a unique probabilistic model.

\begin{cor}
Let $H$ be a contextuality scenario with $n=|V(H)|$ many vertices. Then $\mathcal{G}(H)$ has at most $n\choose{\lfloor n/2\rfloor}$ many extremal points.
\end{cor}

\begin{proof}
Every extreme point $p\in\mathcal{G}(H)$ is uniquely determined by its associated set $W_p$ from the previous paragraph. These sets are mutually non-contained subsets of $V(H)$, of which there can be at most $n\choose{\lfloor n/2\rfloor}$ many by Sperner's theorem~\cite{Anderson}.
\end{proof}

This is a very crude upper bound and we do not know whether there are arbitrarily large scenarios for which it is tight. Oddly enough, this bound coincides precisely with the maximal number of edges in a contextuality scenario satisfying the previously discussed non-degeneracy requirement $e_1\subseteq e_2 \:\Rightarrow\: e_1=e_2$, since these edges are also mutually non-contained subsets of $V(H)$.

The deterministic models of Definition~\ref{cmdef} are a special class of extremal probabilistic model as follows. Clearly, every deterministic model is an extreme point of $\mathcal{G}(H)$. In terms of Theorem~\ref{extchar}, $p$ is deterministic if and only if each measurement in the associated $H_W$ has only one outcome, i.e.~if every vertex in $H_W$ is its own singleton edge. Those extreme points which are not deterministic are the \emph{maximally contextual} models in the scenario $H$.

\newpage
\section{\textbf{Products of contextuality scenarios and the no-signaling property}}
\label{products}

\subsection{Products of two scenarios}
\label{FRproducts}

Imagine two spatially separated or spacelike separated agents Alice and Bob. Alice is assumed to operate on a contextuality scenario $H_A$, while Bob is taken to operate on a contextuality scenario $H_B$. As usual in the literature on nonlocality, we always refer to these two agents as \emph{parties}. Now the two parties can apply simultaneous measurements on their respective systems and will then obtain simultaneous outcomes. In general, the two systems may be correlated, which can lead to correlations between the outcome of Alice with the outcome of Bob, so that probabilities will have to be assigned to pairs consisting of an outcome for Alice and an outcome for Bob.

The question now is, can this kind of `product' situation itself be described by a single contextuality scenario? How do two contextuality scenarios combine into a joint one? This question was first answered by Foulis and Randall~\cite{FR80}, who noticed that the answer is nontrivial. When combining two scenarios into one, we speak of a \emph{product}.

Clearly, the set of outcomes of a product scenario should be the cartesian product of the sets of outcomes, so that a joint outcome simply is the same thing as an outcome of $H_A$ together with an outcome of $H_B$. Also, every edge on $H_A$ should combine with any edge on $H_B$ into a joint edge. So, naively, one would define the product scenario like this:

\begin{defn}
\label{dirprod}
Let $H_A$ and $H_B$ be contextuality scenarios. The \emph{direct product} is the scenario $H_A\times H_B$ with
$$
V(H_A\times H_B) = V(H_A)\times V(H_B),\qquad E(H_A\times H_B) = E(H_A)\times E(H_B).
$$
\end{defn}

Now, in any actually observed probabilistic model, Bob's outcome probabilities should not depend on Alice's choice of measurement and vice versa, which leads to the requirement that a probabilistic model should have the `no-signaling' property:

\begin{defn}
A probabilistic model $p\in\mathcal{G}(H_A\times H_B)$ is \emph{no-signaling} if
\begin{enumerate}
\item 
$$
\sum_{w\in e} p(v,w)=\sum_{w\in e'} p(v,w) \qquad \forall v\in V(H_A),\: e,e'\in E(H_B).
$$
\item 
$$
\sum_{v\in e} p(v,w)=\sum_{v\in e'} p(v,w) \qquad \forall w\in V(H_B),\: e,e'\in E(H_A).
$$
\end{enumerate}
\end{defn}

This coincides with~\cite[Defn.~8]{BL} and~\cite[Defn.~3.2]{BFRW}, although the terminology is different.

Now the obvious question is, is every $p\in\mathcal{G}(H_A\times H_B)$ no-signaling? Unfortunately, this is not the case; Figure~\ref{signalingex} provides the arguably simplest example. It displays a probabilistic model where all probabilities are $0$ or $1$, and Alice (vertical) learns with certainty which measurement was performed by Bob (horizontal). It is easy to come up with other examples for virtually any non-trivial scenarios $H_A$ and $H_B$. 

\begin{figure}
\begin{center}
\subfigure[The contextuality scenario $\textara{\ea}$.]{
\begin{tikzpicture}[scale=1.5]
\clip (-.5,-.5) rectangle (2.5,.5) ;
\foreach \x in {0,...,2} \node[draw,shape=circle,fill,scale=.5] (e) at (\x,0) {} ;
\draw[thick,blue] (.5,0) ellipse (.8cm and .3cm) ;
\draw[thick,blue] (1.5,0) ellipse (.8cm and .3cm) ;
\end{tikzpicture}}\hspace{.5cm}
\subfigure[The direct product $\textara{\ea}\times\textara{\ea}$ equipped with a (deterministic) probabilistic model.]{
\begin{tikzpicture}
\clip (-1,-.5) rectangle (3,2.5) ;
\foreach \x in {0,...,2} \foreach \y in {0,...,2} \node[draw,shape=circle,fill,scale=.5] (e) at (\x,\y) {} ;
\foreach \x in {0,1} \foreach \y in {0,1} \draw[rounded corners=.65cm,thick,blue] (\x-.3,\y-.4) rectangle (\x+1.5,\y+1.4) ;
\node at (0.25,0) {$0$} ;
\node at (1.25,0) {$0$} ;
\node at (2.25,0) {$1$} ;
\node at (0.25,1) {$1$} ;
\node at (1.25,1) {$0$} ;
\node at (2.25,1) {$0$} ;
\node at (0.25,2) {$0$} ;
\node at (1.25,2) {$0$} ;
\node at (2.25,2) {$1$} ;
\end{tikzpicture}}
\end{center}
\caption{A contextuality scenario $\textara{\ea}$ and a probabilistic model on $\textara{\ea}\times\textara{\ea}$ not satisfying the no-signaling condition.}
\label{signalingex}
\end{figure}
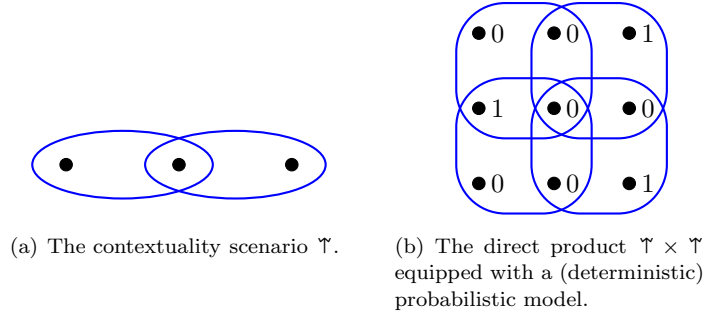

While one solution for this problem is to simply restrict to no-signaling models by fiat~\cite{BL}, a conceptually much more elegant solution is to use a different product of contextuality scenarios, which will have the property that the probabilistic models on this new product scenario will be precisely the no-signaling models on $H_A\times H_B$.

\begin{defn}[\cite{FR80}] \label{FR-prod}
The \emph{Foulis--Randall product} (\emph{FR-product}) is the scenario $H_A\otimes H_B$ with
$$
V(H_A\otimes H_B) = V(H_A)\times V(H_B),\qquad E(H_A\otimes H_B) = E_{A\rightarrow B} \cup E_{A\leftarrow B}
$$
where
\begin{align}
\label{FReq}
\begin{aligned}
E_{A\rightarrow B} & \defin \left\{ \bigcup_{a\in e_A} \{a\} \times f(a) \::\: e_A\in E_A,\: f:e_A\to E_B\right\} ,\\[4pt]
E_{A\leftarrow B} & \defin \left\{ \bigcup_{b\in e_B} f(b) \times \{b\} \::\: e_B\in E_B,\: f:e_B\to E_A\right\} .
\end{aligned}
\end{align}
\end{defn}

Intuitively, an element of $E_{A\rightarrow B}$ is the following: first, an edge $e_A\in E(H_A)$ representing a measurement conducted by Alice; second, a function $f:e_A\to E(H_B)$ which determines the subsequent measurement of Bob as a function of Alice's outcome. This function $f$ maps each vertex $a \in e_A$ to an edge $f(a) \in E_B$. This defines a joint measurement in which we think of Alice measuring first and communicating her outcome to Bob, who then chooses his measurement as a function of Alice's outcome. This is a feasible way to operate on the joint system and therefore should be considered as a measurement conductible on the joint system.\footnote{Whether these measurements should be considered physically realizable or not depends on the concrete physics of the scenario: in the case of space-like separation, they are mathematical idealizations without physical realizability. If Alice and Bob are not space-like separated, then these additional edges describe physically realizable measurements. Our mathematical formalism correctly describes both kinds of situations, although the physical interpretation of the edges in the product scenario is somewhat different.} Its outcomes are pairs $(a,b)$ with $a\in e_A$ and $b\in f(a)$, so that the set of all these outcomes is $\bigcup_{a\in e_A} \{a\} \times f(a)$. Similar remarks apply to the elements of $E_{B\rightarrow A}$, which describe joint measurements in which Bob measures first and then communicates his outcome to Alice, who chooses her measurement as a function of Bob's outcome.

In this way, an edge in $H_A\otimes H_B$ is an element of $E_{A\rightarrow B}$, $E_{A\leftarrow B}$, or of both sets. The latter joint measurements are precisely those of the form $e_A\times e_B$ from Definition~\ref{dirprod}, which can be interpreted as simultaneous measurements.

For example, Figure~\ref{CHSHc} displays the FR-product of~\ref{CHSHa} with~\ref{CHSHb}, which is another copy of~\ref{CHSHa}. $E_{A\rightarrow B}$ contains the edges of Figure~\ref{sim_meas} and~\ref{CHSHab}, while $E_{B\rightarrow A}$ consists of~\ref{sim_meas} and~\ref{CHSHba}.

\afterpage{
\newgeometry{top=1.9cm,bottom=-1.9cm,left=3cm,right=3cm} 
\begin{figure}
\definecolor{darkgreen}{rgb}{0,.5,0}
\begin{centering}
\subfigure[Alice's two binary measurements $B_{1,2,2}$ (see Figure~\ref{singleparty}).]{
\label{CHSHa}
\begin{tikzpicture}
\clip (0,.5) rectangle (5,5);
\foreach \x in {0,1} \foreach \a in {0,1}
{
	\node[draw,shape=circle,fill,scale=.5] (e) at (1,4-\a-2*\x) {} ;
	\node[above=0pt] at (e) {\tiny{$\a|\x$}} ;
}
\draw[thick,blue] (1,1.6) ellipse (.5cm and .9cm) ;
\draw[thick,blue] (1,3.6) ellipse (.5cm and .9cm) ;
\end{tikzpicture}}\hspace{.5cm}
\subfigure[Bob's two binary measurements $B_{1,2,2}$ (see Figure~\ref{singleparty}).]{
\label{CHSHb}
\begin{tikzpicture}
\clip (0,.5) rectangle (5,5);
\foreach \y in {0,1} \foreach \b in {0,1}
{
	\node[draw,shape=circle,fill,scale=.5] (e) at (\b+2*\y+1,4) {} ;
	\node[above=0pt] at (e) {\tiny{$\b|\y$}} ;
}
\draw[thick,blue] (1.5,4.1) ellipse (.9cm and .5cm) ;
\draw[thick,blue] (3.5,4.1) ellipse (.9cm and .5cm) ;
\end{tikzpicture}}\hspace{.5cm}
\subfigure[Simultaneous measurements.]{
\label{sim_meas}
\hspace{.5cm}
\begin{tikzpicture}
\clip (0,.5) rectangle (5,5);
\foreach \x in {0,1} \foreach \y in {0,1} \foreach \a in {0,1} \foreach \b in {0,1}
{
	\node[draw,shape=circle,fill,scale=.5] (e) at (\b+2*\y+1,4-\a-2*\x) {} ;
	\node[above=0pt] at (e) {\tiny{$\a\b|\x\y$}} ;
}
\foreach \x in {0,2} \foreach \y in {0,2} \draw[rounded corners,thick,blue] (\x+0.55,\y+0.68) rectangle (\x+2.45,\y+2.52) ;
\end{tikzpicture}\hspace{.5cm}}\hspace{.5cm}
\subfigure[Bob's measurement choice depends on Alice's outcome.]{\hspace{.5cm}
\label{CHSHab}
\begin{tikzpicture}
\clip (0,.5) rectangle (5,5);
\foreach \x in {0,1} \foreach \y in {0,1} \foreach \a in {0,1} \foreach \b in {0,1}
{
	\node[draw,shape=circle,fill,scale=.5] (e) at (\b+2*\y+1,4-\a-2*\x) {} ;
	\node[above=0pt] at (e) {\tiny{$\a\b|\x\y$}} ;
}
\foreach \y in {0,2} \draw[thick,red,rounded corners] (2.8,\y+0.8) -- (4.33,\y+0.8) -- (4.33,\y+1.4) -- (2.8,\y+1.4) -- (2.25,\y+2.4) -- (0.67,\y+2.4) -- (0.67,\y+1.8) -- (2.2,\y+1.8) -- cycle ;
\foreach \y in {0,2} \draw[thick,red,rounded corners] (2.8,\y+1.8) -- (4.33,\y+1.8) -- (4.33,\y+2.4) -- (2.75,\y+2.4) -- (2.2,\y+1.4) -- (0.67,\y+1.4) -- (0.67,\y+0.8) -- (2.2,\y+0.8) -- cycle ;
\end{tikzpicture}\hspace{.5cm}}\hspace{.5cm}
\subfigure[Alice's measurement choice depends on Bob's outcome.]{\hspace{.5cm}
\label{CHSHba}
\begin{tikzpicture}
\clip (0,.5) rectangle (5,5);
\foreach \x in {0,1} \foreach \y in {0,1} \foreach \a in {0,1} \foreach \b in {0,1}
{
	\node[draw,shape=circle,fill,scale=.5] (e) at (\b+2*\y+1,4-\a-2*\x) {} ;
	\node[above=0pt] at (e) {\tiny{$\a\b|\x\y$}} ;
}
\foreach \x in {0,2} \draw[thick,darkgreen,rounded corners] (\x+0.62,2.4) -- (\x+0.62,0.75) -- (\x+1.3,0.75) -- (\x+1.3,2.1) -- (\x+2.38,2.8) -- (\x+2.38,4.45) -- (\x+1.65,4.45) -- (\x+1.65,3) -- cycle ;
\foreach \x in {0,2} \draw[thick,darkgreen,rounded corners] (5-\x-0.62,2.4) -- (5-\x-0.62,0.75) -- (5-\x-1.3,0.75) -- (5-\x-1.3,2.1) -- (5-\x-2.38,2.8) -- (5-\x-2.38,4.45) -- (5-\x-1.65,4.45) -- (5-\x-1.65,3) -- cycle ;
\end{tikzpicture}\hspace{.5cm}}
\subfigure[Foulis--Randall product: the CHSH scenario $B_{2,2,2}=B_{1,2,2}\otimes B_{1,2,2}$.]{
\label{CHSHc}
\begin{tikzpicture}
\clip (0,.5) rectangle (5,5);
\foreach \x in {0,1} \foreach \y in {0,1} \foreach \a in {0,1} \foreach \b in {0,1}
{
	\node[draw,shape=circle,fill,scale=.5] (e) at (\b+2*\y+1,4-\a-2*\x) {} ;
	\node[above=0pt] at (e) {\tiny{$\a\b|\x\y$}} ;
}
\foreach \x in {0,2} \foreach \y in {0,2} \draw[rounded corners,thick,blue] (\x+0.55,\y+0.68) rectangle (\x+2.45,\y+2.52) ;
\foreach \y in {0,2} \draw[thick,red,rounded corners] (2.8,\y+0.8) -- (4.33,\y+0.8) -- (4.33,\y+1.4) -- (2.8,\y+1.4) -- (2.25,\y+2.4) -- (0.67,\y+2.4) -- (0.67,\y+1.8) -- (2.2,\y+1.8) -- cycle ;
\foreach \y in {0,2} \draw[thick,red,rounded corners] (2.8,\y+1.8) -- (4.33,\y+1.8) -- (4.33,\y+2.4) -- (2.75,\y+2.4) -- (2.2,\y+1.4) -- (0.67,\y+1.4) -- (0.67,\y+0.8) -- (2.2,\y+0.8) -- cycle ;
\foreach \x in {0,2} \draw[thick,darkgreen,rounded corners] (\x+0.62,2.4) -- (\x+0.62,0.75) -- (\x+1.3,0.75) -- (\x+1.3,2.1) -- (\x+2.38,2.8) -- (\x+2.38,4.45) -- (\x+1.65,4.45) -- (\x+1.65,3) -- cycle ;
\foreach \x in {0,2} \draw[thick,darkgreen,rounded corners] (5-\x-0.62,2.4) -- (5-\x-0.62,0.75) -- (5-\x-1.3,0.75) -- (5-\x-1.3,2.1) -- (5-\x-2.38,2.8) -- (5-\x-2.38,4.45) -- (5-\x-1.65,4.45) -- (5-\x-1.65,3) -- cycle ;
\end{tikzpicture}}
\subfigure[Alternative drawing of~\subref{CHSHc} after rearranging the vertices.]{
\label{CHSHd}
\begin{tikzpicture}
\clip (0,.5) rectangle (5,5);
\foreach \x in {0,1} \foreach \a in {0,1} \foreach \b in {0,1}
{
	\node[draw,shape=circle,fill,scale=.5] (e) at (1-\b+2*\x+1,3+\a-2*\x) {} ;
	\node[above=0pt] at (e) {\tiny{$\a\b|\x\x$}} ;
}
\foreach \x in {0,1} \foreach \a in {0,1} \foreach \b in {0,1}
{
	\node[draw,shape=circle,fill,scale=.5] (e) at (\b-2*\x+3,4-\a-2*\x) {} ;
	\pgfmathtruncatemacro{\y}{1 - \x}
	\node[above=0pt] at (e) {\tiny{$\a\b|\x\y$}} ;
}
\foreach \x in {0,2} \foreach \y in {0,2} \draw[rounded corners,thick,blue] (\x+0.55,\y+0.68) rectangle (\x+2.45,\y+2.52) ;
\foreach \y in {0,2} \draw[thick,red,rounded corners] (0.67,\y+0.8) rectangle (4.33,\y+1.4) ;
\foreach \y in {0,2} \draw[thick,red,rounded corners] (0.67,\y+1.8) rectangle (4.33,\y+2.4) ;
\foreach \x in {0,2} \draw[thick,darkgreen,rounded corners] (\x+0.62,0.75) rectangle (\x+1.38,4.45) ;
\foreach \x in {0,2} \draw[thick,darkgreen,rounded corners] (\x+1.62,0.75) rectangle (\x+2.38,4.45) ;
\end{tikzpicture}}
\end{centering}
\caption{Construction of the CHSH scenario $B_{2,2,2}$ as a Foulis--Randall product $B_{2,2,2} = B_{1,2,2}\otimes B_{1,2,2}$.}
\label{CHSH_sce}
\end{figure}
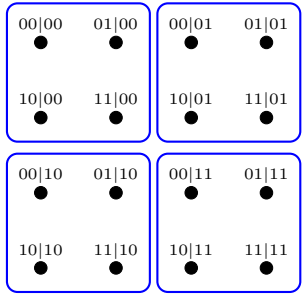
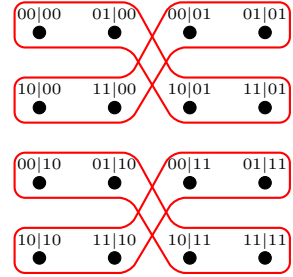
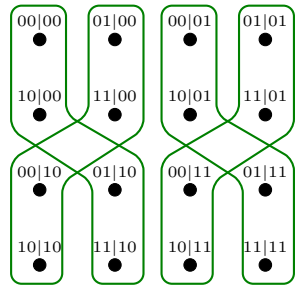
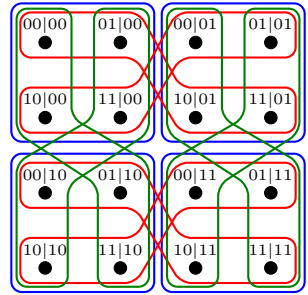
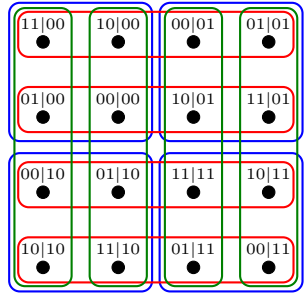
\clearpage
}
\restoregeometry

Since $H_A\otimes H_B$ contains the same vertices as $H_A\times H_B$ but more edges, we have an inclusion $\mathcal{G}(H_A\otimes H_B)\subseteq \mathcal{G}(H_A\times H_B)$. The following observation is due to Barnum, Fuchs, Renes and Wilce:

\begin{prop}[{\cite[Cor.~3.5]{BFRW}}]
\label{nosigprop}
$\mathcal{G}(H_A\otimes H_B)\subseteq \mathcal{G}(H_A\times H_B)$ is exactly the set of no-signaling models.
\end{prop}

We will give an explicit proof of a more general statement in Corollary~\ref{multinosigcor}.

It is in this sense that $H_A\otimes H_B$, in contrast to $H_A\times H_B$, automatically incorporates the no-signaling requirement. This is the reason why we regard it as the `right' product of contextuality scenarios.

Both the inclusion $\mathcal{G}(H_A\otimes H_B)\subseteq \mathcal{G}(H_A\times H_B)$ and Proposition~\ref{nosigprop} can intuitively be understood in terms of the duality between states and effects~\cite{DAriano}: restricting the probabilistic models to those satisfying no-signaling makes more measurements well-defined and in particular allows measurements in which the parties use signaling; on the other hand, allowing measurements in which the parties use signaling is possible only if the joint system itself, on which the measurements are conducted, does not have internal signaling. Compare Wilce~\cite{Wilce2}, who prefers the term \emphalt{influence-free} over \emphalt{no-signaling}.

One can also do all this for the case of \emphalt{unidirectional} no-signaling: defining a product of $H_A$ and $H_B$ by only using the $E_{A\to B}$ of~(\ref{FReq}) gives probabilistic models which are no-signaling from Bob to Alice. See~\cite{BFRW} for more details. The resulting product contextuality scenario may be interpreted as describing a temporal succession of operating on $H_B$ after having operated on $H_A$.

Given two contextuality scenarios $H_A$ and $H_B$ together with probabilistic models
$$
p_A\in \mathcal{G}(H_A), \qquad p_B\in \mathcal{G}(H_B) ,
$$
there should exist a probabilistic model $p_A\otimes p_B$ on $H_A\otimes H_B$ having the interpretation of placing physical systems behaving as $p_A$ and $p_B$ `side by side' so that measurements can be conducted on both in parallel, revealing no correlations between the two systems, but independent statistics. To this end, one should obviously define $p_A\otimes p_B$ as the mapping
$$
p_A\otimes p_B \::\: V(H_A)\times V(H_B) \longrightarrow [0,1],\qquad (v_A,v_B)\mapsto p_A(v_A) p_B(v_B) .
$$

\begin{prop}\label{PM_prod}
This $p_A\otimes p_B$ is a probabilistic model on $H_A\otimes H_B$.
\end{prop}

\begin{proof}
We need to prove that $\sum_{v \in e} p_A\otimes p_B (v) = 1$ for each edge $e \in E(H_A\otimes H_B)$. Without loss of generality, we can assume $e \in E_{A\rightarrow B}$, so that $e = \bigcup_{a\in e_A} \{a\} \times f(a)$ for some $e_A \in E_A$ and some $f: e_A \to E_B$, which maps each vertex in $e_A$ to an edge in $H_B$. Therefore,
\begin{equation*}
\begin{aligned}
\sum_{v \in e} p_A\otimes p_B (v) = \sum_{a \in e_A}  \sum_{b \in f(a)} p_A(a) p_B(b) = \sum_{a \in e_A} p_A(a) \sum_{b \in f(a)}  p_B(b) = \sum_{a\in e_A} p_A(a) \cdot 1 = 1,
\end{aligned}
\end{equation*}
since $p_B$ and $p_A$ are probabilistic models on $H_B$ and $H_A$, respectively. 
\end{proof}

We write $\mathcal{G}(H_A)\otimes \mathcal{G}(H_B)$ for the set of all probabilistic models of the form $p_A\otimes p_B$. We have just shown that $\mathcal{G}(H_A)\otimes \mathcal{G}(H_B)\subseteq\mathcal{G}(H_A\otimes H_B)$.

\begin{rem}
\label{Gproduct}
Often $\mathcal{G}(H_A\otimes H_B)$ is strictly larger than the convex hull of $\mathcal{G}(H_A)\otimes \mathcal{G}(H_B)$. For example for the Bell scenario $B_{2,2,2}=B_{1,2,2}\otimes B_{1,2,2}$ discussed below, the Popescu--Rohrlich box~\cite{PR}, which was originally discovered by Tsirelson~\cite{TsiHS}*{eq.~(1.11)}, is an element of $\mathcal{G}(B_{1,2,2}\otimes B_{1,2,2})$, but does not lie in the convex hull of $\mathcal{G}(B_{1,2,2})\otimes \mathcal{G}(B_{1,2,2})$.
\end{rem}

\subsection{Non-orthogonality graph of a product}

The Foulis--Randall product of contextuality scenario translates into the strong product $\boxtimes$ of the associated non-orthogonality graphs (see Appendix~\ref{appcap} for details on graph theoretical definitions).

\begin{lem}
\label{ortproduct}
Let $H_A$ and $H_B$ be contextuality scenarios. Then,
$$
\mathrm{NO}(H_A\otimes H_B) = \mathrm{NO}(H_A) \boxtimes \mathrm{NO}(H_B) .
$$
\end{lem}

\begin{proof}
Clearly both sides are graphs having $V(H_A)\times V(H_B)$ as their set of vertices, so what needs to be shown is that the adjacency relations coincide.

We first prove that if $(u_A,u_B)\perp(v_A,v_B)$ in $\mathrm{NO}(H_A \otimes H_B)$, then these two vertices are also not adjacent in $\mathrm{NO}(H_A)\boxtimes \mathrm{NO}(H_B)$. The assumption means that there is an edge $e\in E(H_A\otimes H_B)$ which contains both $(u_A,u_B)$ and $(v_A,v_B$); this edge has one of the two forms of~(\ref{FReq}). If it is in $E_{A\to B}$, then $u_A,v_A\in e_A$, meaning that $u_A\perp v_A$. Similarly, if the edge is in $E_{A\leftarrow B}$, then $u_B\perp v_B$. The conclusion follows from either case.

For proving the opposite implication, we show that $(u_A,u_B)\perp(v_A,v_B)$ in $\mathrm{NO}(H_A) \boxtimes \mathrm{NO}(H_B)$ implies the same in $\mathrm{NO}(H_A\otimes H_B)$. The assumption means that $u_A\perp v_A$ or $u_B\perp v_B$; by symmetry, it is enough to consider the case $u_A\perp v_A$. Then, there exists some $e_A\in E(H_A)$ with $u_A,v_A\in e_A$. Now choose $e_B, e'_B\in E_B$ such that $u_B \in e_B$ and $v_B \in e'_B$, and some function $f:e_A\to E_B$ with $f(u_A) = e_B$ and $f(v_A) = e'_B$. Then
$$
\bigcup_{a\in e_A} \{a\} \times f(a)
$$
is an edge in $H_A\otimes H_B$ containing $(u_A,u_B)$ and $(v_A,v_B)$, which proves the claim. 
\end{proof}

\subsection{Products of more than two scenarios}
\label{higherprodsmain}

It is not difficult to check that the Foulis--Randall product `$\otimes$' is a commutative binary operation on contextuality scenarios. But now what about having more than two parties which operate in their respective scenarios simultaneously? How does this binary operation behave when applying it to three or more scenarios?

Given three scenarios $H_A$, $H_B$, $H_C$, we can first form the product $H_A\otimes H_B$, and then the product of this with $H_C$, which gives $(H_A\otimes H_B)\otimes H_C$. Alternatively, we may first form the product $H_B\otimes H_C$, and then the product $H_A\otimes (H_B\otimes H_C)$. One might hope that these two ways of taking the product result in the same scenario, but this is generally not the case:

\begin{prop}
\label{nonass}
There are scenarios for which $(H_A\otimes H_B)\otimes H_C\neq H_A\otimes (H_B\otimes H_C)$.
\end{prop}

In other words, the Foulis--Randall product `$\otimes$' is not associative! 

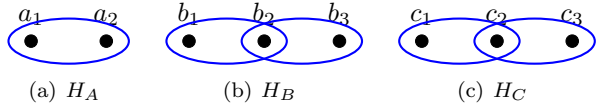
\begin{figure}
\subfigure[$H_A$]{
\begin{centering}
\begin{tikzpicture}
\node[draw,shape=circle,fill,scale=.5,label=$a_1$] at (0,0) {};
\node[draw,shape=circle,fill,scale=.5,label=$a_2$] at (1,0) {};
\draw[thick,blue] (.5,0) ellipse (.8cm and .3cm) ;
\end{tikzpicture}
\end{centering}}
\subfigure[$H_B$]{
\begin{centering}
\begin{tikzpicture}
\foreach \x in {1,...,3} \node[draw,shape=circle,fill,scale=.5,label=$b_\x$] at (\x-1,0) {} ;
\draw[thick,blue] (.5,0) ellipse (.8cm and .3cm) ;
\draw[thick,blue] (1.5,0) ellipse (.8cm and .3cm) ;
\end{tikzpicture}
\end{centering}}
\subfigure[$H_C$]{
\begin{centering}
\begin{tikzpicture}
\foreach \x in {1,...,3} \node[draw,shape=circle,fill,scale=.5,label=$c_\x$] at (\x-1,0) {} ;
\draw[thick,blue] (.5,0) ellipse (.8cm and .3cm) ;
\draw[thick,blue] (1.5,0) ellipse (.8cm and .3cm) ;
\end{tikzpicture}
\end{centering}}
\caption{Contextuality scenarios in the proof of Proposition~\ref{nonass}.}
\label{nonassfig}
\end{figure}

\begin{proof}
With the scenarios of Figure~\ref{nonassfig}, the set of pairs
\[
\left\{ (b_1,c_1), (b_2,c_1), (b_2,c_2), (b_3,c_2) \right\}
\]
is an edge in $H_B\otimes H_C$, representing the joint measurement in which the $H_C$-measurement $\{c_1,c_2\}$ in is followed by the $H_B$-measurement $\{b_1,b_2\}$ if the outcome was $c_1$, and by $\{b_2,b_3\}$ if the outcome was $c_2$. Similarly, the set of pairs
\[
\left\{ (b_1,c_2), (b_2,c_2), (b_2,c_3), (b_3,c_3) \right\}
\]
is also an edge in $H_B\otimes H_C$. Therefore, the set of triples
\begin{align}
\begin{split}
\label{bigmeas}
\big\{ & (a_1,b_1,c_1), (a_1,b_2,c_1), (a_1,b_2,c_2), (a_1,b_3,c_2) \\[4pt]
 & (a_2,b_1,c_2), (a_2,b_2,c_2), (a_2,b_2,c_3), (a_2,b_3,c_3) \big\}
\end{split}
\end{align}
is an edge in $H_A\otimes (H_B\otimes H_C)$.

We now show that this set of vertices is not an edge in $(H_A\otimes H_B)\otimes H_C$. If it is, then it has to arise as a sequence of two measurements, with the first fixed one conducted on $H_A\otimes H_B$ or on $H_C$. This first measurement cannot be on $H_C$, since otherwise~\eqref{bigmeas} could not contain all vertices $c_1,c_2,c_3$ of $H_C$, and therefore it has to be on $H_A\otimes H_B$, meaning that the set of pairs
\[
\{ (a_1,b_1), (a_1,b_2), (a_1,b_3), (a_2,b_1), (a_2,b_2), (a_2,b_3) \}
\]
has to be an edge in $H_A\otimes H_B$. Since this set is all of $V(H_A\otimes H_B)$, this can clearly not be the case; one can see this more formally by noting that any protocol realizing it would have to begin with the measurement $\{a_1,a_2\}$ on $H_A$, but then $\{b_1,b_2,b_3\}$ would have to be an edge in $H_B$, which is not the case.
\end{proof}

Intuitively speaking, the reason for this non-associativity is this: in $H_A\otimes (H_B\otimes H_C)$, it is possible for Alice on $H_A$ to measure first, while Charlie on $H_C$ subsequently conducts a measurement depending on Alice's outcome, and finally Bob on $H_B$ a measurement depending on Charlie's outcome. This kind of protocol is not implementable in the bracketing $(H_A\otimes H_B)\otimes H_C$, however, since Charlie cannot measure in between Alice and Bob.

We resolve this problem by first defining new `minimal' and `maximal' $n$-fold Foulis--Randall products and studying their relationship. We show in Appendix~\ref{multiproducts} that these two products, and all other intermediate ones like the ones obtained from applying binary products with an arbitrary bracketing, are observationally equivalent. A reader not interested in the technical details of how this works may want to skip to Section~\ref{cm}.

Let $H_1,\ldots, H_n$ be the contextuality scenarios of which we want to take the product.

\begin{defn}
The \emph{minimal Foulis--Randall product} $^{\min}\otimes_{i=1}^n H_i$ has vertices
\[
V\left(^{\min}\otimes_{i=1}^n H_i\right) \defin \prod_i V(H_i) = V(H_1)\times\ldots\times V(H_n),
\]
and $\bigcup_k E_k$ as its set of edges, where the elements of $E_k$ indexed by party $k=1,\ldots,n$ are of the form
\beq
\label{nFRedge}
\left\{\: (v_1,\ldots,v_n) \:|\: v_i\in e_i \:\forall i\neq k,\: v_k\in f(\vec{v}) \:\right\}
\eeq
for some edge $e_i\in E(H_i)$ for every party $i\neq k$ and a function $\vec{v}\mapsto f(\vec{v})$ which assigns to every joint outcome $\vec{v}=(v_1,\ldots,\bcancel{v_k},\ldots,v_n)$ of all parties except $k$ an edge $f(\vec{v})\in E(H_k)$.
\end{defn}

The measurement~\eqref{nFRedge} can alternatively be written in the form
\beq
\label{nFRedge2}
\bigcup_{\vec{v}'} \{\vec{v}'\} \times f(\vec{v}'),
\eeq
where $\vec{v}'$ ranges over $\prod_{i\neq k} e_i$ and it is understood that the second factor of the cartesian product has to be inserted into the $k$th position of $\vec{v}'$. This measurement can be interpreted as follows: each party $i\neq k$ starts by conducting their measurement $e_i$. These parties then announce their joint outcome $\vec{v}$ to the remaining party $k$, who conducts a measurement $f(\vec{v})$ chosen as a function of the previous joint outcome. This kind of protocol is a direct generalization of~\eqref{FReq}. If party $k$ conducts a fixed measurement $e_k$ independently of $\vec{v}$, then the resulting edge~\eqref{nFRedge} is just the cartesian product edge $\prod_{i=1}^n e_i$. Therefore, $E(H_1\times\ldots\times H_n)\subseteq E\left(^{\min}\otimes_{i=1}^n H_i\right)$.

\begin{prop}
\label{multinosigprop}
A probabilistic model $p\in \mathcal{G}(H_1\times\ldots\times H_n)$ lies in $\mathcal{G}\left(^{\min}\otimes_{i=1}^n H_i\right)$ if and only if it satisfies the \emph{no-signaling equations}
\beq
\label{multinosig}
\sum_{w\in e_k} p(\vec{v},w) = \sum_{w\in e'_k} p(\vec{v},w)
\eeq
for all parties $k$, joint outcomes $\vec{v}\in\prod_{i\neq k} V(H_i)$ and measurements $e_k,e'_k\in E(H_k)$.
\end{prop}

\begin{proof}
We show that the no-signaling equation~\eqref{multinosig} for a given party $k$, outcomes $v_i$ for $i\neq k$, and edges $e_k,e'_k\in E(H_k)$ follows from the normalization constraints on $p\in\mathcal{G}\left(^{\min}\otimes_{i=1}^n H_i\right)$. To this end, we choose an auxiliary edge $e_i\in E(H_i)$ with $v_i\in e_i$ for each party $i\neq k$. We then consider the normalization equation for the independent measurement $\prod_{i=1}^n e_i$, which is
\[
\sum_{\vec{v}'} \sum_{w\in e_k} p(\vec{v}',w) = 1,
\]
where the summation index is now $\vec{v}'=(v'_1,\ldots,\bcancel{v'_k},\ldots,v'_n)$, since $\vec{v}$ already stands for the fixed outcomes of~\eqref{multinosig}. We also write down the normalization equation for the measurement of the form~\eqref{nFRedge} given by
\[
f(\vec{v}') \defin \begin{cases} e'_k & \textrm{ if } \vec{v}'=\vec{v},\\ e_k & \textrm{ otherwise},\end{cases}
\]
for which the normalization equation reads
\[
\sum_{w\in e'_k} p(\vec{v},w) + \sum_{\vec{v}'\neq\vec{v}}\sum_{w\in e_k} p(\vec{v}',w) = 1.
\]
The claim now follows from comparing the two normalization equations.

Conversely, we need to show that the no-signaling equations~\eqref{multinosig} together with normalization for independent measurements imply normalization for an edge of the form~\eqref{nFRedge}. This follows from the computation
\[
\sum_{\vec{v}}\sum_{w\in f(\vec{v})} p(\vec{v},w) \stackrel{\eqref{multinosig}}{=} \sum_{\vec{v}}\sum_{w\in e_k} p(\vec{v},w) = 1,
\]
where $e_k\in E(H_k)$ is arbitrary, and the second equation is the normalization equation with respect to the product measurement $\prod_i e_i$, which already holds on $H_1\times\ldots\times H_n$.
\end{proof}

It may seem odd to single out exactly one party as the one conducting their measurement as a function of the others' outcomes. After all, why not have only one party conduct an initial measurement, a second party conduct a measurement as a function of the first party's outcome, and so on, until the last party conducts their measurement still as a function of all previous outcomes? In principle, it should even be allowed to choose the ordering of the parties in this protocol itself as a function of the previous outcomes. Such a strategy resembles the `dynamic wirings' discussed in~\cite{LO2}. We can make this precise as follows:

\begin{defn}
A \emph{measurement protocol} $\mathcal{P}$ for $S$ consists of the following data:
\begin{enumerate}
\item if $S=\emptyset$, the unique protocol is $\mathcal{P}=\emptyset$,
\item\label{measprotb} otherwise, the protocol is a triple $\mathcal{P}=(k,e,f)$, where $k\in S$ is a party, $e\in E(H_k)$ is an edge, and $f$ is a function assigning to each vertex $v\in e$ a measurement protocol $f(v)$ on $S\setminus\{k\}$.
\end{enumerate}
\end{defn}

This may look like a circular definition, since it uses the concept of measurement protocol for defining what a measurement protocol is. However, since the measurement protocol $\mathcal{P}'$ in~\ref{measprotb} is for a smaller subset of parties, this is a perfectly sensible recursive definition which always reduces to the base case $S=\emptyset$. Mathematically speaking, we are dealing with an inductive definition~\cite{hott}.

In this way, a measurement protocol corresponds to an initial measurement by party $k$, together with a function assigning to every outcome of this initial measurement another (shorter) measurement protocol for the remaining parties $S\setminus\{k\}$. We will use a closely related definition in Appendix~\ref{reltosheaf}.

\begin{defn}
The \emph{set of outcomes} $O(\mathcal{P})$ of a measurement protocol $\mathcal{P}$ for $S$ is
\begin{enumerate}
\item if $S=\emptyset$, a one-element set $O(\mathcal{P})=\{\ast\}$ containing a dummy outcome $\ast$,
\item otherwise, if $\mathcal{P}=(k,e,f)$, then
\[
O(\mathcal{P}) \defin \bigcup_{v\in e} \{v\}\times O(f(v)).
\]
\end{enumerate}
\end{defn}

In other words, the unique measurement protocol $\mathcal{P}=\emptyset$ for no parties has a unique---and hence deterministic---outcome denoted $\ast$. An outcome of a non-trivial protocol $\mathcal{P}=(k,e,f)$ consists of an outcome $v\in e$ of the initial measurement, together with an outcome of the remaining protocol $f(v)$.

In this way, an outcome of a measurement protocol $\mathcal{P}$ for all parties $S$ has exactly one component in each $V(H_i)$ for each $i\in S$, so that it can be regarded as an element of $\prod_{i\in S} V(H_i)$. The set of all outcomes of $\mathcal{P}$ is therefore a subset of the vertices of $^{\max}\otimes_{i\in S} H_i$, which we take to be the edge determined by $\mathcal{P}$. The collection of all these edges defines the scenario $^{\max}\otimes_{i\in S} H_i$.
\begin{defn}
The \emph{maximal Foulis--Randall product} $^{\max}\otimes_{i=1}^n H_i$ has vertices $
V\left(^{\max}\otimes_{i=1}^n H_i\right) \defin \prod_i V(H_i)$ and set of edges $\bigcup_{\mathcal{P}} O(\mathcal{P})$ where $\mathcal{P}$ is a measurement protocol for $\{1, \ldots, n\}$.
\end{defn}

\begin{lem}
\label{allintermed}
Any way of permuting factors and choosing brackets in the expression $H_1\otimes\ldots\otimes H_n$ yields a scenario intermediate between the minimal and the maximal one,
\[
E\left(^{\min}\otimes_{i=1}^n H_i\right) \subseteq E\left(H_1\otimes\ldots\otimes H_n\right) \subseteq E\left(^{\max}\otimes_{i=1}^n H_i\right).
\]
\end{lem}

\begin{proof}
This can be shown by induction on $n$. For $n=1$, there is nothing to prove, since all three scenarios trivially coincide.

Otherwise, for $n>1$, we can permute the scenarios $H_i$ such that the outermost bracketing has the form
\[
(H_1\otimes\ldots\otimes H_j)\otimes (H_{j+1}\otimes\ldots\otimes H_n),
\]
where now the products inside each bracketing are themselves arbitrarily permuted and bracketed. By the induction hypotheses, we know that
\begin{align*}
E\left(^{\min}\otimes_{i=1}^j H_i\right) &\subseteq E\left(H_1\otimes\ldots\otimes H_j\right) \subseteq E\left(^{\max}\otimes_{i=1}^j H_i\right),\\[4pt]
E\left(^{\min}\otimes_{i=j+1}^n H_i\right) &\subseteq E\left(H_{j+1}\otimes\ldots\otimes H_n\right) \subseteq E\left(^{\max}\otimes_{i=j+1}^n H_i\right).
\end{align*}
So for the first inclusion, it is sufficient to show that
\[
E\left(^{\min}\otimes_{i=1}^n H_i\right) \subseteq E\left(\left(^{\min}\otimes_{i=1}^j H_i\right) \otimes \left(^{\min}\otimes_{i=j+1}^n H_i\right) \right),
\]
which is clear, since any element of the left-hand side is of the form~\eqref{nFRedge}, which can in particular be written as a measurement in which either the group $1,\ldots,j$ measures first and communicates their joint outcome to the others, or in which the group $j+1,\ldots,n$ likewise measures first and communicates their joint outcome to $1,\ldots,j$.

For the other inclusion, it remains to prove that
\[
E\left( \left(^{\max}\otimes_{i=1}^j H_i\right)\otimes \left(^{\max}\otimes_{i=j+1}^n H_i\right)\right) \subseteq E\left(^{\max}\otimes_{i=1}^n H_i\right),
\]
which means: given any measurement protocols for parties $1,\ldots,j$ and a subsequent measurement protocol for $j+1,\ldots,n$ given as a function of the outcome of the first, this can be regarded as a measurement protocol for all $n$ parties; and likewise if parties $j+1,\ldots,n$ measure first and the protocol of parties $1,\ldots,j$ is a function of the outcome of the first. But this is also clear.
\end{proof}

In Appendix~\ref{multiproducts}, we study in which sense the different ways of constructing a product of more than two scenarios, meaning the minimal and maximal and all intermediate ones, need to be distinguished for the purposes of this paper. It turns out that as far as general probabilistic models, classical models (Section~\ref{cm}) and quantum models (Section~\ref{qm}) are concerned, these products are equivalent.

We therefore omit this distinction from now on and write $H_1\otimes\ldots\otimes H_n$ when referring to any of these products of scenarios $H_1,\ldots,H_n$.

With these definitions and results, we easily obtain the multipartite generalization of Proposition~\ref{nosigprop}:

\begin{cor}
\label{multinosigcor}
$\mathcal{G}(H_1\otimes\ldots\otimes H_n)\subseteq\mathcal{G}(H_1\times\ldots\times H_n)$ is exactly the subset of models satisfying the no-signaling equations
\beq
\label{multinosigeq}
\sum_{v_i\in e} p(v_1,\ldots,v_n) = \sum_{v_i\in e'} p(v_1,\ldots,v_n)
\eeq
for all parties $i=1,\ldots,n$, all edges $e,e'\in E(H_i)$ and all vertices $v_j\in V(H_j)$ for $j\neq i$.
\end{cor}

\begin{proof}
See Proposition~\ref{multinosigprop} and Corollary~\ref{equalcomp}.
\end{proof}

If one understands this product to be $^{\min}\otimes$, then this result generalizes a well-known fact for Bell scenarios: no-signaling along any bipartition follows from the no-signaling equations of the form~\ref{multinosigeq}, where the sums range only over the outcome of one party.

\subsection{Bell scenarios}
\label{bellscen}

We now explain how Bell scenarios~\cite{nonlocalreview} are examples of contextuality scenarios. The Bell scenario $B_{n,k,m}$ consists of $n$ parties having access to $k$ local measurements each, each of which has $m$ possible outcomes. At the single-party level, the outcomes form a contextuality scenario $B_{1,k,m}$ as depicted in Figure~\ref{singleparty}. As contextuality scenarios, we define
\beq
\label{BSdefn}
B_{n,k,m} \defin \underbrace{B_{1,k,m}\otimes \cdots\otimes B_{1,k,m}}_{n\text{ times}},
\eeq
and we will see in the following how this leads to the usual concepts studied as `nonlocality'. The Foulis--Randall product here can be taken to be any of the products of Appendix~\ref{multiproducts}; for $n\geq 3$, these different products give different, but observationally equivalent scenarios. The Bell scenario $B_{n,k,m}$ for $n\geq 3$ is therefore defined only up to this equivalence.

It is straightforward to generalize this definition and all our upcoming results to scenarios where the parties have access to different numbers of measurements and outcomes per measurement, but we will not consider this explicitly.
\begin{ex}[The CHSH scenario]
Figure~\ref{CHSH_sce} illustrate how the CHSH scenario $B_{2,2,2}$~\cite{CHSH} arises as $B_{1,2,2}\otimes B_{1,2,2}$. A vertex $a b | x y $ represents the event where Alice (resp.~Bob) chooses measurement $x$ (resp.~$y$) and obtains output $a$ (resp.~$b$). In this scenario, the edges are as follows:
\begin{itemize}
\item For simultaneous measurements, the $f$ of~(\ref{FReq}) are constant, and the measurements are as in Figure~\ref{sim_meas}:
\begin{align*}
\{00|00,\: 01|00,\: 10|00,\: 11|00\} ,\\ 
\{00|01,\: 01|01,\: 10|01,\: 11|01\} ,\\
\{00|10,\: 01|10,\: 10|10,\: 11|10\} ,\\
\{00|11,\: 01|11,\: 10|11,\: 11|11\} .
\end{align*}
\item If Alice measures first and Bob's choice of setting depends on her outcome, then the events are of the form $ab|xf(a)$, where $f$ is not a constant. Thus we have two possibilities: $f(a)=a$ or $f(a)=1-a$. In the first case we obtain the edges 
\begin{align*}
\{00|00,\: 01|00,\: 10|01,\: 11|01\} ,\\
\{00|10,\: 01|10,\: 10|11,\: 11|11\} ,
\end{align*}
and in the second case,
\begin{align*}
\{00|01,\: 01|01,\: 10|00,\: 11|00\} ,\\
\{00|11,\: 01|11,\: 10|10,\: 11|10\} .
\end{align*}
These are the red edges in Figures~\ref{CHSHab} and~\ref{CHSHc},\ref{CHSHd}. 
\item Similarly, Bob measuring first with Alice's subsequent choice of setting depending on his outcome gives rise to the edges  
\begin{align*}
\{00|00,\: 01|10,\: 10|00,\: 11|10\} ,\\
\{00|01,\: 01|11,\: 10|01,\: 11|11\} ,\\
\{00|10,\: 01|00,\: 10|10,\: 11|00\} ,\\
\{00|11,\: 01|01,\: 10|11,\: 11|01\} .
\end{align*}
These are the green edges in Figures~\ref{CHSHba} and~\ref{CHSHc},\ref{CHSHd}.
\end{itemize}
\end{ex}

\begin{prop}
\label{Bnosig}
Let $B_{n,k,m}$ be a Bell scenario. Then $\mathcal{G}(B_{n,k,m})$ is the standard no-signaling polytope containing all no-signaling boxes of type $(n,k,m)$, i.e.~conditional probability distributions $p(a_1\ldots a_n|x_1\ldots x_n)$ satisfying the no-signaling equations
\beq
\label{boxnosig}
\sum_{a_k} p(a_1\ldots a_k\ldots a_n|x_1\ldots x_k\ldots x_n) = \sum_{a_k} p(a_1\ldots a_k\ldots a_n|x_1\ldots x'_k\ldots x_n).
\eeq
\end{prop}

\begin{proof}
This is an instance of Proposition~\ref{multinosigprop}, Nevertheless, it is instructive to rephrase part of the material in Appendix~\ref{multiproducts} in the present case, since this shows more explicitly how the no-signaling equations are equivalent to normalization equations for certain joint measurements.

We identify the vertices of $B_{n,k,m}$ with the events
$$
a_1\ldots a_n|x_1\ldots x_n,\quad a_i\in \{1,\ldots,m\}, \: x_i\in \{1,\ldots,k\}
$$
in the usual Bell scenario notation. 

We show first that a non-signaling box of type $(n,k,m)$ satisfies the normalization of probabilities with respect to any measurement in which the choice of measurement $x_i$ of each party is a function of the outcomes of the previous parties, $x_i = f_i(a_1,\ldots,a_{i-1})$. To check this normalization, we need to consider
\beq
\label{targetsum}
\sum_{a_1,\ldots,a_n} p(a_1\ldots a_n | f_1() \ldots f_n(a_1,\ldots a_{n-1})) ,
\eeq
where $x_1=f_1()$ is a function without arguments, i.e.~a constant. Since the list of settings does not depend on $a_n$, the no-signaling equations imply that the last function $f_n(a_1,\ldots,a_{n-1})$ can be replaced by an arbitrary constant setting $x_n$ without changing the value of the sum. After applying this modification, the list of settings does not depend on $a_{n-1}$, and then the setting of party $n-1$ can be taken to be some fixed $x_{n-1}$. Applying this procedure repeatedly eventually replaces all functions $f_i(a_1,\ldots,a_{i-1})$ by constant settings $x_i$. Then the normalization equation
\beq
\label{boxnorm}
\sum_{a_1,\ldots,a_n} p(a_1\ldots a_n|x_1\ldots x_n) = 1
\eeq
implies that the sum has the value $1$, as has been claimed.

Conversely, suppose that $p$ is a probabilistic model on $B_{n,k,m}$; by the results of Appendix~\ref{multiproducts}, we can take this to mean that all sums of the form~\eqref{targetsum} are normalized. Then $p$ satisfies the normalization equation since taking all functions $f_i$ to be constants $x_i$ gives precisely~(\ref{boxnorm}). In order to prove the no-signaling equation, we fix arbitrary outputs $b_j$ and choose all functions to be constants $f_j=x_j$, except for
$$
f_n(a_1,\ldots,a_{n-1}) = \left\{ \begin{array}{cl} x_n & \textrm{if } a_j = b_j \textrm{ for all } j < n, \\ x'_n & \textrm{otherwise,} \end{array}\right. 
$$
which gives the equation
$$
\sum_{a_n} p(b_1\ldots b_{n-1} a_n|x_1\ldots x_n) + \sum_{a_n} \sum_{(a_1,\ldots, a_{n-1}) \ne (b_1,\ldots, b_{n-1}) } p(a_1\ldots a_n|x_1\ldots x'_n) = 1 .
$$
Upon combining this with the already proven normalization equation
$$
\sum_{a_n} p(b_1\ldots b_{n-1} a_n|x_1\ldots x_n') + \sum_{a_n} \sum_{(a_1,\ldots, a_{n-1}) \ne (b_1,\ldots, b_{n-1}) } p(a_1\ldots a_n|x_1\ldots x_n') = 1 ,
$$
we obtain~(\ref{boxnosig}) with $i=n$ and $b_1\ldots b_{n-1}$ in place of $a_1\ldots a_{n-1}$. The other no-signaling equations can be obtained in the same way, choosing different orderings of the parties.
\end{proof}

\newpage
\section{\textbf{Classical models}} 
\label{cm}

For each scenario $H$, one can define several important subsets of $\mathcal{G}(H)$, the set of all probabilistic models on $H$. In the following sections, we will define these and study some of their properties in some detail, starting with set of classical models $\mathcal{C}(H)$ to be treated in this section. We will use the Bell scenarios $B_{n,k,m}$ as `running examples' illustrating that our formalism behaves exactly as it should in order to recover the usual notions~\cite{nonlocalreview} known for Bell scenarios.

\subsection{Definition of classical models}

What we mean by \emph{classical} here comprises the idea of noncontextual deterministic hidden variables as they occur in results of Bell~\cite{Bell}, Fine~\cite{Fine} and Kochen--Specker~\cite{KS}.

\begin{defn}
\label{cmdef}
Let $H$ be a contextuality scenario.
\begin{enumerate}
\item A probabilistic model $p\in\mathcal{G}(H)$ is \emph{deterministic} if $p(v)\in\{0,1\}$ for all $v\in V(H)$.
\item A probabilistic model $p\in\mathcal{G}(H)$ is \emph{classical} if it is a convex combination of deterministic ones: there exist weights $q_\lambda\in[0,1]$ indexed by some parameter $\lambda$ such that $\sum_\lambda q_\lambda = 1$ and deterministic models $p_\lambda$ such that
\[
p(v) = \sum_{\lambda} q_\lambda p_\lambda(v) \qquad\forall v\in V(H).
\]
\end{enumerate}
\end{defn}

Following Fine~\cite{Fine} and certain refinements of his results to considerations of contextuality~\cite{LSW}*{Thm.~6},~\cite{AB}*{Thm.~8.1}, we note that classical models are precisely those which can be explained in terms of noncontextual deterministic hidden variables.

Since, for finite $H$, there are only finitely many deterministic models, the set of classical models is a polytope. We denote this polytope by $\mathcal{C}(H)$.

\begin{ex}[Cabello's~\cite{Cab1} proof of the Kochen--Specker theorem]
For $H_{\mathrm{KS}}$ the contextuality scenario of Figure~\ref{CKSfig}, we claim that $\mathcal{C}(H_{\mathrm{KS}})=\emptyset$, since $H_{\mathrm{KS}}$ does not allow any deterministic models at all. To see this, let $V_1$ be the set of vertices to which a given deterministic model assigns a $1$. Since the set $V_1$ is required to intersect every edge in precisely one vertex, and every vertex appears in precisely two edges, $2|V_1|$ has to be equal to the number of edges. Since the latter is odd, we conclude that this is impossible. Therefore, no deterministic model exists, which means that $\mathcal{C}(H_{\mathrm{KS}})=\emptyset$. See~\cite[Sec.~7.1]{AB} for a very general version of this argument.
\end{ex}

\begin{rem}
\label{transversals}
As we just exemplified, a deterministic model $p$ is determined by the set of vertices
\beq
\label{V1}
V_1 = \left\{ v\in V \:|\: p(v) = 1 \right\}.
\eeq
By definition of deterministic model, $V_1$ has the property that it intersects every edge in exactly one vertex: $V_1$ is an \emph{exact transversal}~\cite{Eiter}. Conversely, every exact transversal $V_1$ defines a deterministic model in this way. We have that $\mathcal{C}(H)\neq \emptyset$ if and only if $H$ has an exact transversal.
\end{rem}

We now apply this definition to Bell scenarios. In the same way that probabilistic models on a Bell scenario coincide with the usual no-signaling boxes (Proposition~\ref{Bnosig}), also classical models coincide with those no-signaling boxes which are commonly called `local':

\begin{prop}
\label{stand-bell-poly}
Let $B_{n,k,m}$ be a Bell scenario. Then $\mathcal{C}(B_{n,k,m})$ is the standard Bell polytope.
\end{prop}

\begin{proof}
This is clear since one way to define the Bell polytope is as the convex hull of deterministic models~\cite{Fine}, and a deterministic model in the contextuality scenario $B_{n,k,m}$ is the same as a local deterministic model in the Bell sense. (This follows e.g.~from an application of Proposition~\ref{Bnosig} to deterministic models.)
\end{proof}

\subsection{Classicality from the fractional packing number}

We now start to relate contextuality scenarios and probabilistic models to graph theory and show how to detect classicality using the \emph{weighted fractional packing number} $\alpha^*$ of the non-orthogonality graph (see Appendix~\ref{appcap} for the definitions of graph-theoretic invariants).
\begin{prop}
\label{Cvsfpn}
A probabilistic model $p\in\mathcal{G}(H)$ is in $\mathcal{C}(H)$ if and only if 
\[
\alpha^*(\mathrm{NO}(H),p)=1.
\]
\end{prop}

The normalization $\sum_{v\in e} p(v)=1$ for every $e\in E(H)$ implies that $\alpha^*(\mathrm{NO}(H),p)\geq 1$, so that the condition $\alpha^*(\mathrm{NO}(H),p)=1$ is equivalent to the seemingly weaker requirement $\alpha^*(\mathrm{NO}(H),p)\leq 1$, which we use in the first part of the proof.

\begin{proof}
We start by showing that if $p$ is classical, then $\alpha^*(\mathrm{NO}(H),p)\leq 1$. By definition, $\alpha^*(\mathrm{NO}(H),p)\leq 1$ means that if $q:V(H)\to [0,1]$ are vertex weights satisfying $\sum_{v \in C} q(v) \leq 1$ for all cliques $C \subseteq\mathrm{NO}(H)$, then also
\beq
\label{sumqp}
\sum_{v\in V(H)} q(v)\, p(v) \leq 1.
\eeq
In order to prove this for all classical $p$, it is sufficient to consider deterministic $p$. In this case, the associated set $V_1=\{v\in V(H)\:|\: p(v)=1\}$ is itself a clique in ${\mathrm{NO}(H)}$, while all other $p(v)$ vanish, and hence~(\ref{sumqp}) follows from the assumption on $q$.

For the other direction, we use the dual formulation~(\ref{fpndual}) of the weighted fractional packing number. The assumption $\alpha^*(\mathrm{NO}(H),p)=1$ then means that there exists a number $x_C\geq 0$ associated to every clique $C\subseteq\mathrm{NO}(H)$ such that $p(v) \leq \sum_{C\ni v} x_C$ and $\sum_{C} x_C = 1$. We claim that every $C$ for which $x_C\neq 0$ corresponds to a deterministic model via~(\ref{V1}); in other words, if $x_C\neq 0$, then $|e\cap C|=1$ for every $e\in E(H)$. First, $|e\cap C|\leq 1$, since $e$ is an independent set in $\mathrm{NO}(H)$ while $C$ is a clique. Second, the chain of inequalities
$$
1 = \sum_{v\in e} p(v) \leq \sum_{v\in e} \sum_{C\ni v} x_C = \sum_{C \textrm{ with } C\cap e\neq \emptyset} x_C \leq \sum_C x_C = 1
$$
actually has to be a chain of equalities, which proves the claim that if $x_C\neq 0$, then $|e\cap C|=1$ for every $e\in E(H)$. Furthermore, we also conclude that $p(v) = \sum_{C\ni v} x_C$, or $p = \sum_C x_C \mathbbm{1}_C$. This is an explicit decomposition of $p$ as a convex combination of deterministic models.
\end{proof}

\begin{prob}
Can this result be used to derive a combinatorial characterization of the facets of $\mathcal{C}(H)$, similar in spirit to Theorem~\ref{extchar}?
\end{prob}

\subsection{Classical models on products}

In particular for Bell scenarios, which are explicitly defined as products~\eqref{BSdefn}, it is important to understand what a classical model on a product scenario looks like. We start with the case of a product of two scenarios before considering products of more than two scenarios.

\begin{prop}
\beq
\label{Cproduct}
\mathcal{C}(H_A\otimes H_B) = \mathrm{conv}\left( \mathcal{C}(H_A) \otimes \mathcal{C}(H_B) \right),
\eeq
where $\mathrm{conv(S)}$ denotes the convex hull of the elements in $\mathrm{S}$.
\end{prop}

This is supposed to be seen in contrast to Remark~\ref{Gproduct}. 

\begin{proof}
Let $p_A\in\mathcal{C}(H_A)$ and $p_B\in\mathcal{C}(H_B)$ be deterministic models. Then also $p_A\otimes p_B$ is a deterministic model on $H_A\otimes H_B$, which proves $\mathcal{C}(H_A\otimes H_B) \supseteq \mathrm{conv}\left( \mathcal{C}(H_A) \otimes \mathcal{C}(H_B) \right)$ by convexity of $\mathcal{C}(H_A\otimes H_B)$.

Conversely, consider a deterministic model $p_{AB}$ on $H_A\otimes H_B$. Let $V_1$ be the set of vertices in $H_A\otimes H_B$ for which $p_{AB}(v)=1$, and define $p_A\in\mathcal{C}(H_A)$ and $p_B\in\mathcal{C}(H_B)$ as follows: for each $v_A \in V_A$, set $p_A(v_A)=1$ if and only if there exists $v_B \in V_B$ such that  $(v_A, v_B) \in V_1$, and $p_A(v_A)=0$  otherwise. Similarly, define $p_B$. We want to check that these are indeed probabilistic models, i.e.~show that $\sum_{v_A \in e_A} p_A(v_A) =1$ and $\sum_{v_B \in e_B} p_B(v_B) =1$ for every edge $e_A$ of $H_A$ and $e_B$ of $H_B$. As $V_1$ is an exact transversal of $H_A\otimes H_B$, no two elements of $V_1$ belong to the same edge. This implies that if both $(v_A, v_B),(v'_A, v'_B) \in V_1$, then there is no $e_A\in E(H_A)$ with $\{v_A,v'_A\} \subseteq e_A$: for if there was, then we could construct an edge in $H_A\otimes H_B$ as in the proof of Lemma~\ref{ortproduct} which contains both $(u_A,u_B)$ and $(u'_A,u'_B)$. It follows that for each edge $e_A \in E_A$, there is at most one vertex $v_A \in e_A$ with $p_A(v_A)=1$. In fact, there is exactly one such vertex, since $e_A\times e_B$ is an edge on $H_A\otimes H_B$ for any $e_B\in E(H_B)$, and this edge must intersect $V_1$. Hence, $p_A$ is a deterministic probabilistic model on $H_A$. The same applies to $p_B$. Since $p_{AB}=p_A\otimes p_B$ by construction, the claim follows by convexity.
\end{proof}

For more than two contextuality scenarios $H_1,\ldots,H_n$, we write $\otimes_{i=1}^n H_i$ for any of the products discussed in Section~\ref{higherprodsmain}.

\begin{cor}
\label{multiC}
\[
\mathcal{C}\left(\otimes_{i=1}^n H_i\right) = \mathrm{conv}\left(\mathcal{C}(H_1)\otimes\ldots\otimes\mathcal{C}(H_n)\right).
\]
\end{cor}

\begin{proof}
As shown in Appendix~\ref{multiproducts}, the left-hand side does not depend on the particular choice of product, so it is enough to prove the claim when $\otimes_{i=1}^n H_i$ stands for an iterated binary product. This follows from repeated application of the previous proposition.
\end{proof}

For Bell scenarios, this indeed recovers the usual Bell polytopes; in terms of the notation of Proposition~\ref{Bnosig}, we have:

\begin{ex}
\label{Bell1CG}
For the `Bell scenario' $B_{1,k,m}$ with one party, $\mathcal{C}(B_{1,k,m}) = \mathcal{G}(B_{1,k,m})$. This corresponds to the known fact that any probabilistic local strategy in a Bell scenario can be rewritten as a convex combination of deterministic local strategies. In our formalism, this can be seen e.g.~as a consequence of the upcoming Theorem~\ref{perfection} together with the fact that $\mathrm{NO}(B_{1,k,m})$ does not contain any independent sets other than the edges, which gives $\mathcal{G}(B_{1,k,m}) = \CE^1(B_{1,k,m})$. We therefore obtain that
\[
\mathcal{C}(B_{n,k,m}) = \mathrm{conv}\left(\mathcal{C}(B_{1,k,m})\otimes\ldots\otimes\mathcal{C}(B_{1,k,m})\right) = \mathrm{conv}\left(\mathcal{G}(B_{1,k,m})\otimes\ldots\otimes\mathcal{G}(B_{1,k,m})\right),
\]
which can be regarded as one version of Fine's theorem~\cite{Fine}.
\end{ex}

\newpage
\section{\textbf{Quantum models}} 
\label{qm}

Quantum models are those probabilistic models which can arise in a world complying with the laws of quantum theory. Understanding the set of quantum models represents one approach for understanding the counterintuitive aspects of quantum theory: if one can find a simple physical or information-theoretic principle which characterizes the set of quantum models, one would have found an indirect explanation for why our world obeys the laws of quantum theory.

In this section then we study how these models may be included in our formalism.We begin with the basic definitions of quantum models, then study quantum models on products, and conclude this section with an explanation of the Kochen--Specker theorem within our framework. 

\subsection{Definition and basic properties}

We denote by $\mathcal{B}(\H)$ the set of all bounded operators on a Hilbert space $\H$. The notation $\mathcal{B}_+(\mathcal{H})$ stands for the subset of positive semi-definite operators. A quantum state $\rho$ is given by a normalized density operator, i.e.~by some $\rho \in \mathcal{B}_{+,1}(\mathcal{H})$, where $\mathcal{B}_{+,1}(\mathcal{H}) \defin \left\{\rho \in \mathcal{B}_{+}(\mathcal{H})\:|\: \mathrm{tr}\, \rho=1 \right\}$.
\begin{defn}
\label{qmdef}
Let $H$ be a contextuality scenario. An assignment of probabilities $p: V(H)\to [0,1]$ is a \emph{quantum model} if there exist a Hilbert space $\mathcal{H}$, a quantum state $\rho\in\B_{+,1}(\H)$ and a projection operator $P_v\in\B(\H)$ associated to every $v\in V$ which constitute projective measurements in the sense that
\beq
\label{qmeas}
\sum_{v\in e} P_v = \mathbbm{1}_{\H} \quad\forall e\in E(H) ,
\eeq
and reproduce the given probabilities,
\beq
\label{qrep}
p(v) = \tr\left( \rho P_v \right) \quad\forall v\in V(H) .
\eeq
\end{defn}

The set of all quantum models is the \emph{quantum set} $\mathcal{Q}(H)$. Thanks to~(\ref{qmeas}), it is clear that $\mathcal{Q}(H)\subseteq\mathcal{G}(H)$, i.e.~every quantum model is a probabilistic model.

\begin{prop}
\label{Qprop}
\begin{enumerate}
\item\label{Qconvex} $\mathcal{Q}(H)$ is convex.
\item\label{CsubQ} Every classical model is a quantum model: $\mathcal{C}(H)\subseteq \mathcal{Q}(H)$.
\end{enumerate}
\end{prop}

\begin{proof}
\begin{enumerate}
\item[\ref{Qconvex}] Let $p_1,p_2\in\mathcal{Q}(H)$ be quantum models described in terms of Hilbert spaces $\H_1$, $\H_2$, projection operators $P_{1,v}$, $P_{2,v}$ and states $\rho_1$, $\rho_2$ on the respective Hilbert space. Then for any coefficient $\lambda\in[0,1]$, we construct a quantum representation of $\lambda p_1 + (1-\lambda) p_2$ by setting
$$
\H \defin \H_1 \oplus \H_2,\qquad P_v\defin P_{1,v} \oplus P_{2,v},\qquad \rho \defin \lambda \rho_1 \oplus (1-\lambda) \rho_2 .
$$
It is immediate to verify that this is indeed a quantum representation of $\lambda p_1 + (1-\lambda) p_2$.
\item[\ref{CsubQ}] This follows from~\ref{Qconvex} upon showing that every deterministic model is quantum. A deterministic model $p$ can be seen to be quantum by setting $\H\defin\C$, $P_v \defin p(v)\cdot\mathbbm{1}$ and $\rho \defin \mathbbm{1}$.\qedhere
\end{enumerate}
\end{proof}

It is important to note that the dimension of $\mathcal{H}$ is not fixed in the definition of quantum model. In general, $\mathcal{H}$ can be infinite-dimensional, and we suspect that in some scenarios, allowing infinite-dimensional $\mathcal{H}$ is necessary for obtaining all quantum models; see Section~\ref{invsandwichsec} for a discussion.

We now prove that there is no quantum analogue of Proposition~\ref{Cvsfpn} relating the property of a probabilistic model $p$ to be quantum to a graph invariant of $\mathrm{NO}(H)$ with weights $p$. Unfortunately, this will require some of the concepts and results from upcoming sections.

\begin{thm}
\label{Qinv}
There exist two contextuality scenarios $H$ and $H'$ with $V(H)=V(H')$ and $\mathrm{NO}(H)=\mathrm{NO}(H')$ together with vertex weights $v\mapsto p(v)$ which define a probabilistic model $p$ both on $H$ and $H'$ such that $p$ is quantum on $H$, but not on $H'$.
\end{thm}

\begin{proof}
Our construction appeals to two ingredients: first, any contextuality scenario $H_0$ with the property that $\mathcal{Q}(H_0)\subsetneq\mathcal{Q}_1(H_0)$, where $\mathcal{Q}_1$ is a semidefinite relaxation of $\mathcal{Q}$ which we will introduce in Section~\ref{dfnhierarchy} and discuss in detail in Section~\ref{secQ1lov}. As a concrete example of such an $H_0$, one may take the Bell scenario $B_{2,2,2}$. The second ingredient is a `gadget' depicted in Figure~\ref{gadget} which can, when suitably added to a given contextuality scenario, control whether a certain operator constraint is enforced by including the dashed edge or not including it.

So let us choose some $p_0 \in \mathcal{Q}_1(H_0) \setminus \mathcal{Q}(H_0)$. We will now modify $H_0$ in two different ways, obtaining scenarios $H$ and $H'$ as in the statement of the theorem, such that $p_0$ extends to the desired probabilistic model $p$. These extensions of $H_0$ are built in two steps: the first step is to use the construction of Section~\ref{CSWtransfer} to add `no-detection' events and obtain an extended scenario $H_1$ for which $\mathcal{Q}(H_1)=\mathcal{Q}_1(H_1)$. So the vertices of $H_1$ are those of $H_0$, together with one additional vertex for every edge of $H_0$,
\[
V(H_1)\defin V(H_0) \cup \{ w_e \::\: e\in E(H_0) \} ,
\]
and the edges are the same, except that each edge contains its corresponding no-detection event,
\[
E(H_1) \defin \left\{ e \cup \{w_e\} \::\: e\in E(H_0) \right\} .
\]
The probabilistic model $p_0$ is easily extended to this new scenario by assigning zero weight to the additional vertices:
$$ p_1(v) \defin \left\{
\begin{array}{ll}
p_0(v) &  \text{if} \quad v \in V(H_0),\\
0 & \text{otherwise.} 
\end{array}
\right.
$$
Thanks to Proposition~\ref{collapse}, the quantum set and its semidefinite relaxation coincide for the extended scenario, $\mathcal{Q}(H_1) =  \mathcal{Q}_1(H_1)$. The probabilistic model $p_1$ belongs to both sets, and when representing it as a quantum model by an assignment of projections $v\mapsto P_v$, it has the following crucial property: although $p_1(w_e)=0$ for every no-detection event $w_e$, there exists at least one no-detection event $w_e$ with $P_{w_e} \ne 0$. For if this were not the case, then $p_0$ would be a quantum model for $H_0$, in contrast to the assumption.

All that is left to do is to extend the scenario one more time with the help of the gadget of Figure~\ref{gadget}, which consists of 10 vertices (the vertex $w'$ present in both parts of the figure should be identified) and either 7 our 8 edges, with the dashed edge being either absent (in $H_\mathrm{gad}$) or present (in $H_\mathrm{gad}'$). More precisely, we construct two scenarios $H$ and $H'$ by taking $H_1$ and attaching a copy of the gadget at each no-detection event $w_e$, where the $w$ vertex of the gadget gets identified with $w_e$. Then $V(H)=V(H')$ and also $\mathrm{NO}(H)=\mathrm{NO}(H')$, since $H$ and $H'$ only differ by copies of the dashed edge of Figure~\ref{gadget}, which does not introduce new orthogonality relations.

We also extend the model $p_1\in\mathcal{G}(H_1)$ to $p\in\mathcal{G}(H^{(}{}'^{)})$ by defining
\beq
\label{gadgetprob}
p(v) \defin \left\{
\begin{array}{ll}
p_1(v) &  \text{if} \quad  v \in V(H_1),\\
1 & \text{if} \quad v \in \{t, t', x, x'\},\\
 0 & \text{otherwise,} 
\end{array}
\right.
\eeq
where it is understood that the second condition applies to all the copies of the vertices $t,t',x,x'$ of Figure~\ref{gadget} which we attached to $H_1$. It is easy to show that all normalization equations hold, including all those for the copies of the dashed edge.

Consider the scenario $H'$. Let us show by contradiction that $p\notin\mathcal{Q}(H')$. The gadget can be analyzed as follows. First, the normalization conditions for the left part of the scenarios depicted on Figure~\ref{gadget} ensure that both the weight and any projection corresponding to vertex $w'$ vanish, $P_{w'}=0$. Consider now the vertices on the right for the scenario $H_\mathrm{gad}'$, i.e.~including the dashed edge. An assignment of projections $P_v$ satisfying the normalization requirement must also satisfy $P_w=0$, since the equations
\[
P_{w'} + P_{x'} + P_{y'} = \mathbbm{1},\qquad P_w + P_{x'} + P_{y'} = \mathbbm{1}
\]
imply that $P_w = P_{w'} = 0$. (In terms of the concepts of Section~\ref{virtualcomp}, we could also state this as $\{w'\}\simeq\emptyset$.) Therefore, since we attached the gadget to each no-detection vertex $w_e$ and identified this vertex with $w$, we need to have $P_{w_e}=0$ for any quantum model on $H'$ and any no-detection event $w_e$. Then, if $p$ admitted a quantum model on $H'$, this would imply that the original $p_0$ must already have been quantum, which we assumed not to be the case.

On the other hand, we now show that $p$ is quantum on $H$. To this end, we take projections $P_v$ for the vertices in $H_1$ and a state $\rho$ which witness that $p_1\in\mathcal{Q}(H_1)$; it remains to assign projections to the vertices of each copy of the gadget such that the other probabilities in~\eqref{gadgetprob} are reproduced and normalization holds for all edges except for the dashed one. This can be done by putting
\[
P_t = P_{t'} = P_x \defin \mathbbm{1}, \qquad P_{x'} \defin \mathbbm{1}-P_w, \qquad P_y = P_w,
\]
where $P_w = P_{w_e}$ is part of the given data, and we assign the zero projection to all other vertices. Checking normalization is straightforward, while the probabilities are reproduced thanks to the assumption $\tr(\rho P_{w_e}) = 0$, which implies $\tr(\rho P_{x'}) = 1$. This proves that $p \in \mathcal{Q}(H)$.
\end{proof}

\begin{figure}
\begin{center}
\begin{tikzpicture}[scale=1.4]
\node[draw,shape=circle,fill,scale=.5] (x) at (2.2,0) {} ;
\node[left of=x,node distance=3mm] {$w'$};
\node[draw,shape=circle,fill,scale=.5](t') at (1,-1) {} ;
\node[draw,shape=circle,fill,scale=.5] at (1,1) {} ;
\node[draw,shape=circle,fill,scale=.5] at (0,-1) {} ;
\node[draw,shape=circle,fill,scale=.5](t) at (0,1) {} ;
\draw[thick,blue] (0.5,1) ellipse (1.2cm and .3cm) ;
\draw[thick,blue] (0.5,-1) ellipse (1.2cm and .3cm) ;
\draw[thick,blue] (0,0) ellipse (.3cm and 1.2cm) ;
\draw[thick,blue] plot [smooth cycle,tension=.5] coordinates { (2.4,0) (0.9,1.2) (0.9,-1.2) } ;
\node[draw,shape=circle,fill,scale=.5] (v') at (3,1.73) {} ;
\node[draw,shape=circle,fill,scale=.5] (w) at (5,1.73) {} ;
\node[draw,shape=circle,fill,scale=.5] (u') at (7,1.73) {} ;
\node[draw,shape=circle,fill,scale=.5] (u) at (4,0) {} ;
\node[draw,shape=circle,fill,scale=.5] (v) at (6,0) {} ;
\node[draw,shape=circle,fill,scale=.5] (w') at (5,-1.73) {} ;
\node[left of=t,node distance=4mm] {$t$};
\node[right of=t',node distance=3mm] {$t'$};
\node[above of=u,node distance=3mm] {$w'$};
\node[above of=u',node distance=3mm] {$w$};
\node[above of=v,node distance=3mm] {$x'$};
\node[above of=v',node distance=3mm] {$x$};
\node[above of=w,node distance=3mm] {$y'$};
\node[above of=w',node distance=3mm] {$y$};
\draw[thick,blue] plot [smooth cycle,tension=.7] coordinates { (3-0.28,1.73+0.28) (5.28,1.73+0.28) (4,-0.4)} ;
\draw[thick,blue] plot [smooth cycle,tension=.7] coordinates { (5-0.28,1.73+0.28) (7.28,1.73+0.28) (6,-0.4)} ;
\draw[thick,blue] plot [smooth cycle,tension=.7] coordinates { (4-0.28,0.28) (6.28,0.28) (5,-1.73-0.4)} ;
\draw[thick,red,dashed] plot [smooth cycle,tension=.7] coordinates { (4-0.4,-0.4) (6.4,-0.4) (5,1.73+0.57)} ;
\end{tikzpicture}
\end{center}
\caption{Description of the contextuality scenarios $H_\mathrm{gad}$ and $H_\mathrm{gad}'$. The vertex $w'$ should be identified in both parts of the figure, so that there are $10$ vertices in total. $H_\mathrm{gad}$ is defined to comprise all edges except for the dashed one, while $H_\mathrm{gad}'$ does contain the dashed edge in addition to the other ones. This results in $V(H_\mathrm{gad})=V(H_\mathrm{gad}')$ and $\mathrm{NO}(H_\mathrm{gad})=\mathrm{NO}(H_\mathrm{gad}')$. The left part plays the role of forcing $P_{w'} = 0$ for any quantum model. In $H_\mathrm{gad}$, the dashed edge also forces $P_w=0$, although $P_w$ can be arbitrary in $H_\mathrm{gad}$.}
\label{gadget}
\end{figure}
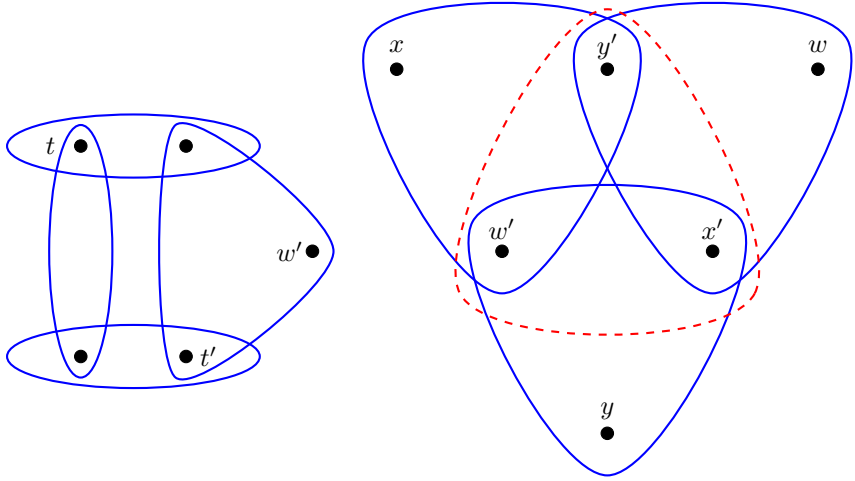

\subsection{Quantum models on products}

What is the set of quantum models on a product scenario $H_A\otimes H_B$? The following characterization generalizes the \emph{commutativity paradigm} of quantum correlations in Bell scenarios~\cite{Jetal,TF}. For a related argument, see~\cite[App.~(iv)]{CSW}.

\begin{prop}
\label{FRquantum}
Let $H_A$ and $H_B$ be two contextuality scenarios and $p\in\mathcal{G}(H_A\otimes H_B)$. Then $p$ is quantum if and only if there is a Hilbert space $\H$, a quantum state $\rho\in\B_{+,1}(\H)$ and projection operators $P_{A,u}\in\B(\H)$ and $P_{B,v}\in\B(\H)$ assigned to every $u\in V(H_A)$ and $v\in V(H_B)$ such that
\begin{align}
\sum_{u\in e_A} P_{A,u} = \mathbbm{1}_\H &= \sum_{v \in e_B} P_{B,v} && \forall e_A \in E(H_A),\: e_B\in E(H_B) , \nonumber\\[.2cm]
 [P_{A,u},P_{B,v}]&= 0 && \forall u\in V(H_A),\: v\in V(H_B) ,\nonumber\\
\intertext{and the given probabilistic model is reproduced,}
\label{usualqcFR}
 p(u,v) &= \tr \left( \rho\, P_{A,u} P_{B,v} \right) && \forall u\in V(H_A),\: v\in V(H_B).
\end{align}
\end{prop} 

\begin{proof}
We start from the alternative conditions of the theorem and a probabilistic model $p$ of the form~(\ref{usualqcFR}) and show that it is a quantum model in the original sense. To this end, we assign to every vertex $(u,v)\in V(H_A \otimes H_B)$ the projection 
$$P_{(u,v)} \defin P_{A,u} P_{B,v},$$
so that~(\ref{qrep}) holds by the assumption~(\ref{usualqcFR}). It remains to show~\eqref{qmeas}, i.e.~normalization of the projections for any measurement on $H_A\otimes H_B$. By symmetry, it is sufficient to prove this for an edge $e \in E_{A \rightarrow B}$ given by
$$
e = \bigcup_{a\in e_A} \{a\} \times f(a) \quad\text{ with }\quad e_A\in E_A,\: f:e_A\to E_B.
$$
In this case, 
$$\sum_{w \in e} P_w = \sum_{u \in e_A} P_{A,u} \sum_{v \in f(u)} P_{B,v} = \sum_{u\in e_A} P_{A,u} \cdot \mathbbm{1}_{\mathcal{H}} = \mathbbm{1}_\mathcal{H}, $$
which is analogous to the computation in the proof of Proposition~\ref{PM_prod}. 

Conversely, one can construct the `local' observables $P_{A,u}$ and $P_{B,v}$ from a quantum model on $\mathcal{Q}(H_A \otimes H_B)$ by noting that the operators
\beq
\label{localops}
P_{A,u} \defin \sum_{v \in e_B} P_{(u,v)}, \qquad P_{v} \defin \sum_{u \in e_A} P_{(u,v)}
\eeq
do not depend on the choice of $e_B \in E(H_B)$ or $e_A \in E(H_A)$, respectively. To see this, it is enough to prove the operator-valued no-signaling equations
\beq
\label{NSgoal}
\sum_{v\in e_B} P_{(u,v)} = \sum_{v\in e'_B} P_{(u,v)}
\eeq
for any $u\in V(H_A)$ and $e_B,e'_B\in E(H_B)$, which can be done just as in the proof of Proposition~\ref{Bnosig}: choosing some $e_A\in E(H_A)$ with $u\in e_A$ and considering the function $f:e_A\to E(H_B)$ with
$$
f(u') \defin \begin{cases} e_B & \textrm{if } u'=u, \\ e'_B & \textrm{otherwise}, \end{cases}
$$
results in the normalization equation
\[
\sum_{v\in e_B} P_{(u,v)} + \sum_{u'\in e_A\setminus\{u\}} \sum_{v\in e'_B} P_{(u',v)} = \mathbbm{1}_\H.
\]
Comparing this with the normalization equation for the edge $e_A\times e_B'$ gives
$$
\sum_{v\in e_B} P_{(u,v)} + \sum_{u'\in e_A\setminus\{u\}} \sum_{v\in e'_B} P_{(u',v)} = \sum_{u'\in e_A} \sum_{v\in e'_B} P_{(u',v)} ,
$$
which reduces to~(\ref{NSgoal}) after canceling terms. This shows that the `local' operators~(\ref{localops}) are well-defined.

The normalization condition $\sum_{u\in e_A} P_{A,u} = \mathbbm{1}_\H = \sum_{v\in e_B} P_{B,v}$ for any $e_A\in E(H_A)$ and $e_B\in E(H_B)$ now is an immediate consequence of the normalization $\sum_{u\in e_A,v\in E_B} P_{(u,v)} = \mathbbm{1}_\H$. Finally, the commutativity $[P_{A,u},P_{B,v}]=0$ for given $u\in V(H_A)$ and $v\in V(H_B)$ follows again from the normalization 
$$
\sum_{u' \in e_A} \sum_{v' \in e_B} P_{(u',v')} = \mathbbm{1}_\mathcal{H} ,
$$ 
for any $e_A$ and $e_B$ which contain $u$ and $v$, respectively: the terms in this sum are necessarily mutually orthogonal projections, and hence commute pairwise; but now both $P_{A,u}$ and $P_{B,v}$ are partial sums of this big sum, and therefore these commute as well. Also, mutual orthogonality implies $P_{(u,v)} = P_{A,u} P_{B,v}$, which yields the desired probabilities~(\ref{usualqcFR}).
\end{proof}

Quantum models on product scenarios arise typically as follows: $\H$ itself may be given as a tensor product $\H_A\otimes\H_B$, such that every $P_{A,u}$ operates on the first factor, while every $P_{B,v}$ operates on the second, while $\rho$ is a state on $\H_A\otimes\H_B$, possibly entangled. Quantum models of this form are said to follow the \emph{tensor paradigm}. We do not know whether every quantum model on $\mathcal{Q}(H_A\otimes H_B)$ arises, at least approximately, from the tensor paradigm. This question is a generalization of Tsirelson's problem~\cite{Jetal,TF}, and can be asked more precisely like this:

\begin{prob}
Is the set of all quantum models with the tensor paradigm dense in $\mathcal{Q}(H_A\otimes H_B)$?
\end{prob}

If we have a quantum model on $H_A$ and a quantum model on $H_B$, the tensor product of the underlying Hilbert spaces and projection operators therefore gives a quantum model, with the tensor paradigm, on $H_A\otimes H_B$. In other words, we have:

\begin{cor}
\label{Qtensor}
\beq
\label{eqQtensor}
\mathcal{Q}(H_A) \otimes \mathcal{Q}(H_B) \subseteq \mathcal{Q}(H_A\otimes H_B)
\eeq
\end{cor}

Again, the Bell scenario $B_{2,2,2} = B_{1,2,2}\otimes B_{1,2,2}$ exemplifies that~(\ref{eqQtensor}) is not an equality in general.

For products of more than two scenarios $H_1,\ldots,H_n$, we write $\otimes_{i=1}^n H_i$ for any of the products discussed in Section~\ref{higherprodsmain}. Then we have a generalization of Proposition~\ref{FRquantum}:

\begin{prop}
\label{multiQFR}
Let $p\in\mathcal{G}(\otimes_{i=1}^n H_i)$. Then $p$ is quantum if and only if there is a Hilbert space $\H$, a quantum state $\rho\in\B_{+,1}(\H)$ and projection operators $P_{i,v}$ assigned to every party $i$ and $v\in V(H_i)$ such that
\begin{align*}
\sum_{v\in e} P_{i,v} & = \mathbbm{1}_\H  && \forall i,\: e\in E(H_i), \\[.2cm]
[P_{i,u},P_{j,v}] &= 0 && \forall i\neq j,\: u\in V(H_i),\: v\in V(H_j),
\end{align*}
and the given probabilistic model is reproduced,
\beq
p(v_1,\ldots,v_n) = \tr\left( \rho P_{1,v_1}\ldots P_{n,v_n}\right) \qquad \forall v_1\in V(H_1),\ldots,v_n\in V(H_n).
\eeq
\end{prop}

In contrast to the tensor paradigm mentioned above, quantum models of this form satisfy the \emph{commutativity paradigm}. The tensor paradigm is a special case of the commutativity paradigm.

\begin{proof}
As shown in Appendix~\ref{multiproducts}, the quantum set $\mathcal{Q}(\otimes_{i=1}^n H_i)$ does not depend on the particular choice of product, so it is sufficient to prove the statement for the minimal product $^{\min}\otimes_{i=1}^n H_i$. In this case, the proof is analogous to the previous one of Proposition~\ref{FRquantum}.
\end{proof}

Again, one can also introduce the tensor paradigm for products of more than two scenarios, and the question arises whether the set of quantum models with the tensor paradigm is dense in the quantum set. We are far from being able to answer this question.

Upon recalling that the Bell scenario $B_{n,k,m}$ is an $n$-fold product of scenarios $B_{1,k,m}$ which describe $k$ independent $m$-outcome measurements, we immediately deduce that quantum models on $B_{n,k,m}$ correspond exactly to the usual `quantum correlations' in Bell scenarios, at least when considering the commutativity paradigm as in Proposition~\ref{multiQFR}:

\begin{cor}
\label{Bellquantum}
$\mathcal{Q}(B_{n,k,m})$ is the set of quantum correlations in the Bell sense with the commutativity paradigm.
\end{cor}

\subsection{The Kochen--Specker theorem and state-independent proofs}
\label{siKS}

We conclude this section by discussing the formulation of `state-independent' proof of the Kochen--Specker theorem like~\cite{Cab1,simplest} in our formalism. A scenario $H_{\mathrm{KS}}$ provides a proof of the Kochen--Specker theorem as soon as $\mathcal{C}=\emptyset$, although $\mathcal{Q}(H_{\mathrm{KS}})\neq\emptyset$; see Figure~\ref{CKSfig} for an example. 

\begin{theo}[Kochen--Specker]
\label{CKSthm}
There exists a contextuality scenario $H_{\mathrm{KS}}$ for which
$$
\mathcal{C}(H_{\mathrm{KS}}) = \emptyset ,\qquad \mathcal{Q}(H_{\mathrm{KS}}) \neq \emptyset .
$$
\end{theo}

This is automatically `state-independent' in the following sense: since $\mathcal{Q}(H_{\mathrm{KS}})$ is not empty, there exists an assignment of a projection $P_v$ to each $v\in V(H_{\mathrm{KS}})$ satisfying the normalization~\eqref{qmeas}. Then one can take any state on the same Hilbert space and obtain a quantum model, which is necessarily not classical. In particular, this non-classicality is independent of the particular state that one chooses. One can find candidate scenarios $H_{\mathrm{KS}}$ by searching for configurations of vectors in a finite-dimensional Hilbert space such that each vector occurs in a basis containing only vectors from the configuration~\cite{P3M}. Upon taking the vertices of the scenario to be given by these vectors and the edges by the bases, one automatically has a quantum representation. The problem lies in choosing the configuration such that the resulting scenario satisfies $\mathcal{C}(H_{\mathrm{KS}})=\emptyset$. This is exactly the approach pursued in~\cite{P3M}.

One may wonder whether these scenarios permit probabilistic models which are not quantum. In the particular example of Figure~\ref{CKSfig}, It is not clear to us whether $\mathcal{Q}(H_{\mathrm{KS}}) = \mathcal{G}(H_{\mathrm{KS}})$ holds, but we suspect that this is not the case. A natural question is whether there exists a proof of the Kochen--Specker in which it is the case:

\begin{prob}
\label{strongKS}
Does there exist a contextuality scenario $H$ for which $\mathcal{C}(H)=\emptyset$, but $\mathcal{Q}(H) = \mathcal{G}(H)\neq\emptyset$?
\end{prob}

Some hypergraph $H$ constructed from the GHZ paradox~\cite{GHZ} might be a good candidate for this hypothetical phenomenon. Our earlier results allow us to reformulate this problem:

\begin{prop}
There exists $H$ as in Problem~\ref{strongKS} if and only there exists some $H'$ with a unique probabilistic model which is quantum, but not classical.
\end{prop}

\begin{proof}
Clearly if such an $H'$ exists, then we can take $H=H'$ in Problem~\ref{strongKS}. Conversely, such an $H'$ can be constructed as an induced subscenario of any $H$ of Problem~\ref{strongKS} by using Theorem~\ref{extchar}, whose proof adapts immediately to show that the resulting unique probabilistic model $H'$ will also be quantum.
\end{proof}

\newpage
\section{\textbf{A hierarchy of semidefinite programs characterizing quantum models}}
\label{npahierarchy}

In general, it is very difficult to determine whether a given probabilistic model $p\in\mathcal{G}(H)$ is quantum or not. In fact, as we discuss in Section~\ref{invsandwichsec}, it is conceivable that no algorithm exists for determining this! Hence, it is important to have good approximations to $\mathcal{Q}(H)$ for which membership can be algorithmically determined. For Bell scenarios, this is achieved by the hierarchy of semidefinite programs characterizing quantum correlations with the commutativity paradigm due to Navascu{\'e}s, Pironio and Ac{\'i}n~\cite{NPA0,NPA}. Since its discovery, it has found manifold applications in quantum information theory like~\cite{rand,DIQKD}. Here, we extend this hierarchy of semidefinite programs from Bell scenarios to all contextuality scenarios. While this formulation is new at this precise level of generality, it may also be considered a special case of the general hierarchy for noncommutative polynomial optimization~\cite{NPA2}.

\subsection{Definition of the hierarchy}
\label{dfnhierarchy}

We introduce the main idea before getting to the technical details. Given a quantum model as in Definition~\ref{qmdef}, not only can one consider the expectation values $\tr\left(\rho P_v\right)$, but also any expectation value of the form
\beq
\label{longexp}
\tr\left( \rho P_{v_1}\ldots P_{v_n}\right),
\eeq
where $\mathbf{v}=v_1\ldots v_n\in V(H)^n$ is any finite sequence of vertices. The idea is to find properties of these collections of values which characterize quantum models. For $n\geq 2$, these values are typically not determined by the probabilities $p(v) = \tr\left( \rho P_v \right)$ alone; the hierarchy works with these quantities as unknown variables whose values have to be determined in such a way that the whole collection of values becomes consistent with and in fact specifies a quantum model in an essentially unique way.

Now for some notation. We write $P_{\mathbf{v}}$ as shorthand for the operator $P_{v_1}\ldots P_{v_n}$, although this product is in general not a projection. When $V$ is a set, we write $V^{*n}$ for the set of all strings of up to $n$ elements of $V$, i.e.~$V^{*n} = \bigcup_{k\leq n} V^k$, and $V^* = \bigcup_{k\in\N} V^k$ for the set of all strings of any length over $V$. $\emptyset\in V^{*}$ is the empty string of length $0$ and the associated operator is $P_\emptyset \defin \mathbbm{1}$. For $\mathbf{v}=v_1\ldots v_n$ a string, we denote its reverse by $\mathbf{v}^\dag = v_n\ldots v_1$. This notation makes sense in our context since $P_{\mathbf{v}^\dag} = P_{\mathbf{v}}^\dag$. For strings $\mathbf{v}\in V^*$ and $\mathbf{w}\in V^*$, we write their concatenation simply as $\mathbf{vw}\in V^{*}$, so that $P_{\mathbf{vw}}=P_{\mathbf{v}}P_{\mathbf{w}}$. We also use $v_1 \ldots \bcancel{v_i} \ldots v_n$ as a shorthand for $v_1 \ldots v_{i-1} v_{i+1} \ldots v_n$.

So in our new notation,~\eqref{longexp} can be written as $\tr\left(\rho P_{\mathbf{v}}\right)$. We now start by studying the properties of the collection of these values, indexed by $\mathbf{v}$.

\begin{lem}
\label{npan}
Let $p\in\mathcal{Q}(H)$ be a quantum model with projections $v\mapsto P_v$ on a Hilbert space $\H$ and state $\rho\in\B_{+,1}(\H)$. Then the matrix $M$ indexed by $\mathbf{v},\mathbf{w}\in V(H)^{*n}$ with entries
\beq
\label{npamatrix}
M_{\mathbf{v},\mathbf{w}} \defin \tr\left(\rho P_{\mathbf{v}} P_{\mathbf{w}}^\dag \right) = \tr\left(\rho P_{\mathbf{v}\mathbf{w}^\dag}\right)
\eeq
has the following properties:
\begin{enumerate}
\item\label{Mposdef} $M$ is positive semidefinite. 
\item  \beq \label{npanorm} M_{\emptyset,\emptyset} = 1. \eeq 
\item For every $e\in E(H)$,
\beq
\label{npalin}
\sum_{x \in e} M_{\mathbf{v}x,\mathbf{w}} = M_{\mathbf{v},\mathbf{w}} .
\eeq
\item If $\mathbf{v}=v_1\ldots v_k$ and $\mathbf{w}=w_1\ldots w_m$ with $v_k\perp w_m$, then
\beq
\label{npaort}
M_{\mathbf{v}, \mathbf{w}} = 0 .
\eeq
\end{enumerate}
\end{lem}

Hermiticity of $M$, which is contained in claim~\ref{Mposdef}, implies that~(\ref{npalin}) also holds with $x$ appended to $\mathbf{w}$ rather than to $\mathbf{v}$.

\begin{proof}
\begin{enumerate}
\item It needs to be shown that for any vector $x\in\C^{V(H)^{*n}}$ with components $x_{\mathbf{v}}\in\C$, $\mathbf{v}\in V(H)^{*n}$, the number 
$$
\sum_{\mathbf{v},\mathbf{w}} x^*_{\mathbf{v}} M_{\mathbf{v},\mathbf{w}} x_{\mathbf{w}} .
$$
is nonnegative. By the definition~(\ref{npamatrix}), this is equal to
$$
\sum_{\mathbf{v},\mathbf{w}} \tr\left(\rho\: x^*_{\mathbf{v}} P_{\mathbf{v}} P_{\mathbf{w}}^\dag x_{\mathbf{w}} \right) .
$$
With $Q \defin \sum_{\mathbf{v}} x_{\mathbf{v}} P_{\mathbf{v}}^\dag$, this is of the form $\tr\left(\rho Q^\dag Q\right)$, and therefore indeed nonnegative.

\item Since $\rho$ is a normalized state, $M_{\emptyset,\emptyset} = \tr(\rho)=1$.

\item The normalization requirement~(\ref{qmeas}) implies that
$$
\sum_{x\in e} P_{\mathbf{v}x} = P_{\mathbf{v}},
$$
from which~(\ref{npalin}) directly follows.

\item This is a direct consequence of $P_{v_k} \perp P_{w_m}$ for $v_k\perp w_m$, which implies that $P_{\mathbf{v}}P_{\mathbf{w}}^\dag=0$.\qedhere
\end{enumerate}
\end{proof}

The diagonal entries $M_{\mathbf{v},\mathbf{v}}$ represent the expectation values $\tr\left( \rho P_{v_1}\ldots P_{v_n} P_{v_n}\ldots P_{v_1} \right)$ which can be interpreted as the probability to obtain the sequence of outcomes $v_1,\ldots,v_n$, given that a sequence of measurement $e_1,\ldots,e_n$ is being conducted with $v_i\in e_i\:\forall i$ and the state collapses as usual for a projective measurement. We suspect that this interpretation can be used to find an interpretation of the `higher' levels of the hierarchy in terms of the lowest level of a temporally extended scenario, but we have not been able to get this idea to work. 

We now define a hierarchy of probabilistic models and its levels, based on the properties of the matrix $M$ discovered in Lemma~\ref{npan}.

\begin{defn}
\label{npadefn}
Let $H$ be a contextuality scenario. We say that $p:V(H)\to [0,1]$ is a \emph{$\mathcal{Q}_n$-model} if there exists a positive semidefinite matrix $M$, with entries $M_{\mathbf{v},\mathbf{w}}$ indexed by $\mathbf{v},\mathbf{w}\in V(H)^{*n}$, such that~(\ref{npanorm}),~(\ref{npalin}),~(\ref{npaort}) hold and the given probabilities are recovered,
\beq
\label{nparep}
p(v) = M_{v,\emptyset} .
\eeq
\end{defn} 

It is an easy consequence of~\eqref{npanorm} and~\eqref{npalin} that every $\mathcal{Q}_n$-model is a probabilistic model. By definition, testing whether a given probabilistic model lies in $\mathcal{Q}_n$ is a semidefinite programming problem of size roughly $|V(H)|^n\times |V(H)|^n$. By making judicious use of the equations~(\ref{npalin}) and the upcoming~(\ref{npaorth2}), this size can be significantly reduced if $H$ has many edges; any practical computation should take this into account. Furthermore, it can be assumed that all matrix entries are actually in $\R$, i.e.~no imaginary components are needed: if a certain complex matrix $M$ satisfies all the given requirements, then so does its complex conjugate $\bar{M}$, and therefore also the real matrix $\tfrac{1}{2}(M+\bar{M})$.

Definition~\ref{npadefn} is our semidefinite hierarchy for contextuality scenarios. In the special case of a bipartite Bell scenario $B_{2,k,m}$, our hierarchy is equivalent to the original one~\cite{NPA0,NPA}, although our level `$n$' is somewhat different from the hierarchy level `$n$' used in~\cite{NPA}. In particular, our set $\mathcal{Q}_1(B_{2,k,m})$ is the set $Q^{1+AB}$ of~\cite{NPA}; see Section~\ref{Q1products}.

\begin{prop}\label{QinQ1}
$$
\mathcal{Q}(H) \subseteq \ldots \subseteq \mathcal{Q}_n(H) \subseteq \ldots \subseteq \mathcal{Q}_1(H) .
$$
\end{prop}

\begin{proof}
Every matrix $M$ showing that $p$ is a $\mathcal{Q}_{n+1}$-model can be restricted to a matrix showing that $p$ is a $\mathcal{Q}_n$-model, so that $\mathcal{Q}_{n+1}(H)\subseteq \mathcal{Q}_n(H)$. Furthermore, Lemma~\ref{npan} shows that every quantum model is a $\mathcal{Q}_n$-model, which means that $\mathcal{Q}(H)\subseteq\mathcal{Q}_n(H)$.
\end{proof}

\begin{rem}
\label{npaproplist}
Besides those of Lemma~\ref{npan}, there are other properties satisfied by the matrices $M$ which follow from~(\ref{npanorm})--(\ref{npaort}), and are satisfied in particular by those $M$ of the form~(\ref{npamatrix}). In the following list, it is understood that all relevant strings $\mathbf{v},\mathbf{w},\ldots$ are assumed to be of a length which guarantees that all matrix entries are defined at the hierarchy level $n$ that is being considered.
\begin{enumerate}
\item\label{npahankelitem} If $\mathbf{v} \mathbf{w}^\dag = \mathbf{v}' \mathbf{w}'^\dag$, then
\beq
\label{npahankel}
M_{\mathbf{v},\mathbf{w}} = M_{\mathbf{v}',\mathbf{w}'} .
\eeq 
This follows by induction from $M_{v_1\ldots v_m,\mathbf{w}} = M_{v_1\ldots v_{m-1},\mathbf{w}v_m}$, which in turn can be shown as follows. Upon choosing some $e\in E(H)$ with $v_m\in e$, properties~(\ref{npalin}) and~(\ref{npaort}) yield
$$
M_{v_1\ldots v_m,\mathbf{w}} \stackrel{(\ref{npalin})}{=} \sum_{x\in e} M_{v_1\ldots v_m,\mathbf{w}x} \stackrel{(\ref{npaort})}{=} M_{v_1\ldots v_m,\mathbf{w}v_m} .
$$
Applying the same trick on the other side shows that this also equals $M_{v_1\ldots v_{m-1},\mathbf{w}v_m}$, as claimed.

Property~(\ref{npahankel}) implies in particular that all matrix entries $M_{\mathbf{v},\mathbf{w}}$ are determined by those of the first row, i.e.~those of the form $M_{\emptyset,\mathbf{v}}$, although this requires $\mathbf{v}\in V(H)^{*2n}$.
\item Every matrix entry can be bounded by diagonal ones,
\beq
\label{bounddiag}
|M_{\mathbf{v},\mathbf{w}}|^2\leq M_{\mathbf{v},\mathbf{v}}\cdot M_{\mathbf{w},\mathbf{w}}.
\eeq
This follows from positive semidefiniteness of the $2\times 2$-submatrix
\[
\left(\begin{matrix} M_{\mathbf{v},\mathbf{v}} & M_{\mathbf{v},\mathbf{w}} \\ M_{\mathbf{w},\mathbf{v}} & M_{\mathbf{w},\mathbf{w}} \end{matrix}\right)
\]
by taking the determinant.
\item Choosing some $e\in E(H)$ with $v\in e$ and applying~(\ref{npalin}) and~(\ref{npaort}) also shows that
\beq
\label{diagonal}
M_{v,\emptyset} = M_{v,v} .
\eeq
In particular, $p(v) = M_{v,v}$ by~(\ref{nparep}).
\item 
A diagonal entry can be bounded by ``shorter'' diagonal ones: if $j\leq m$, then
\beq
\label{temporal}
M_{v_1\ldots v_m,v_1\ldots v_m} \leq M_{v_1\ldots v_j,v_1\ldots v_j}.
\eeq
It is sufficient to show this when $m=j+1$; the general case then follows by induction. In this case, we choose $e\in E(H)$ with $v_{j+1}\in e$ and obtain
\[
M_{v_1\ldots v_j,v_1\ldots v_j} \stackrel{\eqref{npalin}}{=} \sum_{x,y\in e} M_{v_1\ldots v_j x,v_1\ldots v_j y} \stackrel{\eqref{npaort}}{=} \sum_{x\in e} M_{v_1\ldots v_j x, v_1\ldots v_j x}.
\]
Since each summand on the right-hand side is a diagonal matrix element, all these summands are non-negative and the claim follows.
\end{enumerate}
The following properties hold in addition if the length of $\mathbf{v}\mathbf{w}^\dag$ is at most $n$:
\begin{enumerate}[resume]
\item For every $e\in E(H)$,
\beq
\label{npalin2}
\sum_{v_i \in e} M_{\mathbf{v},\mathbf{w}} = M_{v_1\ldots \bcancel{v_i}\ldots v_m,\mathbf{w}} .
\eeq
This is a consequence of~(\ref{npalin}) and~(\ref{npahankel}).
\item Erasing a repetition $v_{j+1}=v_j$ from the index string gives the same matrix entry,
\beq
\label{npaip}
 M_{v_1\ldots v_j v_{j+1} \ldots v_m,\mathbf{w}} = M_{v_1\ldots v_j \bcancel{v_{j+1}} \ldots v_m,\mathbf{w}} .
\eeq
Upon using~(\ref{npahankel}), this follows from a very similar argument.
\item Having subsequent orthogonal indices makes the matrix entry vanish,
\beq
\label{npaorth2}
v_j\perp v_{j+1} \quad\Longrightarrow\quad M_{v_1\ldots v_j v_{j+1} \ldots v_m,\mathbf{w}} = 0 .
\eeq
This follows from~(\ref{npalin2}) together with~(\ref{npaip}).
\end{enumerate}
\end{rem}

\subsection{Convergence of the hierarchy}

As we now know, the $\mathcal{Q}_n$-family of sets constitutes a sequence of outer approximations to the quantum set. But does this sequence converge to the quantum set? In other words, if $p\in\mathcal{Q}_n$ for all $n\in\N$, does this imply that $p\in\mathcal{Q}$? We will now see that the answer to this is positive, so that the hierarchy in fact \emphalt{characterizes} quantum models.

We might also consider infinite matrices $M$ with entries $M_{\mathbf{v},\mathbf{w}}$ indexed by strings of arbitrary length $\mathbf{v},\mathbf{w}\in V(H)^*$; starting from a quantum model and considering~(\ref{npamatrix}) as the resulting definition of the matrix, the same proof as before shows that the properties of Lemma~\ref{npan} still hold, if we take positive semidefiniteness to mean that
$$
\sum_{\mathbf{v},\mathbf{w}\in V(H)^*} x_{\mathbf{v}}^* M_{\mathbf{v},\mathbf{w}} x_{\mathbf{w}} \geq 0
$$
for all finitely supported $(x_{\mathbf{v}})_{\mathbf{v}\in V(H)^*}$.

\begin{prop}
\label{GNS}
If such an infinite matrix exists, then $p\in\mathcal{Q}$.
\end{prop}

\begin{proof}
We prove this first by giving the short high-level explanation, and then provide more details on what this means explicitly.

Abstractly, such an infinite matrix $M$ can be understood to be a ($*$-algebraic) state $\phi$ on the $*$-algebra with generators $\{P_v,\: v\in V(H)\}$ and relations
\beq
\label{algpres}
P_v = P_v^2 = P_v^* ,\qquad \sum_{v\in e} P_v = \mathbbm{1} \quad \forall e\in E(H)
\eeq
via the assignment
$$
\phi\left(P_{v_1}\ldots P_{v_n}\right) \defin M_{v_1\ldots v_n,\emptyset} .
$$
and extending by linearity. Then, the GNS construction (see e.g.~\cite{KR}) turns this into a quantum representation recovering the given probabilities~(\ref{nparep}). For this reason, a probabilistic model is quantum if and only if there exists such an infinite matrix $M$ having the properties of Lemma~\ref{npan}.

If one turns this prescription into an explicit construction, one obtains the following. First, we claim that
\beq
\label{projwelldef}
\sum_{\mathbf{v},\mathbf{w}\in V(H)^*} x_{\mathbf{v}}^* M_{\mathbf{v}u,\mathbf{w}u} x_{\mathbf{w}} \leq
\sum_{\mathbf{v},\mathbf{w}\in V(H)^*} x_{\mathbf{v}}^* M_{\mathbf{v},\mathbf{w}} x_{\mathbf{w}} 
\eeq
for any fixed $u\in V(H)$ and finitely supported $(x_{\mathbf{v}})_{\mathbf{v}\in V(H)^*}$. To see this, choose any $e\in E(H)$ with $u\in e$ and write
\begin{align*}
\sum_{\mathbf{v},\mathbf{w}\in V(H)^*} x_{\mathbf{v}}^* \left( M_{\mathbf{v},\mathbf{w}} - M_{\mathbf{v}u,\mathbf{w}u} \right) x_{\mathbf{w}} & \stackrel{\eqref{npalin2}}{=} \sum_{\mathbf{v},\mathbf{w}\in V(H)^*} x_{\mathbf{v}}^* \left( \sum_{u'\in e,\: u'\neq u} M_{\mathbf{v}u',\mathbf{w}u'} \right) x_{\mathbf{w}} \\[4pt]
&\;\;= \sum_{u'\in e,\: u'\neq u} \; \sum_{\mathbf{v},\mathbf{w}\in V(H)^*} x_{\mathbf{v}}^* M_{\mathbf{v}u',\mathbf{w}u'} x_{\mathbf{w}} \geq 0,
\end{align*}
where the last inequality is due to positive semidefiniteness of $M$. This proves~(\ref{projwelldef}).

We now start the construction by taking the infinite-dimensional vector space spanned by all strings, $\H_0\defin\mathrm{lin}_{\C}\left( V(H)^* \right)$. The formula
$$
\left\langle \sum_{\mathbf{v}\in V(H)^*} x_{\mathbf{v}}\mathbf{v}, \sum_{\mathbf{w}\in V(H)^*} y_{\mathbf{w}}\mathbf{w} \right\rangle \defin \sum_{\mathbf{v},\mathbf{w}\in V(H)^*} x_{\mathbf{v}}^* M_{\mathbf{v},\mathbf{w}} y_{\mathbf{w}} .
$$
defines a positive semidefinite inner product on $\H_0$ in terms of the matrix $M$. The Cauchy--Schwarz inequality shows that
$$
\mathcal{N} \defin \left\{ \sum_{\mathbf{v}\in V(H)^*} x_{\mathbf{v}} \mathbf{v} \in\H_0 \:\Bigg|\: \left\langle \sum_{\mathbf{v}} x_{\mathbf{v}}\mathbf{v}, \sum_{\mathbf{v}} x_{\mathbf{v}}\mathbf{v} \right\rangle = 0 \right\}
$$
is a linear subspace of $\H_0$. The induced inner product on the quotient space $\H_0/\mathcal{N}$ is therefore positive definite by definition. We take the Hilbert space $\H$ to be the completion of $\H_0/\mathcal{N}$ with respect to the norm coming from this inner product.

Now for $u\in V(H)$, the operator $P_u$ is defined to act on $\H_0$ as 
$$
P_u\left(\sum_{\mathbf{v}\in V(H)^*} x_{\mathbf{v}} \mathbf{v} \right) \defin \sum_{\mathbf{v}\in V(H)^*} x_{\mathbf{v}} \mathbf{v}u .
$$
Thanks to~(\ref{projwelldef}), this maps $\mathcal{N}$ to itself, and therefore descends to a well-defined operator on $\B(\H)$, which we also denote by $P_u$. The equation $M_{\mathbf{v}u,\mathbf{w}}=M_{\mathbf{v},\mathbf{w}u}$ guarantees that $P_u$ is self-adjoint, while $M_{\mathbf{v}uu,\mathbf{w}}=M_{\mathbf{v}u,\mathbf{w}}$ shows that $P_u^2=P_u$ since
$$
\sum_{\mathbf{v}\in V(H)^*} x_{\mathbf{v}} \left(\mathbf{v}uu - \mathbf{v}u\right) \:\in\mathcal{N} ,
$$
which follows from~\eqref{npaip}. The equation $\sum_{u\in e} P_u = \mathbbm{1}_\H$ holds since
$$
\sum_{\mathbf{v}\in V(H)^*} x_{\mathbf{v}} \left( \mathbf{v} - \sum_{u\in e} \mathbf{v} u \right) \:\in\mathcal{N},
$$
thanks to~(\ref{npalin}). Finally, the rank-one density operator associated to the empty string $\emptyset\in\H$ is the desired quantum state, since
$$
\langle \emptyset, P_u \emptyset\rangle = M_{\emptyset,u} = p(u) .
$$
This ends our explicit description of the GNS construction.
\end{proof}

From this reasoning, we find that the sequence of sets $(\mathcal{Q}_n)_{n\in\N}$ converges in the following sense:

\begin{thm}
\label{convNPA}
For every contextuality scenario $H$,
$$
\mathcal{Q}(H) = \bigcap_{n\in \N} \mathcal{Q}_n(H) .
$$
\end{thm}

\begin{proof}[Proof~(\cite{NPA})]
Since we already know that $\mathcal{Q}(H)\subseteq\mathcal{Q}_n(H)$, it remains to be shown that if $p\in\mathcal{Q}_n(H)$ for all $n\in\N$, then $p\in\mathcal{Q}(H)$. To this end, we show that if a matrix $(M^n_{\mathbf{v},\mathbf{w}})_{\mathbf{v},\mathbf{w}\in V(H)^{*n}}$ exists with the required properties for every $n$, then there also exists a corresponding infinite matrix $(M^{\infty}_{\mathbf{v},\mathbf{w}})_{\mathbf{v},\mathbf{w}\in V(H)^*}$.

For $\mathbf{v}\in V(H)^{*n}$, positive semidefiniteness gives the estimate
$$
\left(M^{2n}_{\mathbf{v},\mathbf{v}}\right)^2 \stackrel{(\ref{npahankel})}{=} \left(M^{2n}_{\mathbf{v}\mathbf{v}^\dag,\emptyset}\right)^2 \leq M^{2n}_{\mathbf{v}\mathbf{v}^\dag,\emptyset} \cdot M^{2n}_{\emptyset,\emptyset} = M^{2n}_{\mathbf{v},\mathbf{v}} ,
$$
which implies $0\leq M^{2n}_{\mathbf{v},\mathbf{v}}\leq 1$, and hence
$$
|M^{2n}_{\mathbf{v},\mathbf{w}}|^2\stackrel{\eqref{bounddiag}}{\leq} M^{2n}_{\mathbf{v},\mathbf{v}} M^{2n}_{\mathbf{w},\mathbf{w}} \leq 1 .
$$
We therefore have $M^k_{\mathbf{v},\mathbf{w}}\in[-1,+1]$ for all $\mathbf{v},\mathbf{w}\in V(H)^{*n}$ with $n\leq 2k$.

Now consider the truncation of any $M^{2n}$ to a matrix indexed by $\mathbf{v},\mathbf{w}\in V(H)^{*n}$. Upon filling this truncation up with $0$'s, we obtain an infinite matrix $M'^{2n}$ indexed by $\mathbf{v},\mathbf{w}\in V(H)^*$ with all elements in $[-1,+1]$. In this way, every matrix $M'^{2n}$ becomes an element of $[-1,+1]^{V(H)^*\times V(H)^*}$. The space $[-1,+1]^{V(H)^*\times V(H)^*}$, equipped with the product topology, is second countable, and also compact thanks to Tychonoff's theorem. Hence, the sequence $(M'^n)_{n\in\N}$ has a convergent subsequence, and we write $M^\infty$ for its limit. By construction, this $M^\infty$ is an infinite matrix indexed by $\mathbf{v},\mathbf{w}\in V(H)^*$ having all the desired properties. The claim now follows from Proposition~\ref{GNS}.
\end{proof}

Since each $\mathcal{Q}_n(H)$ is defined in terms of a semidefinite program, we say that this represents a \emph{hierarchy of semidefinite programs} characterizing $\mathcal{Q}(H)$. It is a subfamily of the hierarchies of semidefinite programs in noncommutative optimization introduced in~\cite{NPA2}, which generalize the `commutative' hierarchies originally discovered in the context of convex optimization~\cite{Lasserre}.

\subsection{Equivalent characterizations of $\mathcal{Q}_1$ and the Lov{\'a}sz number}
\label{secQ1lov}

In this and the following subsection, we take a closer look at $\mathcal{Q}_1$, the first level of our semidefinite hierarchy, starting with a long list of equivalent characterizations:

\begin{prop}
\label{CSWeq}
For $p\in\mathcal{G}(H)$, the following are equivalent:
\begin{enumerate}
\item\label{pQ1} $p\in\mathcal{Q}_1(H)$;
\item\label{pM} There exist a Hilbert space $\mathcal{H}$, a unit vector $|\Psi\rangle \in \mathcal{H}$ and a vector $|\phi_v\rangle$ for every $v\in V(H)$ such that
\begin{enumerate}
\item \label{ortQ1} $u\perp v \quad \Longrightarrow \quad \langle \phi_u | \phi_v \rangle = 0$,
\item \label{sumQ1} $\sum_{v\in e} |\phi_v\rangle = |\Psi\rangle \quad \forall e\in E(H)$,
\item \label{repQ1} $p(v) = \langle \phi_v|\phi_v\rangle$;
\end{enumerate}
\item\label{pCSW} There exist a Hilbert space $\mathcal{H}$, a unit vector $|\Psi\rangle \in \mathcal{H}$ and a unit vector $|\psi_v\rangle$ for every $v\in V(H)$ such that
\begin{enumerate}
\item $u\perp v \quad \Longrightarrow \quad \langle \psi_u | \psi_v \rangle = 0$,
\item $p(v) = |\langle \psi_v|\Psi\rangle |^2$;
\end{enumerate}
\item\label{port} There exist a Hilbert space $\mathcal{H}$, a unit vector $|\Psi\rangle \in \mathcal{H}$ and a projection $P_v$ for every $v\in V(H)$ such that
\begin{enumerate}
\item $u\perp v \quad \Longrightarrow \quad P_u \perp P_v$,
\item $p(v) = \langle \Psi|P_v |\Psi \rangle \quad \forall v\in V(H)$;
\end{enumerate}
\item\label{psum} There exist a Hilbert space $\mathcal{H}$, a unit vector $|\Psi\rangle \in \mathcal{H}$ and a projection $P_v$ for every $v\in V(H)$ such that
\begin{enumerate}
\item $\sum_{v \in e} P_v \leq \mathbbm{1}_\mathcal{H} \quad \forall e \in E(H)$,
\item $p(v) = \langle \Psi|P_v |\Psi \rangle \quad \forall v\in V(H)$;
\end{enumerate}
\end{enumerate}
In all cases, $\H$ can also be taken to be the real Hilbert space $\R^{|V(H)|}$.
\end{prop}

In terms of the terminology of Appendix~\ref{appcap}, the vectors $|\psi_v\rangle$ in~\ref{pCSW} form an orthonormal labeling of the non-orthogonality graph $\mathrm{NO}(H)$. Characterizations~\ref{pCSW} and~\ref{psum} also show that our $\mathcal{Q}_1(H)$ coincides with the set `$\mathcal{E}_{\mathrm{QM}}^1$' considered in~\cite{CSW}.

\begin{proof}
\begin{itemize}
\item[\underline{\ref{pQ1}$\Rightarrow$\ref{pM}:}] The assumption is that there exists a positive semidefinite matrix $M$ with rows and columns indexed by $V(H)$ together with $\emptyset$ such that $M_{\emptyset,\emptyset}=1$, for any $e\in E(H)$ and any $v\in V(H)$ we have $\sum_{u\in e} M_{u,v} = M_{\emptyset,v}$ as well as $\sum_{u\in e} M_{u,\emptyset} = M_{\emptyset,\emptyset}$, and finally $M_{u,v}=0$ for $u\perp v$, such that $p(v) = M_{v,v}$.

By positive semidefiniteness, we can write this $M$ as a Gram matrix, meaning that there exist vectors $|\Psi\rangle,|\phi_v\rangle$ in $\H=\C^{|V(H)|}$ such that
$$
M_{\emptyset,\emptyset} = \langle\Psi|\Psi\rangle,\qquad M_{\emptyset,v} = \langle \Psi|\phi_v\rangle,\qquad M_{u,v} = \langle \phi_u|\phi_v\rangle,
$$
from which~\ref{ortQ1} and~\ref{repQ1} follow by the assumptions.

Now we fix $e\in E(H)$ and show~\ref{sumQ1}. We decompose $|\Psi\rangle$ into orthogonal components $|\Psi\rangle = |\Psi^{\parallel}\rangle + |\Psi^{\perp}\rangle$, where $|\Psi^{\parallel}\rangle \in \mathrm{lin}_{\C}\{|\phi_v\rangle \::\: v\in e\}$. Due to the orthogonality of the $\{|\phi_v\rangle\}_{v\in e}$, the equation
\[
\langle\phi_v|\phi_v\rangle = M_{v,v} = M_{v,\emptyset} = \langle\phi_v|\Psi\rangle 
\]
implies that $|\Psi^\parallel\rangle = \sum_{v\in e} |\phi_v\rangle$. Moreover, the computation
$$
\langle \Psi^{\parallel} | \Psi^{\parallel}\rangle + \langle \Psi^{\perp} | \Psi^{\perp} \rangle = M_{\emptyset,\emptyset} = \sum_{v\in e} M_{\emptyset,v} = \sum_{v,u\in e} M_{u,v} = \sum_{v,u\in e} \langle \phi_u |\phi_v \rangle = \langle \Psi^{\parallel} | \Psi^{\parallel} \rangle
$$
shows that $|\Psi^{\perp}\rangle = 0$, so that $\sum_{v\in e} |\phi_v\rangle = |\Psi\rangle$, as desired.

\item[\underline{\ref{pM}$\Rightarrow$\ref{pCSW}:}] 
Normalizing the $|\phi_v\rangle$ to $|\psi_v\rangle\defin\tfrac{1}{\sqrt{\langle \phi_v|\phi_v\rangle}}|\phi_v\rangle$ guarantees the orthogonality relations, and choosing some edge $e\in E(H)$ with $v\in e$ gives, upon plugging in $|\Psi\rangle = \sum_{u\in e}|\phi_u\rangle$,
$$
|\langle \psi_v|\Psi\rangle|^2 = \frac{1}{\langle \phi_v|\phi_v\rangle} \left|\left\langle \phi_v\Bigg| \sum_{u\in e} \phi_u\right\rangle\right|^2 = \frac{1}{\langle \phi_v|\phi_v\rangle} \langle \phi_v|\phi_v\rangle^2 = \langle \phi_v|\phi_v\rangle ,
$$
due to the orthogonality relations.
\item[\underline{\ref{pCSW}$\Rightarrow$\ref{port}:}] Define $P_v \defin |\psi_v\rangle\langle\psi_v|$.
\item[\underline{\ref{port}$\Rightarrow$\ref{psum}:}] This is clear since for fixed $e\in E(H)$, all projections $P_v$ for $v\in e$ are mutually orthogonal, which implies $\sum_{v\in e} P_v\leq \mathbbm{1}_\mathcal{H}$.
\item[\underline{\ref{psum}$\Rightarrow$\ref{pQ1}:}] Define $M_{v,w}\defin\langle\Psi|P_vP_w^\dagger|\Psi\rangle$, and similarly without one or both of the projections when $v$ or $w$ is replaced by $\emptyset$. We check that $M$ satisfies conditions~(\ref{npanorm}) to~(\ref{npaort}) and is positive semidefinite:
\begin{itemize}
\item[(\ref{npanorm})] $M_{\emptyset,\emptyset}=\langle\Psi|\Psi\rangle=1$, since $|\Psi\rangle$ is a unit vector. 
\item[(\ref{npalin})] Consider an edge $e\in E$. Since $p(v)$ is a probabilistic model,
$$
\langle\Psi|\Psi\rangle=1=\sum_{v\in e} p(v)=\langle \Psi | \sum_{v\in e} P_v |\Psi\rangle,
$$
which implies $\sum_{v\in e} P_v |\Psi\rangle = |\Psi\rangle$. Then, 
$$
\sum_{v\in e} M_{v,w}=\langle \Psi | \sum_{v\in e} P_v P_w | \Psi \rangle = \langle \Psi |P_w |\Psi \rangle = M_{\emptyset,w}.
$$
\item[(\ref{npaort})] If $v\perp w$, then there is an edge $e\in E(H)$ with $v,w\in e$. Hence, $P_v \perp P_w$, so that $M_{v,w}=\langle\Psi|P_vP_w|\Psi\rangle=0$. 
\end{itemize}
Positive semidefiniteness of $M$ can be shown as in the proof of Lemma~\ref{npan}. 
\end{itemize}
Finally, the proof of~\ref{pQ1}$\Rightarrow$\ref{pM} also shows that $\mathcal{H}$ can be taken to be the real Hilbert space $\R^{|V(H)|}$: this is what one gets upon starting with real $M$, applying the construction of the proof, and then restricting $\mathcal{H}$ to the real linear span of the $|\phi_v\rangle$'s. In all other implications, the Hilbert space does not change, and hence the same applies to all other characterizations.
\end{proof}

We can now relate the set $\mathcal{Q}_1$ to the Lov\'asz number $\vartheta$ of the non-orthogonality graph. This graph invariant is defined in Appendix~\ref{appcap}.

\begin{prop}
\label{Q1vsLov}
A probabilistic model $p\in\mathcal{G}(H)$ is in $\mathcal{Q}_1$ if and only if $\vartheta(\mathrm{NO}(H),p)=1$. 
\end{prop}

The corresponding result was already noticed in~\cite{CSW}, where the CSW approach based on the subnormalization of probability has been developed first (see also Section~\ref{CSWtransfer}).

\begin{proof}
We use the characterization of $\mathcal{Q}_1(H)$ given in Proposition~\ref{CSWeq}\ref{pCSW}. Assuming $p\in\mathcal{Q}_1(H)$, we choose corresponding vectors $|\psi_v\rangle,|\Psi\rangle\in\R^{|V(H)|}$; then, by Definition~\ref{wdefs},
$$
\vartheta(\mathrm{NO}(H),p) \leq \max_{v\in V} \frac{p(v)}{|\langle \Psi|\psi_v\rangle|^2} = \frac{p(v)}{p(v)} = 1 .
$$
On the other hand, the inequality $\vartheta(\mathrm{NO}(H),p)\geq 1$ follows from $\alpha(\mathrm{NO}(H),p)\geq 1$, which holds true because any $e\in E(H)$ defines an independent set in $\mathrm{NO}(H)$ and $\sum_{v\in e}p(v) = 1$.

Conversely, if $\vartheta(\mathrm{NO}(H),p)=1$, then there is an orthonormal labeling $(|\psi_v\rangle)_{v\in V(H)}$ and a vector $|\Psi\rangle\in\R^{|V(H)|}$ such that $|\langle\Psi|\psi_v\rangle|^2\geq p(v)$ for all $v$. By choosing $\H = \R^{|V(H)|} \oplus \R^{|V(H)|}$ and setting
$$
|\psi'_v\rangle \defin \frac{\sqrt{p(v)}}{|\langle\Psi|\psi_v\rangle|}|\psi_v\rangle \oplus \sqrt{1 - \frac{p(v)}{|\langle\Psi|\psi_v\rangle|^2}} |e_v\rangle \:\in \H
$$
where the $|e_v\rangle$ form the standard basis of $\mathbbm{R}^{|V(H)|}$, one obtains $|\langle\Psi|\psi'_v\rangle|^2 = p(v)$ with suitably orthogonal unit vectors $|\psi'_v\rangle$, as desired.
\end{proof}

This relation to graph theory has a simple first application:

\begin{prop}
\label{Q1tensor}
\begin{enumerate}
\item $\mathcal{Q}_1$ is closed under $\otimes$:
\beq
\label{eqQ1tensor}
\mathcal{Q}_1(H_A) \otimes \mathcal{Q}_1(H_B) \subseteq \mathcal{Q}_1(H_A\otimes H_B).
\eeq
\item $\mathcal{Q}_1(H)$ is convex.
\end{enumerate}
\end{prop}

\begin{proof}
\begin{enumerate}
\item Combine Proposition~\ref{Q1vsLov} with multiplicativity of $\vartheta$ (Proposition~\ref{lovmult}).
\item While this can be derived directly from the definition of $\mathcal{Q}_1$, it also follows from the subadditivity of $\vartheta$ in Lemma~\ref{thetaconv}.\qedhere
\end{enumerate}
\end{proof}

Again, the CHSH scenario $B_{2,2,2} = B_{1,2,2}\otimes B_{1,2,2}$ exemplifies that~(\ref{eqQ1tensor}) is not an equality in general, even after taking the convex hull on the left-hand side. The reason is that $\mathcal{Q}_1(B_{1,2,2})=\mathcal{C}(B_{1,2,2})$ by Example~\ref{Bell1CG}, but 
\[
\mathrm{conv}\left(\mathcal{Q}_1(B_{1,2,2}) \otimes \mathcal{Q}_1(B_{1,2,2})\right) = \mathrm{conv}\left(\mathcal{C}(B_{1,2,2}) \otimes \mathcal{C}(B_{1,2,2})\right) \stackrel{\eqref{Cproduct}}{=} \mathcal{C}(B_{2,2,2}) \subsetneq \mathcal{Q}_1(B_{2,2,2}).
\]

\subsection{$\mathcal{Q}_1$ on product scenarios}
\label{Q1products}

Naturally, there is the question of what one obtains when applying our hierarchy to Bell scenarios. Does it coincide with the original semidefinite hierarchy of~\cite{NPA}? Since the Bell scenario $B_{n,k,m}$ equals the product $B_{1,k,m}\otimes\ldots\otimes B_{1,k,m}$, we may as well ask the more general question: how can our hierarchy be analyzed on a product scenario? We will answer this question now for the case of $\mathcal{Q}_1$.

\begin{thm}
\label{AQthm}
Let $H_1\otimes\ldots\otimes H_n$ stand for an iterated binary product or for $^{\max}\otimes_{i=1}^n H_i$. Then a probabilistic model $p\in\mathcal{G}(H_1\otimes\ldots\otimes H_n)$ lies in $\mathcal{Q}_1(H_1\otimes\ldots\otimes H_n)$ if and only if there exist a Hilbert space $\H$, a state $|\Psi\rangle\in\H$ and an assignment of projections $E_k^v$ to every party $k=1,\ldots,n$ and vertex $v\in V(H_k)$ such that
\renewcommand{\labelenumi}{(\roman{enumi})}
\renewcommand{\theenumi}{(\roman{enumi})} 
\begin{enumerate}
\item\label{AQsubnorm} $\sum_{v\in e} E_k^v \leq \mathbbm{1}$ for all $k$ and $e\in E(H_k)$,
\item\label{AQperminv} $E_1^{v_1}\ldots E_n^{v_n}|\Psi\rangle = E_{\pi(1)}^{v_{\pi(1)}}\ldots E_{\pi(n)}^{v_{\pi(n)}}|\Psi\rangle$ for all permutations $\pi$ of the parties and sequences of vertices $v_1,\ldots,v_n$,
\item $p(v_1,\ldots,v_n) = \langle\Psi|E_1^{v_1}\ldots E_n^{v_n}|\Psi\rangle$.
\end{enumerate}
\renewcommand{\labelenumi}{(\alph{enumi})}
\renewcommand{\theenumi}{(\alph{enumi})}
\end{thm}

In the special case $n=1$, this recovers the characterization of Proposition~\ref{CSWeq}\ref{psum}. In general, it relates to the definition of `almost quantum' correlations in Bell scenarios~\cite{Mattyprep}, and we will see that these coincide indeed with our $\mathcal{Q}_1$.

\begin{proof}
We begin with the `if' direction and use Proposition~\ref{CSWeq}\ref{pCSW} as the relevant characterization of $\mathcal{Q}_1$. Upon writing $\vec{v}=(v_1,\ldots,v_n)$ for any vertex of the product $H_1\otimes\ldots\otimes H_n$, we define the vectors $|\psi_{\vec{v}}\rangle$ as
\beq
\label{defpsi}
|\psi_{\vec{v}}\rangle \defin \frac{E_1^{v_1}\ldots E_n^{v_n} |\Psi\rangle}{\sqrt{\langle\Psi| E_1^{v_1}\ldots E_n^{v_n}|\Psi\rangle}}.
\eeq
In order to check that this is normalized, we need to show that
\beq
\label{AQE}
\langle\Psi|E_n^{v_n}\ldots E_1^{v_1}E_1^{v_1}\ldots E_n^{v_n}|\Psi\rangle = \langle\Psi|E_1^{v_1}\ldots E_n^{v_n}|\Psi\rangle.
\eeq
This follows from the assumptions since we can repeatedly apply the computation
\begin{align*}
\phantom{=} \langle\Psi|E_n^{v_n}\ldots E_k^{v_k} E_1^{v_1}\ldots E_n^{v_n}|\Psi\rangle & = \langle\Psi|\underbrace{E_n^{v_n}\ldots E_k^{v_k} E_1^{v_1}\ldots E_{k-1}^{v_{k-1}}}_{\textrm{permutation of the $E_i^{v_i}$'s}} E_k^{v_k} \ldots E_n^{v_n}|\Psi\rangle \\
& \stackrel{\ref{AQperminv}}{=} \langle\Psi|\rlap{$\overbrace{\phantom{E_n^{v_n}\ldots E_{k-1}^{v_{k-1}} E_1^{v_1}\ldots E_k^{v_k}}}$}E_n^{v_n}\ldots E_{k-1}^{v_{k-1}} E_1^{v_1}\ldots \underbrace{E_k^{v_k} E_k^{v_k}}_{||}\ldots E_n^{v_n}|\Psi\rangle  \\[-3pt]
& \hspace{.43cm} = \langle\Psi|E_n^{v_n}\ldots E_{k-1}^{v_{k-1}} E_1^{v_1}\ldots \overbrace{E_k^{v_k}} \ldots E_n^{v_n}|\Psi\rangle \\[4pt]
& = \langle\Psi|E_n^{v_n}\ldots E_{k-1}^{v_{k-1}} E_1^{v_1}\ldots \ldots E_n^{v_n}|\Psi\rangle
\end{align*}
which works for any $k=1,\ldots,n$ and reduces this $k$ by one until one ends up with the desired expression, corresponding to $k=0$. Equation~\eqref{AQE} also shows that if the denominator in~\eqref{defpsi} vanishes, then so does the numerator. When this happens, we take $|\psi_{\vec{v}}\rangle$ to be any unit vector orthogonal to $|\Psi\rangle$ and all other $|\psi_{\vec{v}}\rangle$'s, which may require an enlargement of $\H$. It is then straightforward to see that the required properties also hold in this case.

We need to check that $\langle\psi_{\vec{u}}|\psi_{\vec{v}}\rangle=0$ for $\vec{u}\perp\vec{v}$. By the local orthogonality result of Proposition~\ref{NOnoLO}, the latter assumption means that $u_i\perp v_i$ for some party $i$. Assumption~\ref{AQsubnorm} then implies that $E_i^{u_i}\perp E_i^{v_i}$, and the permutation invariance~\ref{AQperminv} then again gives the conclusion,
\begin{align*}
\langle\psi_{\vec{u}}|\psi_{\vec{v}}\rangle & = \langle\Psi|E_n^{u_n}\ldots E_1^{u_1} E_1^{v_1}\ldots E_n^{v_n}|\Psi\rangle \\[4pt]
& = \langle\Psi|E_n^{u_n}\ldots\bcancel{E_i^{u_i}}\ldots E_1^{u_1} \underbrace{E_i^{u_i} E_i^{v_i}}_{=0} E_1^{v_1} \ldots \bcancel{E_i^{v_i}}\ldots E_n^{v_n}|\Psi\rangle = 0.
\end{align*}
Finally, we need to check that this data indeed recovers the given probabilities via the Born rule,
\[
|\langle\psi_{\vec{v}}|\Psi\rangle|^2 = \frac{|\langle\Psi|E_n^{v_n}\ldots E_1^{v_1}|\Psi\rangle|^2}{\langle\Psi| E_1^{v_1}\ldots E_n^{v_n}|\Psi\rangle} = \langle\Psi| E_1^{v_1}\ldots E_n^{v_n}|\Psi\rangle = p(v).
\]

Concerning the `only if' direction, we use the characterization of Proposition~\ref{CSWeq}\ref{pM} involving vectors $|\phi_{\vec{v}}\rangle$ associated to the vertices $\vec{v}$. We first note that for any party $j$ and any two $e,e'\in E(H_j)$, we have
\[
\sum_{v_j\in e} |\phi_{(v_1,\ldots,v_n)}\rangle = \sum_{v_j\in e'} |\phi_{(v_1,\ldots,v_n)}\rangle.
\]
This is a no-signaling-type equation which can be proven as in the proof of Proposition~\ref{multinosigprop}. (See also the concepts of Appendix~\ref{virtualcomp}.) This result shows that
\beq
\label{defnphi}
|\phi_k^{v_k}\rangle \defin \sum_{v_1,\ldots,\bcancel{v_k},\ldots,v_n} |\phi_{v_1,\ldots,v_n}\rangle,\qquad |\Psi\rangle \defin \sum_{v_1,\ldots,v_n} |\phi_{(v_1,\ldots,v_n)}\rangle,
\eeq
where each sum over $v_j$ for $j\neq k$ ranges over $v_j\in e_j$ for some $e_j\in E(H_j)$ does not depend on the particular choice of $e_j$. We now put 
\[
E_k^{v_k}\defin \frac{|\phi_k^{v_k}\rangle\langle\phi_k^{v_k}|}{\langle\phi_k^{v_k}|\phi_k^{v_k}\rangle}
\]
and claim that these projections have the required properties, in combination with the above $|\Psi\rangle$. First, $\sum_{v_k\in e} E_k^{v_k}\leq\mathbbm{1}$ follows from $E_k^{u_k}\perp E_k^{v_k}$ for $u_k\perp v_k$, which is a consequence of $\langle\phi_k^{u_k}|\phi_k^{v_k}\rangle = 0$. This in turn is a consequence of the definition~\eqref{defnphi}, since $u_k\perp v_k$ implies $\vec{u}\perp\vec{v}$, and therefore any summand $|\phi_{(u_1,\ldots,u_n)}\rangle$ in the sum for $|\phi_k^{u_k}\rangle$ is necessarily orthogonal to any $|\phi_{(v_1,\ldots,v_k)}\rangle$ in the sum for $|\phi_k^{v_k}\rangle$. 

For the permutation invariance condition~\ref{AQperminv}, we evaluate the action of some $E_k^{v_k}$ on $|\Psi\rangle$. To this end, we choose the $e_j$'s occurring in the sums to be the same in $|\phi_k^{v_k}\rangle$ as in $|\Psi\rangle$. With this, we obtain
\begin{align*}
E_k^{v_k} |\Psi\rangle & = \langle\phi_k^{v_k}|\phi_k^{v_k}\rangle^{-1} \sum_{v_1,\ldots,\bcancel{v_k},\ldots,v_n} \sum_{v'_1,\ldots,\bcancel{v'_k},\ldots,v'_n} \sum_{(u_1,\ldots,u_n)} |\phi_{(v_1,\ldots,v_n)}\rangle \langle\phi_{(v'_1,\ldots,v_k,\ldots,v'_n)}|\phi_{(u_1,\ldots,u_n)}\rangle \\[4pt]
 & = \langle\phi_k^{v_k}|\phi_k^{v_k}\rangle^{-1} \sum_{v_1,\ldots,\bcancel{v_k},\ldots,v_n} \sum_{v'_1,\ldots,\bcancel{v'_k},\ldots,v'_n} \sum_{(u_1,\ldots,u_n)} |\phi_{(v_1,\ldots,v_n)}\rangle \cdot p(v'_1,\ldots,v_k,\ldots,v'_n)\cdot \delta_{u_1,v'_1}\ldots\delta_{u_k,v_k}\ldots\delta_{u_n,v'_n} \\[4pt]
 & = \langle\phi_k^{v_k}|\phi_k^{v_k}\rangle^{-1} \sum_{v_1,\ldots,\bcancel{v_k},\ldots,v_n}  |\phi_{(v_1,\ldots,v_n)}\rangle \sum_{v'_1,\ldots,\bcancel{v'_k},\ldots,v'_n} p(v'_1,\ldots,v_k,\ldots,v'_n)\\[4pt]
 & = \langle\phi_k^{v_k}|\phi_k^{v_k}\rangle^{-1} \cdot p(v_k) \sum_{v_1,\ldots,\bcancel{v_k},\ldots,v_n} |\phi_{v_1,\ldots,v_n}\rangle = \sum_{v_1,\ldots,\bcancel{v_k},\ldots,v_n} |\phi_{v_1,\ldots,v_n}\rangle.
\end{align*}
One can apply the same reasoning to compute $E_j^{v_j} E_k^{v_k} |\Psi\rangle = \sum_{v_1,\ldots,\bcancel{v_j},\ldots,\bcancel{v_k},\ldots,v_n} |\phi_{v_1,\ldots,v_n}\rangle$, and so on. So eventually one will end up with
\beq
\label{AQEiterated}
E_{\pi(1)}^{v_{\pi(1)}}\ldots E_{\pi(n)}^{v_{\pi(n)}}|\Psi\rangle = \sum_{\bcancel{v_1},\ldots,\bcancel{v_n}} |\phi_{(v_1,\ldots,v_n)}\rangle = |\phi_{(v_1,\ldots,v_n)}\rangle
\eeq
for any permutation $\pi$. Since the right-hand side does not depend on $\pi$, the claim~\ref{AQperminv} follows. 

Finally, we have
\[
\langle\Psi|E_1^{v_1}\ldots E_n^{v_n}|\Psi\rangle = \langle\phi_{\vec{v}}|\phi_{\vec{v}}\rangle = p(v),
\]
again by~\eqref{AQEiterated} and the assumed orthogonality relations.
\end{proof}

For Bell scenarios $B_{n,k,m}$, this result can be strengthened as follows:

\begin{cor}
A no-signaling box $p(a_1\ldots a_n|x_1\ldots x_n)$ is in $\mathcal{Q}_1(B_{n,k,m})$ if and only if there exist a Hilbert space $\H$, a state $|\Psi\rangle\in\H$ and an assignment of projections $E_k^{a,x}$ to every party $k=1,\ldots,n$ and event $a|x$ such that
\renewcommand{\labelenumi}{(\roman{enumi})}
\renewcommand{\theenumi}{(\roman{enumi})} 
\begin{enumerate}
\item\label{BAQsubnorm} $\sum_a E_k^{a,x} = \mathbbm{1}$ for all $k$ and $x$,
\item\label{BAQperminv} $E_1^{a_1,x_1}\ldots E_n^{a_n,x_n}|\Psi\rangle = E_{\pi(1)}^{a_{\pi(1)},x_{\pi(1)}}\ldots E_{\pi(n)}^{a_{\pi(n)},x_{\pi(n)}}|\Psi\rangle$ for all permutations $\pi$ of the parties and sequences of events $a_1|x_1,\ldots,a_n|x_n$.
\item $p(a_1\ldots a_n|x_1\ldots x_n) = \langle\Psi|E_1^{a_1,x_n}\ldots E_n^{a_n,x_n}|\Psi\rangle$.
\end{enumerate}
\renewcommand{\labelenumi}{(\alph{enumi})}
\renewcommand{\theenumi}{(\alph{enumi})}
\end{cor}

This is the exact definition of the `almost quantum' set of no-signaling boxes from~\cite{Mattyprep}, which also coincides with the `$Q^{1+AB}$' set of~\cite{NPA}.

\begin{proof}
This is exactly the statement of the previous Theorem~\ref{AQthm} when specialized to Bell scenarios, except for the equality in~\ref{BAQsubnorm} which previously was an inequality. We can turn it into an equality by redefining
\[
E_k^{0,x} \defin \mathbbm{1} - \sum_{a\neq 0} E_k^{a,x}
\]
for all parties $k$. The permutation invariance~\ref{BAQperminv} then follows from the one for the original projections. Since the no-signaling box is determined by all these probabilities $p(a_1\ldots a_n|x_1\ldots a_n)$ in which $a_k\neq 0$ for all $k$~\cite{Pironio05}, the resulting new no-signaling box $p'$ obtained by putting $p'(a_1\ldots a_n|x_1\ldots x_n)\defin \langle\Psi|E_1^{a_1,x_n}\ldots E_n^{a_n,x_n}|\Psi\rangle$ in terms of the new projections must coincide with the original one.
\end{proof}

\newpage
\section{\textbf{Consistent Exclusivity and Local Orthogonality}}
\label{LOsec}

\subsection{Introducing Consistent Exclusivity}

It is a fundamental property of quantum theory that the compatibility of observables is a \emphalt{binary} relation: if a collection of quantum observables is such that they commute pairwise, then it follows that there is a basis in which all of them are diagonal, so that a measurement in that basis can be coarse-grained into a measurement of each observable. Paraphrasing Specker~\cite{Specker},
\begin{quotation}
A collection of propositions about a quantum mechanical system is precisely then
simultaneously decidable, when they are pairwise simultaneously decidable.
\end{quotation}

For us, this means the following: suppose that $I\subseteq V(H)$ is a set of vertices in a contextuality scenario $H$ such that every two of them belong to a common edge; by definition of $\mathrm{NO}(H)$, this means precisely that $I$ is an independent set in $\mathrm{NO}(H)$. Then the associated projections $(P_v)_{v\in I}$ for any quantum model $p\in\mathcal{Q}(H)$ have the property of being pairwise orthogonal, and hence $\sum_{v\in I} P_v\leq \mathbbm{1}_\H$. This implies
$$
\sum_{v\in I} p(v) = \sum_{v\in I} \mathrm{tr}( \rho P_v ) \leq 1.
$$
We now abstract from the quantum case to a general definition.
\begin{defn}[\cite{Henson}]
\label{defnCE}
A probabilistic model $p\in\mathcal{G}(H)$ satisfies \emph{Consistent Exclusivity} if
\beq
\label{CE}
\sum_{v\in I} p(v) \leq 1
\eeq
holds for any independent set $I\subseteq V(\mathrm{NO}(H))$. We write $\CE^1(H)\subseteq\mathcal{G}(H)$ for the set of probabilistic models satisfying Consistent Exclusivity.
\end{defn}

We also write CE$^1$ for this version of Consistent Exclusivity in order to distinguish it from the upcoming refinement termed CE$^\infty$. We refer to~\cite{CabelloSP} for an exposition of the history of principle and in which contexts it has been applied. 

Intuitively, CE$^1$ is saying that the total probability of any collection of pairwise exclusive outcomes is $\leq 1$. In this formulation, Consistent Exclusivity may almost sound like a trivial consequence of the laws of probability; however, this is not the case, since the probabilities $p(v)$ of a probabilistic model are \emphalt{conditional} probabilities representing the probability that outcome $v$ occurs \emphalt{given that} a measurement $e$ with $v\in e$ has been performed. 

The following result relates the $\CE^1$ set with quantum models and general probabilistic models.

\begin{prop}
\begin{enumerate}
\item $\mathcal{Q}(H)\subseteq\CE^1(H)$ for every $H$.
\item There exists a scenario $H$ with $\CE^1(H)\subsetneq \mathcal{G}(H)$.
\end{enumerate}
\end{prop}

In the graph-theoretic approach of~\cite{CSW}, analogous results have been obtained.

\begin{proof}
\begin{enumerate}
\item Above. 
\item For the triangle scenario $\Delta$ of Figure~\ref{triscen}, $V(\Delta)$ is itself an independent set in $\mathrm{NO}(\Delta)$. Since $\sum_{v\in V(\Delta)} p(v)=\tfrac{3}{2}$ for the unique probabilistic model $p$, this $p$ violates CE$^1$. We conclude that $\CE^1(\Delta)=\emptyset$, although $\mathcal{G}(\Delta)=\{p\}$.\qedhere
\end{enumerate}
\end{proof}

See~\cite{LSW} for further discussion of the triangle scenario and~\cite{LOfp,LO2} for examples in multipartite Bell scenarios.

In~\cite{CSW}, Consistent Exclusivity was imposed in the very \emphalt{definition} of probabilistic models. The problem with this is that the collection of models satisfying it is not closed under $\otimes$, as we will see in the following. Aside from the unclear physical meaning of CE, this is the main reason why we prefer our Definition~\ref{defnpm}: it guarantees that if $p_A$ and $p_B$ are probabilistic models on $H_A$ and $H_B$, respectively, then $p_A\otimes p_B$ is also a probabilistic model on $H_A\otimes H_B$; see Section~\ref{FRproducts}.

We now relate probabilistic models in $\CE^1(H)$ to the weighted independence number $\alpha$ of the non-orthogonality graph (see Definition~\ref{wdefs}).

\begin{prop}
\label{SPvsFIN}
A probabilistic model $p \in \mathcal{G}(H)$ belongs to $\CE^1(H)$ if and only if
\[
\alpha(\mathrm{NO}(H),p)\leq 1.
\]
\end{prop}

Again, due to the normalization equations $\sum_{v\in e} p(v)=1$, the statement $\alpha(\mathrm{NO}(H),p)\leq 1$ is actually equivalent to $\alpha(\mathrm{NO}(H),p)=1$.

\subsection{Consistent Exclusivity in Bell scenarios: Local Orthogonality} The concept of Local Orthogonality (LO) was recently introduced in~\cite{LOfp,LO2} as an information-theoretic principle satisfied by all quantum correlations in Bell scenarios, but violated by many non-quantum no-signaling boxes. The main reason for considering LO is the search for `physical' principles characterizing quantum correlations. It seems intuitively related to Consistent Exclusivity; here we would like to explain in which sense it is indeed a special case of CE when using our definition~(\ref{BSdefn}) of Bell scenario.

Recall~\cite{LOfp} that we call two events $u=a_1\ldots a_n|x_1\ldots x_n$ and $v=a'_1\ldots a'_n|x'_1\ldots x'_n$ in a Bell scenario \emph{locally orthogonal} if there is a party $i$ with $a_i\neq a'_i$, but $x_i=x'_i$. We now show that two events are locally orthogonal if and only if they are different vertices belonging to a common edge in the hypergraph $B_{n,k,m}$:

\begin{lem}
\label{LOtransfer}
The events $u,v\in V(B_{n,k,m})$ are locally orthogonal if and only if $u\perp v$.
\end{lem}

\begin{proof}
Suppose that $u=a_1\ldots a_n|x_1\ldots x_n$ and $v=a'_1\ldots a'_n|x'_1\ldots x'_n$ are locally orthogonal. By relabeling the parties, we can arrange for $a_1\neq a'_1$ and $x_1=x'_1$. Now choose any functions $f_2,\ldots,f_n$ with $f_i(a_1)=x_i$ and $f_i(a'_1)=x'_i$. Then the set of events of the form
$$
b_1\ldots b_n | x_1 f_2(b_1) \ldots f_n(b_1)
$$
defines an edge in $B_{n,k,m}$ containing both $u$ and $v$. Intuitively, Alice communicates her outcome to the other parties who then choose their measurement settings as a function of that outcome.

Conversely, $u\perp v$ means that there is an edge $e\in E(B_{n,k,m})$ with $u,v\in e$. More concretely, this states that there is an ordering of the parties $\sigma(1),\ldots,\sigma(n)$ and functions $f_{\sigma(i)}(b_{\sigma(1)},\ldots,b_{\sigma(i-1)})$ such that $e$ contains exactly those events which have the form
$$
b_{\sigma(1)}\ldots b_{\sigma(n)} | f_{\sigma(1)}() \ldots f_{\sigma(n)}(b_{\sigma(1)},\ldots,b_{\sigma(n-1)}) 
$$
where we have now written the parties in the order given by the permutation $\sigma$. Since both given events $u=a_1\ldots a_n|x_1\ldots x_n$ and $v=a'_1\ldots a'_n|x'_1\ldots x'_n$ are assumed to be of this form, we know that $x_{\sigma(i)}=f_{\sigma(i)}(a_{\sigma(1)},\ldots,a_{\sigma(i-1)})$ and $x'_{\sigma(i)}=f_{\sigma(i)}(a'_{\sigma(1)},\ldots,a'_{\sigma(i-1)})$. Now let $\sigma(j)$ be the smallest index with $a_{\sigma(j)}\neq a'_{\sigma(j)}$. Then, since $x_{\sigma(j)}$ and $x'_{\sigma(j)}$ only depend on $a_{\sigma(i)}$ and $a'_{\sigma(i)}$ with $i<j$, we conclude that $x_{\sigma(j)}=x'_{\sigma(j)}$, which proves the claim.
\end{proof}
Hence, when working within our framework for contextuality scenarios, the LO$^1$ principle studied in~\cite{LOfp} becomes a special case of CE$^1$ of Definition~\ref{defnCE}; the orthogonality between two events naturally arises from the FR product. Those readers not familiar with~\cite{LOfp} may regard this as the definition of LO$^1$. In~\cite{Cabello}, this relation between LO$^1$ and CE$^1$ was already implicitly used.

\begin{prob}
In~\cite{LOfp}, we have introduced LO$^1$ as a limitation for winning maximally difficult guessing problems using nonlocality as a resource. Since LO$^1$ coincides with $\CE^1(B_{n,k,m})$, it would be good to know whether this characterization of $\CE^1(B_{n,k,m})$ can be generalized to all contextuality scenarios.
\end{prob}

\begin{prob}
In~\cite{LOfp}, we also showed that LO$^1$ is equivalent to the no-signaling principle in bipartite Bell scenarios, i.e.~$\CE^1(B_{2,k,m})=\mathcal{G}(B_{2,k,m})$. More generally, under which conditions on $H$ does $\CE^1(H)=\mathcal{G}(H)$ hold?
\end{prob}

\subsection{Consistent Exclusivity and the Shannon capacity of graphs}

If $p\in\mathcal{G}(H)$ is a probabilistic model which is realizable in a world obeying certain physical laws, then it is reasonable to assume that any $p^{\otimes n}\in\mathcal{G}(H^{\otimes n})$ is realizable as well, since it simply corresponds to conducting $n$ copies of the same experiment in parallel. If we regard CE as delimiting the set of physically realizable probabilistic models, then this means that if $p^{\otimes n}\not\in \CE^1(H^{\otimes n})$, then we already know that $p$ itself is not physically realizable. This naturally gives a hierarchy of subsets of $\CE^1(H)$.

\begin{defn}[CE hierarchy of sets]
Let $H$ be a contextuality scenario and $p\in\mathcal{G}(H)$. We write $p\in\CE^n(H)$ if and only if $p^{\otimes n}\in\CE^1(H^{\otimes n})$. Furthermore,
$$
\CE^{\infty}(H) \defin \bigcap_{n\in\N} \CE^n(H) .
$$
\end{defn}

This is indeed relevant since, as we saw in~\cite{LOfp}, for example $\CE^2(B_{2,2,2})\neq\CE^1(B_{2,2,2})$. See~\cite{Cabello} for another example showing that violations of CE can be `activated' by considering copies $p^{\otimes n}$ of the same model $p$. If $p\in\CE^n(H)$, then we also say that $p$ satisfies CE$^n$. In particular, $p\in\CE^\infty(H)$ if and only if $p\in\CE^n(H)$ for all $n\in\N$, in which case we say that $p$ satisfies CE$^\infty$. In the special case of Bell scenarios, our previous results imply that CE$^\infty$ is precisely LO$^\infty$ of~\cite{LOfp}.

We now relate the $\CE^\ast$ family of sets to the weighted independence number $\alpha$ and Shannon capacity $\Theta$ (see Appendix~\ref{appcap} for definitions).

\begin{lem}
\label{LOchar} For a probabilistic model $p in \mathcal{G}(H)$,
\begin{enumerate}
\item $p\in\CE^n(H)$ if and only if 
\[
\alpha(\mathrm{NO}(H)^{\boxtimes n},p^{\otimes n}) \leq 1.
\]
\item $p\in\CE^\infty(H)$ if and only if
\[
\Theta(\mathrm{NO}(H),p)\leq 1,
\]
or, equivalently, if $\alpha(\mathrm{NO}(H),p)=\Theta(\mathrm{NO}(H),p)=1$.
\end{enumerate}
\end{lem}

\begin{proof}
\begin{enumerate}
\item By definition, $p\in\CE^n(H)$ if and only if $\alpha(\mathrm{NO}(H^{\otimes n}),p^{\otimes n})\leq 1$. The claim now follows from Lemma~\ref{ortproduct}.
\item The first statement holds by the definition of $\Theta$~(\ref{wcapdef}). For the second statement, $p\in\CE^\infty(H)$ implies that $\Theta(\mathrm{NO}(H),p)\leq 1$. But since $\alpha(\mathrm{NO}(H),p)=1$ due to $p\in\CE^1(H)$, we find $\Theta(\mathrm{NO}(H),p)=1=\alpha(\mathrm{NO}(H),p)$. The converse is clear.\qedhere
\end{enumerate}
\end{proof}

It follows from Corollary~\ref{Qtensor} that $\mathcal{Q}(H)\subseteq\CE^\infty(H)$.

\begin{lem} For every $k,n\in\N$, the following inclusions hold:
\label{LOincs}
$$
\CE^\infty(H) \subseteq \ldots \subseteq \ldots \CE^{n}(H) \subseteq \ldots \subseteq \CE^1(H) .
$$
\end{lem}
This should be seen in contrast to Remark~\ref{notmonotone}.
\begin{proof}
We choose any $p\in\CE^1(H)$. Thanks to Corollary~\ref{win-ineq}, we know that
$$
\alpha(\mathrm{NO}(H)^{\boxtimes n},p^{\otimes n}) \geq \alpha(\mathrm{NO}(H)^{\boxtimes (n-1)},p^{\otimes (n-1)}) \cdot \alpha(\mathrm{NO}(H),p) .
$$
Now since $\alpha(\mathrm{NO}(H),p)=1$, the sequence $\left( \alpha(\mathrm{NO}(H)^{\boxtimes n},p^{\otimes n}) \right)_{n\in\N}$ is monotonically nondecreasing. The claim now follows from Lemma~\ref{LOchar}.
\end{proof}

\subsection{Does Consistent Exclusivity characterize the quantum set?}

In~\cite{LOfp}, we considered $\CE^\infty(B_{n,k,m})$ for Bell scenarios $B_{n,k,m}$ and asked whether it coincides with $\mathcal{Q}(B_{n,k,m})$. We will answer this question now.

\begin{prop}[Navascu{\'e}s]
\label{miguelito}
For every $H$,
\beq
\label{migueleq}
\mathcal{Q}_1(H) \subseteq \CE^{\infty}(H) .
\eeq
\end{prop}

This observation was first made by Miguel Navascu{\'e}s (and proved in~\cite{Mattyprep}), before this whole formalism had been set up. Using our results on the relationships to invariants of graphs, we are now in a position to give an essentially trivial proof. See~\cite{Mattyprep} for a direct and almost as simple proof in the Bell scenario case.

\begin{proof}
Combine Propositions~\ref{Q1vsLov} and Lemma~\ref{LOchar} together with the fact that $\Theta(G,p) \leq \vartheta(G,p)$ for any weighted graph (see Corollary~\ref{monotone}).
\end{proof}

In particular, together with $\mathcal{Q}(H)\subseteq\mathcal{Q}_1(H)$, this gives another proof of $\mathcal{Q}(H)\subseteq\CE^\infty(H)$, even if an excessively more convoluted one. This completes our exposition of Figure~\ref{chain-inc}.

\begin{cor}
In the CHSH scenario $B_{2,2,2}$, the LO principle does not characterize quantum models: $\mathcal{Q}(B_{2,2,2})\subsetneq\CE^\infty(B_{2,2,2})$.
\end{cor}

\begin{proof}
From~\ref{miguelito}, since $\mathcal{Q}(B_{2,2,2})\subsetneq\mathcal{Q}_1(B_{2,2,2})$~\cite{NPA}.
\end{proof}

Hence, the Consistent Exclusivity principle can at best characterize $\mathcal{Q}_1$, the first level of the hierarchy of semidefinite programs. Alas, even this is not the case:

\begin{thm}\label{ejemplito}
There are contextuality scenarios $H$ for which $\mathcal{Q}_1(H)\subsetneq\CE^\infty(H)$.
\end{thm}

\begin{proof}
Our Proposition~\ref{Q1vsLov} and Lemma~\ref{LOchar} suggest that this is related to the existence of graphs $G$ for which $\alpha(G)=\Theta(G)<\vartheta(G)$. Indeed, we will turn Haemers' example~\cite{Haemers2} of this phenomenon into an example of a contextuality scenario $J_n$ with a probabilistic model $p_J\in\CE^\infty(J_n)$ with $p_J\not\in\mathcal{Q}_1(J_n)$.

Let $n\geq 12$ be an integer divisible by $4$. Let $J_n$ have vertices $V(J_n)$ being all $3$-element subsets of $\{1,\ldots,n\}$. Following~\cite{Haemers2}, an edge of $J_n$ is given in terms of a partition of $\{1,\ldots,n\}$ into $4$-element subsets; a vertex ($3$-element subset) belongs to the edge if and only if it is contained in one of the subsets of the partition. We call this scenario $J_n$ due to the relation to Johnson schemes~\cite{Haemers2}.

By construction, every edge $e\in E(J_n)$ has cardinality $|e|=n$, since every partition consists of $n/4$ subsets and each subset hosts $4$ vertices. Therefore, assigning a weight of $\tfrac{1}{n}$ to each vertex defines a probabilistic model $p_J$. Now the non-orthogonality graph $\mathrm{NO}(J_n)$ consists of the $3$-element subsets of $\{1,\ldots,n\}$ two of which are adjacent if and only if they have exactly one element in common. This is the graph that was considered by Haemers~\cite{Haemers2}, who showed that
$$
\alpha(\mathrm{NO}(J_n))=\Theta(\mathrm{NO}(J_n))=n < \vartheta(\mathrm{NO}(J_n)) .
$$
Since the probabilistic model $p_J$ has constant weights $\tfrac{1}{n}$, this means that
$$
\alpha(\mathrm{NO}(J_n),p_J) = \Theta(\mathrm{NO}(J_n),p_J) = 1 < \vartheta(\mathrm{NO}(J_n),p_J) ,
$$
and hence $p_J\in\CE^\infty(J_n)$, but $p_J\not\in\mathcal{Q}_1(J_n)$.
\end{proof}

In fact, we can easily turn this proof into a stronger result:

\begin{cor}
For the contextuality scenario $J_n$ from the previous proof, we have $\mathcal{Q}_1(J_n) = \emptyset$, although $\CE^\infty(J_n)\neq\emptyset$.
\end{cor}

\begin{proof}
The scenario $J_n$ is vertex-transitive: for any two vertices $v,w\in  V(J_n)$, there exists a symmetry transformation which takes $v$ into $w$ given by simply permuting the elements of the ground set $\{1,\ldots,n\}$. Starting with any probabilistic model $p\in\mathcal{G}(J_n)$, we can obtain the above $p_J$ by taking a convex combination of $p$ and all its images obtained by applying symmetry transformations $\pi\in S_n$, where $S_n$ is the permutation group of $\{1,\ldots,n\}$,
\[
p_J = \frac{1}{n!}\sum_{\pi\in S_n} \pi(p).
\]
Therefore, if there existed a model $p\in\mathcal{Q}_1(J_n)$, then we would obtain $p_J\in\mathcal{Q}_1(J_n)$ by invariance under symmetries and convexity of $\mathcal{Q}_1(J_n)$, but this we already know to be false.

This shows that $\mathcal{Q}_1(J_n)=\emptyset$. Since $p_J\in\CE^\infty(J_n)$, we also already know that $\CE^\infty(J_n)$ is not empty.
\end{proof}

\begin{rem}
\label{sizerem}
For instance, for $n=12$ which gives the smallest example, $J_{12}$ is a scenario with $\binom{12}{3}=220$ many vertices and $\tfrac{1}{3!}\binom{12}{4\:4\:4}=5775$ many edges. Two vertices $v,w\in V(J_{12})$ are adjacent in $\mathrm{NO}(J_{12})$ if and only if $|v\cap w| = 1$; and for any given $v$ there are $3\cdot\binom{9}{2}=108$ different $w$'s satisfying this condition. Hence the graph $\mathrm{NO}(J_{12})$ has $\tfrac{1}{2}\cdot 220\cdot 108 = \numprint{11880}$ many edges.
\end{rem}

\subsection{Activation and non-convexity of Consistent Exclusivity}\label{se:act}

In this section, we address the problem of whether violations of Consistent Exclusivity can be obtained by \emph{activation}: are there contextuality scenarios $H_A$ and $H_B$ together with probabilistic models $p_A$ and $p_B$ such that $p_A\otimes p_B\not\in\CE^\infty(H_A\otimes H_B)$, although $p_A\in\CE^\infty(H_A)$ and $p_B\in\CE^\infty(H_B)$? Or is Consistent Exclusivity closed under taking tensor products?

What we will find is that such activation is indeed possible. 

We start the construction by taking any contextuality scenario $H_A$ which has a probabilistic model with $p_A\in \CE^\infty(H_A)$, but $p_A\not\in\mathcal{Q}_1(H_A)$; the proof of Theorem~\ref{ejemplito} provides a concrete example, but any other one will do just as fine. From Lemma~\ref{LOchar} and Proposition~\ref{Q1vsLov}, we obtain that 
\[
\Theta\left(\mathrm{NO}\left(H_A\right),p_A\right) = 1,\qquad \vartheta\left(\mathrm{NO}\left(H_A\right),p_A \right) > 1,
\]
noting that if we use the example of the proof of Theorem~\ref{ejemplito}, then these properties were really what enabled us to show that $p_A\in\CE^\infty(H_A)\setminus\mathcal{Q}_1(H_A)$ in the first place.

Then by Proposition~\ref{thetaalt}, we know that there exists an orthonormal labeling $v\mapsto |\phi_v\rangle$ of the complementary graph $\overline{\mathrm{NO}(H_A)}$ and another unit vector $|\Psi\rangle \in \mathbbm{R}^{|V(H_A)|}$ such that
\beq
\label{yanidea}
\sum_{v \in V_A} p_A(v) \: |\langle \Psi | \phi_v\rangle|^2 >1.
\eeq
Following an idea of Yan~\cite{Yan}, we will turn the inner products $|\langle\Psi|\phi_v\rangle|^2$ into the probabilities of a quantum model on a certain scenario $H_B$ in such a way that this precise inequality witnesses a violation of Consistent Exclusivity.

To define this scenario $H_B$, we start with the non-orthogonality graph $\mathrm{NO}(H_A)$ and apply a construction which we will meet again in Section~\ref{CSWtransfer}: we would like each edge of $\mathrm{NO}(H_A)$ to represent a \emph{subnormalized} measurement. This means that for each edge the vertices of $H_B$ are the vertices of $H_A$ together with one additional `no-detection event' for each edge of $\mathrm{NO}(H_A)$,
\[
V(H_B) \defin V(H_A) \cup E(\mathrm{NO}(H_A)),
\]
where the no-detection event for edge $e\in E(\mathrm{NO}(H_A))$ is denoted by $w_e$, and its r{\^o}le is to turn the subnormalized edges into normalized measurements. So for every edge $e=\{u,v\}\in E(\mathrm{NO}(H_A))$, there is a measurement given by
\[
\{ u,v,w_e \} \subseteq V(H_B),
\]
and these sets constitute the set of new edges $E(H_B)$.

\begin{lem}
\label{yanlemma}
The assignment
\[
p_B (v):= \begin{cases}
 \qquad\qquad|\langle \Psi | \phi_v\rangle|^2 & \textrm{if } v \in V(H_A),\\
 1 - |\langle \Psi | \phi_u\rangle|^2 -|\langle \Psi | \phi_{u'}\rangle|^2 & \textrm{if }  v = w_e \textrm{ for } e=\{u,u'\}\in E(\mathrm{NO}(H_A)),
\end{cases}
\]
defines a quantum model on $H_B$.
\end{lem}

\begin{proof}
$p_B$ is represented by the family of projections
\[
P_v = \begin{cases}
\qquad\qquad|\phi_v \rangle \langle \phi_v| & \textrm{if } \, v \in V_A,\\
\mathbbm{1} - |\phi_u \rangle \langle \phi_u| - |\phi_{u'} \rangle \langle \phi_{u'}| & \textrm{if } v = w_e \textrm{ for } e=\{u,u'\}\in E(\mathrm{NO}(H_A))
\end{cases}
\]
on the Hilbert space $\C^{|V(H_A)|}$ together with the state $|\Psi\rangle\in\C^{|V(H_A)|}$. That an operator of the second kind, $\mathbbm{1}_\mathcal{H} - |\phi_u \rangle \langle \phi_u| - |\phi_{u'} \rangle \langle \phi_{u'}|$, is indeed a projection follows from the orthogonality relation $\langle\phi_u|\phi_{u'}\rangle = 0$, which is guaranteed by the assumption that the family $(|\phi_u\rangle)_{u\in V(H_A)}$ is an orthonormal labeling of $\overline{\mathrm{NO}(H_A)}$. The normalization condition $\sum_{v\in e}P_v = \mathbbm{1}$ holds for any $e\in E(H_B)$ by definition.
\end{proof}

To summarize, we have probabilistic models $p_A\in\CE^\infty(H_A)$ and $p_B\in \mathcal{Q}(H_B)$, so that in particular $p_B\in\CE^\infty(H_B)$. We now consider the probabilistic model $p_A \otimes p_B$ on $H_A \otimes H_B$:

\begin{lem}
$p_A \otimes p_B \notin  \CE^1 \left(H_A \otimes H_B \right)$.
\end{lem}

\begin{proof}(Yan~\cite{Yan})
For any two vertices $u,v\in V_A$, we claim that $(u,u)$ and $(v,v)$ are orthogonal as vertices in $H_A\otimes H_B$. By Proposition~\ref{NOnoLO}, this is clear if $u\perp v$ in $H_A$; otherwise, we have $u\sim v$ in $\mathrm{NO}(H_A)$, and therefore $u\perp v$ in $H_B$ by definition of $H_B$, which also implies the claim by Proposition~\ref{NOnoLO}.

In particular, the diagonal\footnote{Yan's idea of looking at this diagonal is not new in the context of the Lov{\'a}sz number. In fact, it is already contained in Lov{\'a}sz's original paper on the subject~\cite{Lovasz}.} $D := \left\{ (v, v) \: | \: v \in V_A \right\}$ forms an independent set in $\mathrm{NO}(H_A \otimes H_B)$. Therefore, a necessary condition for $p_A \otimes p_B$ to belong to $\CE^1(H_A \otimes H_B)$ is that
$$
\sum_{v \in V(H_A)} (p_A \otimes p_B)(v,v) \stackrel{!}{\leq} 1.
$$
However, evaluating the left-hand side results in~\eqref{yanidea},
$$
\sum_{v \in V(H_A)} (p_A\otimes p_B)(v,v) =  \sum_{v\in V(H_A)} p_A(v) p_B(v) = \sum_{v \in V_A} p_A(v) \: |\langle \Psi | \phi_v\rangle|^2 >1,
$$
which completes the proof.
\end{proof}

What we have thereby shown in particular is that violations of Consistent Exclusivity can be activated. In other words,

\begin{thm}
\label{CEactivate}
There are contextuality scenarios $H_A$ and $H_B$ for which
\[
\CE^\infty(H_A)\otimes\CE^\infty(H_B)\not\subseteq\CE^1(H_A\otimes H_B).
\]
\end{thm}

In fact, we have seen that we can even put $\mathcal{Q}(H_B)$ in place of $\CE^\infty(H_B)$ on the left-hand side and the statement remains valid.

The proof of this result was relatively abstract in the sense that we have not exhibited a concrete example. We now explain how to do this in terms of the scenario $J_{12}$ from the proof of Theorem~\ref{ejemplito} and Remark~\ref{sizerem} equipped with the probabilistic model $p_A:=p_J$, which assigns a uniform weight of $\tfrac{1}{12}$ to each vertex. The reader not interested in such an explicit construction may move on to Theorem~\ref{nonconvexthm}.

The most difficult step consists in finding a suitable  orthonormal labeling of $\overline{\mathrm{NO}(J_{12})}$, i.e.~an assignment $v \mapsto |\phi_{v}\rangle$ of a unit vector $|\phi_{v}\rangle \in \R^{220}$ to every triplet $v \in V(J_{12})$ such that $|v \cap w|=1$ implies that $|\phi_{v}\rangle \perp |\phi_{w}\rangle$. Let us denote by $|v\rangle$ the elements of the canonical basis of $\mathbbm{R}^{220}$. We will construct the $|\phi_{v}\rangle$ in terms of this basis.

To each vertex $v$, we associate the sets of vertices $D_i(v)$ for $i \in \{0, 1, 2, 3\}$,
\[
D_i(v) = \left\{ w\in V(J_{12}) \:\big|\: |v\cap w| = 3 - i \right\}.
\]
The subscript $i$ indicates in how many elements a $w\in D_i(v)$ \emphalt{differs} from $v$. We have 
\[
|D_0(v)| = 1,\qquad |D_1(v)| = 3\cdot 9 = 27,\qquad |D_2(v)|= 3 \cdot \tbinom{9}{2} = 108,\qquad |D_3(v)| = \tbinom{9}{3}=84
\]
for any $v$. For the vectors $|\phi_{v}\rangle$, we make the ansatz
\[
|\phi_{v}\rangle \defin \sum_{i=0}^3 \frac{\alpha_i}{\sqrt{|D_i(v)|}} \sum_{v'\in D_i(v)} |v'\rangle 
\]
for $\alpha_i\in\R$ and the denominators have been chosen such that the normalization condition for this vector simply reads 
 \beq
 \label{constraint-norm}
 \alpha_0^2 + \alpha_1^2 + \alpha_2^2 + \alpha_3^2=1.
 \eeq
We need to ensure that $\langle\phi_v|\phi_w\rangle = 0$ for $|v\cap w|=1$. With our ansatz for the vectors, this means that
\[
\sum_{i=0}^3 \alpha_i^2 \cdot  \frac{|D_i(v) \cap D_i(w)|}{|D_i(v)|} + 2 \sum_{i < j} \alpha_i \alpha_j \cdot \frac{|D_i(v) \cap D_j(w)|}{\sqrt{|D_i(v)\cdot D_j(w)|}} =0.
\]
It is clear that $|D_0(v)\cap D_i(w)|$ is $1$ for $i=2$ and $0$ otherwise; for the cardinalities of the other intersections, see Figure~\ref{intersects}.

\begin{figure}
\subfigure[$|D_1(v) \cap D_1(w)| = 4$.]{
\hspace{1cm}
$\begin{array}{c}
\{1,3,4\}\\
\{1,3,5\}\\
\{2,3,4\}\\
\{2,3,5\}\\\\
\end{array}$
\hspace{1cm}}
\hspace{1cm}
\subfigure[$|D_1(v) \cap D_2(w)| = 16$.]{
\hspace{1cm}
$\begin{array}{c}
\{1,2,4\}\\
\{1,2,5\}\\
\{1,3,x\}\\
\{2,3,x\}\\\\
\end{array}$
\hspace{1cm}}
\hspace{1cm}
\subfigure[$|D_1(v) \cap D_3(w)| = 7$.]{
\hspace{1cm}
$\begin{array}{c}
\\\\\\\{1,2,x\}\\\\
\end{array}$
\hspace{1cm}}
\subfigure[$|D_2(v) \cap D_2(w)| = 49$.]{
\hspace{1cm}
$\begin{array}{c}
\{1,4,x\}\\
\{1,5,x\}\\
\{2,4,x\}\\
\{2,5,x\}\\
\{3,x,y\}\\\\
\end{array}$
\hspace{1cm}}
\hspace{1cm}
\subfigure[$|D_2(v) \cap D_3(w)| = 42$.]{
\hspace{1cm}
$\begin{array}{c}
\\\\\\\{1,x,y\}\\
\{2,x,y\}\\\\
\end{array}$
\hspace{1cm}}
\hspace{1cm}
\subfigure[$|D_3(v) \cap D_3(w)| = 35$.]{
\hspace{1cm}
$\begin{array}{c}
\\\\\\\\\{x,y,z\}\\\\
\end{array}$
\hspace{1cm}}
\caption{The various intersections for $v=\{1,2,3\}$ and $w=\{3,4,5\}$. Here, $x$, $y$ and $z$ stand for arbitrary elements of $\{5,\ldots,12\}$, so that entries containing one, two or all three of these have to be counted with multiplicity $7$, $\binom{7}{2}$ or $\binom{7}{3}$, respectively.}
\label{intersects}
\end{figure}

In terms of the explicit numbers and upon reducing fractions, this equation can be written in matrix form as
\beq
\label{matrixeq}
\left(\begin{matrix}\alpha_0\\\alpha_1\\\alpha_2\\\alpha_3\end{matrix}\right)^T
\left[
\begin{matrix}
0 & 0 & \frac{1}{6\sqrt{3}} & 0 \\
0 & \frac{4}{27} & \frac{8}{27} & \frac{\sqrt{7}}{18} \\
\frac{1}{6\sqrt{3}} & \frac{8}{27} & \frac{49}{108} & \frac{\sqrt{7}}{6} \\
0 & \frac{\sqrt{7}}{18} & \frac{\sqrt{7}}{6}  & \frac{5}{12} 
\end{matrix}
\right]
\left(\begin{matrix}\alpha_0\\\alpha_1\\\alpha_2\\\alpha_3\end{matrix}\right)
= 0 .
\eeq
This matrix has two normalized eigenvectors given by
\[
|\vec{b} \rangle = \frac{1}{2 \sqrt{55}} \left( 
\begin{matrix}
1 \\ 3 \sqrt{3} \\ 6 \sqrt{3} \\ 2\sqrt{21}
\end{matrix} 
\right) = \frac{1}{2\sqrt{55}} \left(
\begin{matrix}
\sqrt{|D_0|} \\ \sqrt{|D_1|} \\ \sqrt{|D_2|} \\ \sqrt{|D_3|}
\end{matrix} \right) ,
\qquad
|\vec{c} \rangle = \frac{1}{2 \sqrt{30}} \left( 
\begin{matrix}
-2 \sqrt{21} \\ 2 \sqrt{7} \\ - \sqrt{7} \\ 1
\end{matrix} 
\right),
\]
with eigenvalues $1$ and $-\tfrac{13}{108}$, respectively. Therefore with
\[
|\vec{\alpha}\rangle = \left(\begin{matrix}\alpha_0\\\alpha_1\\\alpha_2\\\alpha_3\end{matrix}\right) := \frac{\sqrt{13}}{11}|\vec{b}\rangle + \frac{6\sqrt{3}}{11}|\vec{c}\rangle,
\]
both the normalization constraint~\eqref{constraint-norm} and orthogonality~\eqref{matrixeq} are satisfied. Using these values for the $\alpha_i$'s therefore defines an orthonormal labeling of $\overline{\mathrm{NO}(J_{12})}$. We now need to find a unit vector $|\Psi\rangle$ such that $\sum_v \left|\langle\Psi|\phi_v\rangle\right|^2>12$; and indeed, with $|\Psi\rangle\defin \frac{1}{\sqrt{220}} \sum_{v} |v\rangle$ we obtain
\[
\sum_{v} \left| \langle \Psi|\phi_{v}\rangle\right|^2 = 220 \left(\sum_{i=0}^3 \sqrt{\frac{|D_i|}{220}} \, \alpha_i\right)^2 = 220\, |\langle \vec{b}|\vec{\alpha}\rangle|^2 = \frac{260}{11}.
\]
This coincides with the Lov\'asz number of $\mathrm{NO}(J_{12})$~\cite[p.~46]{spectra}\footnote{Though it differs from the formula in~\cite[p.~271]{Haemers2}, which would give a Lov\'asz number of $\approx 42$. This formula contains a typo: the `$1$' in the numerator should be an `$11$'.} and is therefore the maximally possible value. For each individual $v$, we have $\left| \langle \Psi|\phi_{v}\rangle\right|^2 = \tfrac{13}{121}$. This concludes our presentation of the scenario $H_A = J_{12}$, and we now turn to $H_B$.

The scenario $H_B$ has two kinds of vertices: first, again the $3$-element subsets of $\{1,\ldots,12\}$, of which there are $220$; second, (unordered) pairs of $3$-element subsets of $\{1,\ldots,12\}$ having one element in common, of which there are $\numprint{11880}$. In total, there are $220+\numprint{11880}=\numprint{12100}$ vertices. The second kind of vertices also define the edges of $H_B$: an edge consists of such a vertex, i.e.~an unordered pair of $3$-element sets, together with the two vertices defined by these $3$-element sets. The probabilistic model $p_B$ assigns a probability of $\tfrac{13}{121}$ to each vertex of the first kind and $\tfrac{95}{121}$ to each vertex of the second kind.

Then as the above considerations show, we have $p_A\in\CE^\infty(H_A)$ and $p_B\in\mathcal{Q}(H_B)$, but $p_A\otimes p_B\not\in\CE^\infty(H_A\otimes H_B)$. This ends our explicit description of our example for Theorem~\ref{CEactivate}. 

Another question---seemingly unrelated---is whether $\CE^\infty(H)$ is convex for every scenario $H$. We now use the results of the previous subsection to show that this is also not always the case. The scenarios $H_A$ and $H_B$ and probabilistic models $p_A$ and $p_B$ are the same as before, or more generally as in Theorem~\ref{CEactivate}.

\begin{thm}
\label{nonconvexthm}
There are contextuality scenarios for which $\CE^\infty(H)$ is not convex.
\end{thm}

\begin{proof}
Define the contextuality scenario $H$ to be the disjoint union of $H_A$ and $H_B$ in the sense that $V = V_A \cup V_B$ and $e\subseteq V$ is an edge if there exist $e_A \in E_A$ and $e_B \in E_B$ such that $e = e_A \cup e_B$. Since every vertex of either graph is contained in at least one edge, the corresponding non-orthogonality graphs decompose as $\mathrm{NO}(H) = \mathrm{NO}(H_A) + \mathrm{NO}(H_B)$, where `$+$' stands for the disjoint union of graphs as in Appendix~\ref{appcap}.

With this definition, the probabilistic models $p_A$ and $p_B$ can easily be extended to $H$,
\[
p_{A}'(v) := \begin{cases}
p_A(v) & \textrm{if } v \in V_A,\\
0 & \textrm{if } v \in V_B,\\
\end{cases}
\quad \textrm{ and } \quad
p_{B}'(v) := \begin{cases}
0 & \textrm{if } v \in V_A,\\
p_B(v) & \textrm{if } v \in V_B.\\
\end{cases}
\]
The assumptions $p_A\in\CE^\infty(H_A)$ and $p_B\in\CE^\infty(H_B)$ imply that
\[
p_A' \in \CE^\infty(H), \qquad p_B' \in \CE^\infty(H)
\]
while the assumption $p_A\otimes p_B\not\in\CE^\infty(H_A\otimes H_B)$ means that
\beq
\label{productnotCE}
\quad p_A' \otimes p_B' \notin \CE^\infty(H \otimes H).
\eeq
In particular, from the characterization of $\CE^\infty$ given
by Lemma \ref{LOchar}, it follows that
\[
\Theta(\mathrm{NO}(H\otimes H, p_A'\otimes p_B') >1.
\]
We finally define the probabilistic model $p$ on $H$ obtained as a convex mixture of $p_A'$ and $p_B'$:
$$
p := \frac{1}{2}(p_A' + p_B').
$$
We now proceed to show that $p^{\otimes 2}  \notin \CE^\infty\left(H^{\otimes 2}\right)$, which implies that $ p \notin  \CE^\infty\left(H\right)$ and hence that the set $\CE^\infty(H)$ is not convex.
The probabilistic model $p^{\otimes 2}$  can be written as a convex combination,
\beq
\label{p2cc}
p^{\otimes 2} = \frac{1}{4}  \left( p_A'^{\otimes 2} + p_B'^{\otimes 2} + p_A' \otimes p_B' + p_B' \otimes p_A' \right).
\eeq
As a vertex weighing on $\mathrm{NO}(H^{\otimes 2})$, the four summands of this convex combination are weight functions supported on the four disjoint subgraphs
\[
\mathrm{NO}(H_A\otimes H_A),\quad \mathrm{NO}(H_B\otimes H_B),\quad \mathrm{NO}(H_A\otimes H_B),\quad \mathrm{NO}(H_B\otimes H_A),
\]
in this order. Furthermore, there are no edges between these four subgraphs, so that
\[
\mathrm{NO}(H^{\otimes 2}) = \mathrm{NO}(H_A\otimes H_A) + \mathrm{NO}(H_B\otimes H_B) + \mathrm{NO}(H_A\otimes H_B) + \mathrm{NO}(H_B\otimes H_A).
\]
In order to lower bound the Shannon capacity of~\eqref{p2cc}, we can therefore apply Lemma~\ref{proptheta}, which gives
\[
\Theta\left(\mathrm{NO}(H^{\otimes 2}), p^{\otimes 2}\right) \geq \frac{1}{4} + \frac{1}{4} + \frac{1}{2}\Theta(\mathrm{NO}(H_A\otimes H_B), p_A'\otimes p_B'),
\]
where the first two terms correspond to $\Theta(\mathrm{NO}(H_A^{\otimes 2}),p_A^{\otimes 2})=\Theta(\mathrm{NO}(H_B^{\otimes 2}),p_B^{\otimes 2})=1$ and the last two terms coincide and have been summed up. Our assumption~\eqref{productnotCE} together with the Lemma~\ref{LOchar}, the characterization of $\CE^\infty$ in terms of $\Theta$, we obtain $\Theta(\mathrm{NO}(H^{\otimes 2}),p^{\otimes 2})>1$, from which $p^{\otimes 2}\not\in\CE^\infty(H)$ immediately follows.
\end{proof}

Given the previous explicit construction, it is now very easy to write down an explicit example of this phenomenon. The resulting scenario turns out to have $220+\numprint{12100} = \numprint{12320}$ vertices.

\subsection{\textit{Extended} Consistent Exclusivity
principle}\label{se:ECE}

In the previous subsection, we showed that the set of probabilistic models $\CE^\infty$ is neither convex nor closed under $\otimes$. However, it is natural to believe that the collection of physically realizable probabilistic models should be both convex and closed under $\otimes$. Therefore, if some physically realistic $q\in\CE^\infty(H)$ can be combined with some $p\in\CE^\infty(H)$ by using convex combinations and $\otimes$-products such that the combination is not in $\CE^\infty$, then $p$ itself should be considered to violate the CE principle in a certain extended form. In this section, we propose one way of extending the CE principle, such that the set of probabilistic models that satisfies it is convex and closed under $\otimes$. This extension was somehow already implicit in the work of Yan~\cite{Yan}, who showed that the maximum violation of a noncontextuality inequality given by models that satisfy Extended Consistent Exclusivity (ECE) is the same as the maximum `quantum' violation in the CSW formalism~\cite{CSW}; see also the independent work~\cite{ECE}, where this has been made explicit in a way similar to here.

The natural choice for the `physically realistic' models $q$ is to assume them to be the quantum models, so that we arrive at:

\begin{defn}
\label{ECE}
A probabilistic model $p$ on a contextuality scenario $H$ satisfies the \textit{Extended Consistent Exclusivity} principle (at level $n$) if for all contextuality scenarios $H'$ and $q \in \mathcal{Q}(H')$,
$$ p \otimes q \in \CE^n (H \otimes H').$$
We write $\widetilde{\CE}^n(H)$ for the set of probabilistic models satisfying the Extended Consistent Exclusivity principle at level $n$.
\end{defn}


While this may seem like a reasonable proposal for strengthening the Consistent Exclusivity principle, it is at the same time also a considerable weakening: instead of trying to find one single principle which would single out the quantum models as the physically realistic ones, we have already \emphalt{assumed} quantum models to be physically realistic and propose a principle in order to explain why \emphalt{no other} probabilistic models are physically realistic as well.

So how much does this extension of the Consistent Exclusivity principle help us in detecting non-quantum models as physically unrealistic? In particular, is $\widetilde{\CE}^\infty(H)$ convex and closed under $\otimes$? This result provides the answer:

\begin{thm}
\label{ECEQ1}
All $\widetilde{\CE}^n(H)$, and in particular $\widetilde{\CE}^1(H)$ and $\widetilde{\CE}^\infty(H)$, are equal to $\mathcal{Q}_1(H)$.
\end{thm}

In particular, the properties of $\mathcal{Q}_1$ stated in Proposition~\ref{Q1tensor} imply that all $\widetilde{\CE}^n$ and $\widetilde{\CE}^\infty$ are convex and closed under $\otimes$.

\begin{proof}
The construction presented in the proof of Theorem~\ref{CEactivate} about activation of violations of Consistent Exclusivity shows that $\widetilde{\CE}^1(H) \subseteq  \mathcal{Q}_1(H)$, since any probabilistic model $p\not\in\mathcal{Q}_1(H)$ displays a gap between the Lov\'asz number $\vartheta(\mathrm{NO}(H),p)$ and the Shannon capacity $\Theta(\mathrm{NO}(H),p)$, and any such gap can be exploited to find a scenario $H'$ together with a quantum model $q\in\mathcal{Q}(H')$ such that $p\otimes q\not\in \CE^1(H\otimes H')$.

The inclusion $\mathcal{Q}_1(H)\subseteq\widetilde{\CE}^\infty(H)$ is a consequence of Propositions~\ref{Q1tensor} and~\ref{miguelito}, which show in particular that $\mathcal{Q}_1(H) \otimes \mathcal{Q}(H') \subseteq \mathcal{Q}_1(H\otimes H')\subseteq \CE^\infty(H\otimes H')$ for any $H'$. Together with the trivial inclusion $\widetilde{\CE}^\infty(H)\subseteq\widetilde{\CE}^n(H)\subseteq\widetilde{\CE}^1(H)$, we have therefore shown a cyclic sequence of inclusions,
\[
\mathcal{Q}_1(H)\subseteq\widetilde{\CE}^\infty(H)\subseteq\widetilde{\CE}^n(H)\subseteq\widetilde{\CE}^1(H)\subseteq\mathcal{Q}_1(H),
\]
which forces all these inclusions to be equalities.
\end{proof}

We find it remarkable that as far as the Extended Consistent Exclusivity principle is concerned, it does not matter how many copies of a probabilistic model $p$ we consider. While the original Consistent Exclusivity principle is most powerful when applied to $p^{\otimes n}$ for all $n\in\N$, the extended principle unfolds its full power already at the very first level, and there is no need to consider any further levels.

Theorem~\ref{ECEQ1} makes it very easy to find violations of the ECE principle, since testing membership in $\mathcal{Q}_1(H)$ is a semidefinite program. If a given probabilistic model $p$ is not in $\mathcal{Q}_1(H)$, then the proof of the theorem even provides an explicit $q\in\mathcal{Q}(H')$ for an explicit $H'$ for which $p\otimes q\not\in\CE^1(H\otimes H')$.

These results conclusively delineate the extent to which the ECE principle characterizes the quantum set. Since $\mathcal{Q}_1(H)$ is often quite close to $\mathcal{Q}(H)$, the ECE principle goes a long way in achieving a characterization of the quantum set, but there still remains a gap. Such a gap arises for example in all Bell scenarios, since $\mathcal{Q}_1(B_{n,k,m})$ is strictly larger than $\mathcal{Q}(B_{n,k,m})$ for any non-trivial Bell scenario $B_{n,k,m}$~\cite{NPA}. This is completely opposite to the kind of conclusion one might gather from a cursory reading of papers like~\cite{ECE}, where it is claimed that the ECE principle ``singles out the entire set of quantum correlations''. Again, this different conclusion is due to the different definition of `quantum set' in the CSW formalism~\cite{CSW}, in which the `quantum set' corresponds to our $\mathcal{Q}_1$.

\subsection{Contextuality and perfection}
\label{perfectionsection}

We now study under which conditions on $H$ the classical set $\mathcal{C}(H)$ coincides with $\CE^1(H)$ and is therefore characterized by Consistent Exclusivity. When this is the case, no quantum contextuality is possible in particular.

A graph $G$ is called \emph{perfect} if the chromatic number of any induced subgraph is equal to the clique number of this subgraph \cite{Berge61}. The property of a graph to be perfect is among the most important concepts studied in graph theory.

\begin{thm}
\label{perfection}
If $\mathrm{NO}(H)$ is perfect, then $\mathcal{C}(H)=\CE^1(H)$, although $\mathcal{G}(H)$ can still be bigger.
\end{thm}

\begin{proof}
The weak perfect graph theorem of Lov{\'a}sz~\cite{LovPG} states that a graph is perfect if and only if its complement is. Therefore we can as well assume the complement $\overline{\mathrm{NO}(H)}$ to be perfect. A probabilistic model $p\in \CE^1(H)$ can be interpreted as vertex weights $p(v)$ for $v\in V(H)$ with $\sum_{v\in C} p(v)\leq 1$ for every clique $C$ in $\overline{\mathrm{NO}(H)}$. Then, perfection of this complement guarantees~\cite[Thm.~31]{Knuth} that $p$ is a convex combination of indicator functions of independent sets in $\overline{\mathrm{NO}(H)}$, i.e.~there are cliques $U_1,\ldots,U_k$ in $\mathrm{NO}(H)$ and coefficients $\lambda_i\in[0,1]$ with $\sum_i \lambda_i=1$ such that 
\beq
\label{pconvdecomp}
p = \sum_{i=1}^k \lambda_i \mathbbm{1}_{U_i}.
\eeq
We now claim that every $\mathbbm{1}_{U_i}$ is a deterministic model. Since its weights clearly take values in $\{0,1\}$, it is enough to verify the normalization condition $\sum_{v\in e} \mathbbm{1}_{U_i}(v) = 1$ for all $e\in E(H)$. But this follows from~(\ref{pconvdecomp}) together with $\sum_{v\in e} p(v)=1$.

In order to see that $\mathcal{G}(H)$ can still be bigger, consider again the triangle scenario $\Delta$ depicted in Figure~\ref{triscen}. There, $\mathcal{C}(\Delta) = \CE^1(\Delta) = \emptyset$, although $\Delta$ allows a probabilistic model. The graph $\mathrm{NO}(\Delta)$ is the graph on three vertices with no edges, and therefore trivially perfect.
\end{proof}

The converse to Theorem~\ref{perfection} is not true:

\begin{figure}
\begin{center}
\begin{tikzpicture}[scale=1.4]
\node[draw,shape=circle,fill=gray,scale=.5] (v1) at (0:.07) {} ;
\node[left of=v1,node distance=2mm] {$v$};
\foreach \n in {0,.2,...,1}
{
	\node[draw,shape=circle,fill,scale=.5] at (\n*360+18:1.4) {} ;
	\definecolor{currentcolor}{hsb}{\n,1,1} ;
	\draw[thick,rotate=\n*360+18,currentcolor] plot [smooth cycle,tension=.3] coordinates { (36:-.2) (-7:1.5) (5:1.6) (67:1.6) (79:1.5) } ;
}
\node[draw,shape=circle,fill=gray,scale=.5] (v2) at (3,0) {} ;
\node[draw,shape=circle,fill,scale=.5] (x) at (4,0) {} ;
\node[left of=v2,node distance=2mm] {$v$};
\node[draw,shape=circle,fill,scale=.5] at (5,-1) {} ;
\node[draw,shape=circle,fill,scale=.5] at (5,1) {} ;
\node[draw,shape=circle,fill,scale=.5] at (6,-1) {} ;
\node[draw,shape=circle,fill,scale=.5] at (6,1) {} ;
\draw[thick,brown] (3.5,0) ellipse (1cm and .3cm) ;
\draw[thick,brown] (5.5,1) ellipse (1cm and .3cm) ;
\draw[thick,brown] (5.5,-1) ellipse (1cm and .3cm) ;
\draw[thick,brown] (6,0) ellipse (.3cm and 1.2cm) ;
\draw[thick,brown] plot [smooth cycle,tension=.3] coordinates { (3.6,0) (5.1,1.2) (5.1,-1.2) } ;
\end{tikzpicture}
\end{center}
\caption{A scenario $H_0$ with $\mathcal{G}(H_0) = \mathcal{C}(H_0)$, although $\mathrm{NO}(H_0)$ is not perfect. The two nodes labelled $v$ represent the same vertex.}
\label{stupidexample}
\end{figure}
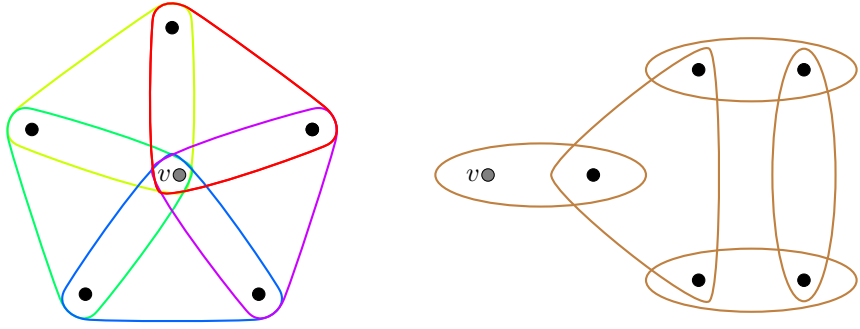

\begin{prop}
For the scenario depicted in Figure~\ref{stupidexample}, $\mathcal{G}(H_0) = \mathcal{C}(H_0)$. However, $\mathrm{NO}(H_0)$ is not perfect.
\end{prop}

\begin{proof}
$\mathrm{NO}(H_0)$ is not perfect since its complement $\overline{\mathrm{NO}(H_0)}$ contains the pentagon $\pentagon$ as an induced subgraph in the left part. The pentagon has clique number $2$, but chromatic number $3$.

On the other hand, every probabilistic model $p$ on $H_0$ is guaranteed to satisfy $p(v)=1$ due to the structure on the right. Hence, $p(u) = 0$ for all $u$ in the pentagon. Therefore, both $\mathcal{G}(H_0)$ and $\mathcal{C}(H_0)$ can be identified with their counterparts for the right part $H_R$ of Figure~\ref{stupidexample}. Since every maximal independent set in $\mathrm{NO}(H_R)$ is itself an edge, we get $\CE^1(H_R)=\mathcal{G}(H_R)$, and since $\mathrm{NO}(H_R)$ is perfect, we have $\mathcal{C}(H_R)=\CE^1(H_R)$.
\end{proof}

Forcing the vanishing of the weights in the pentagon may seem like a cheap trick. However, we do not know of any natural combinatorial condition which one could impose on a contextuality scenario in order to exclude such pathological behavior of $\mathcal{G}(H)$. In particular, the proof of Shultz's Theorem~\ref{shultzthm} uses similar `forcing' ideas~\cite{Shultz}. See Proposition~\ref{APQ1} for a slightly less artificial example of a scenario $\mathrm{AP}_4$ with $\mathcal{Q}_1(\mathrm{AP}_4) = \CE^1(\mathrm{AP}_4)$, although $\mathrm{NO}(\mathrm{AP}_4)$ is not perfect.

There is a deep result from graph theory which we can use to deduce further results on the conditions which a scenario has to satisfy in order for (quantum) contextuality to exist:

\begin{thm}[Strong perfect graph theorem~\cite{CRST}]
A graph $G$ is perfect if and only if neither $G$ nor $\overline{G}$ contains an induced subgraph which is a cycle of odd length $\geq 5$.
\end{thm}

In combination with Theorem~\ref{perfection}, we obtain:

\begin{cor}
\label{contextualoddcycle}
If neither $\mathrm{NO}(H)$ nor $\overline{\mathrm{NO}(H)}$ contains an odd cycle of length $\geq 5$ as an induced subgraph, then $\mathcal{C}(H)=\mathcal{Q}(H)=\CE^1(H)$.
\end{cor}

In this sense, every (quantum) contextuality proof must rely on a `cycle-like' contradiction as it appears in the Klyachko-Can-Binicio{\v{g}}lu-Shumovsky scenario (see~\cite{KCBS} and Section~\ref{circularex}), or on an `anti-cycle-like' contradiction. Within the CSW framework~\cite{CSW}, this observation is due to~\cite{CDLP}, where the anti-cycle case has been studied in a bit more detail.

\newpage
\section{\textbf{Complexity of various decision problems}}
\label{complexity}

We now study the computational complexity of various decision problems associated to contextuality scenarios. 

\subsection{Deciding existence of probabilistic models and classical models}

The most basic decision problem about contextuality scenarios is asking whether a given $H$ admits a probabilistic model or not:

\decprob{ALLOWS\_GENERAL}{A contextuality scenario $H$}{$\mathcal{G}(H)\neq\emptyset$}

Recall that there are indeed contextuality scenarios without any probabilistic models, for example the one depicted in Figure~\ref{empty}. Determining the complexity of \texttt{ALLOWS\_GENERAL} is quite simple:

\begin{prop}
$\mathtt{ALLOWS\_GENERAL}$ is in $\mathbf{P}$.
\end{prop}

\begin{proof}
Determining whether $\mathcal{G}(H)\neq\emptyset$ is a linear program.
\end{proof}

Now we move on to the analogous question about classical models:

\decprob{ALLOWS\_CLASSICAL}{A contextuality scenario $H$}{$\mathcal{C}(H)\neq\emptyset$}

A positive answer to an instance of this problem obviously requires a positive answer to $\mathtt{ALLOWS\_GENERAL}$, since any classical model is in particular a probabilistic model.

\begin{prop}
$\mathtt{ALLOWS\_CLASSICAL}$ is $\mathbf{NP}$-complete.
\end{prop}

\begin{proof}
\texttt{ALLOWS\_CLASSICAL} can be identified with the class of Boolean satisfiability problems which are disjunctions of clauses, where each clause states that exactly one variable in a certain subset of all variables needs to have the value \texttt{TRUE}. Given this, $\mathbf{NP}$-completeness follows from Schaefer's dichotomy theorem~\cite{Schaefer}. Notwithstanding this argument, we now offer an explicit proof.

First, \texttt{ALLOWS\_CLASSICAL} is clearly in $\mathbf{NP}$: any explicit deterministic model $p:V(H)\to\{0,1\}$ witnesses $\mathcal{C}(H)\neq\emptyset$, and verifying that $p$ is a deterministic model can be done in linear time.

To show $\mathbf{NP}$-hardness, let $x_1,\ldots,x_n$ be Boolean variables and
\beq
\label{logform}
B\defin ( y_{11} \lor y_{12} \lor y_{13} ) \land \ldots \land ( y_{m1} \lor y_{m2} \lor y_{m3} )
\eeq
be a logical formula in which each literal $y_{ij}$ stands for some variable $x_l$ or for its negation $\lnot x_l$. The Boolean satisfiability problem \texttt{3SAT} is the following decision problem:
\decprob{3SAT}{a logical formula $B$ in the form~(\ref{logform})}{Is $B$ satisfiable}

This is well-known to be $\mathbf{NP}$-complete~\cite{Karp}. We now prove $\mathbf{NP}$-hardness of \texttt{ALLOWS\_CLASSICAL} by polynomially reducing \texttt{3SAT} to \texttt{ALLOWS\_CLASSICAL}. Denote the clauses in $B$ by
$$
C_i := y_{i1} \lor y_{i2} \lor y_{i3}
$$
and construct a contextuality scenario $H_B$ as follows. We would like the set of vertices to correspond to the set of literals together with $7$ auxiliary variables for each clause in the sense that
$$
V(H_B) \defin \{v_{x_1},\ldots, v_{x_n}, v_{\lnot x_1},\ldots, v_{\lnot x_n} \} \cup \{v_{i,s}\} ,
$$
where $i=1,\ldots,m$ enumerates the clauses and $s\in\{001,010,011,100,101,110,111\}$ runs over the feasible truth value assignments to the literals in a clause. There are three kinds of edges,
\begin{align*}
E(H_B) \defin & \hspace{1.2mm} \phantom{\cup} \left\{ \left\{v_{x_j},v_{\lnot x_j} \right\} \::\: j=1,\ldots,n\right\} \\
& \cup \left\{ \{v_{i,001},\ldots,v_{i,111}\} \::\: i=1,\ldots,m \right\} \\
& \cup \left\{ \{ v_{i,s}, v_{ y_{i1}}, v_{ y_{i2}}, v_{ y_{i3}} \} \::\: i=1,\ldots,m;\: s=001,\ldots,111 \right\}
\end{align*}
where in the third type of edge, the negation $\lnot$ appears if and only if $s$ has a $1$ at the corresponding position. The first type of edge guarantees that in any deterministic model, either $v_{x_j}$ or $v_{\lnot x_j}$ gets the value $1$, but not both; the second kind of edge guarantees that for every $i$, exactly one of the $v_{i,s}$'s is $1$, so that the clause $C_i$ has a unique feasible assignment of truth values $s$; finally, the third type of edge ensures that if $p(v_{i,s})=1$, then the literals of $C_i$ have precisely the values given by $s$. Therefore, the deterministic models on $H_B$ correspond bijectively to the satisfying variable assignments of $B$. So we have that $B$ is satisfiable if and only if $\mathcal{C}(H_B)\neq\emptyset$.
\end{proof}

\subsection{A semidefinite hierarchy converging to $\mathcal{C}(H)$}

For some combinatorial optimization problems, one can construct a contextuality scenario $H$ with polynomially many vertices and edges whose classical set $\mathcal{C}(H)$ coincides with the usual polytope associated to the combinatorial optimization problem~\cite{Schrijver}. We have illustrated how to do this for the case of \texttt{3SAT} above, but similar reductions can be found also e.g.~for coloring problems on graphs. The main idea is that the vertices of $H$ are interpreted as boolean variables, and any formula of propositional logic can be encoded in terms of a collection of edges, possibly using some auxiliary variables. Then, our machinery \emphalt{automatically} produces an associated linear as well a semidefinite relaxation of $\mathcal{C}(H)$: namely $\mathcal{G}(H)$ and $\mathcal{Q}_1(H)$, respectively.

Furthermore, if one takes Definition~\ref{npadefn} and additionally imposes that $M_{\mathbf{v},\mathbf{w}} = M_{\pi(\mathbf{v}),\mathbf{w}}$ for any permutation $\pi$, one obtains a \emph{hierarchy of semidefinite relaxations} $(\mathcal{C}_n(H))_{n\in\N}$ converging to $\mathcal{C}(H)$ \cite{Lasserre,FLS2}; at the first level, the additional constraints do not arise, and therefore we have $\mathcal{C}_1(H)=\mathcal{Q}_1(H)$. In this way, one can efficiently approximate the target set $\mathcal{C}_n(H)$ from the outside. Due to the high number of constraints, this hierarchy converges even after a finite number of steps: we have $\mathcal{C}_{|V(H)|}(H) = \mathcal{C}(H)$ since any matrix entry $M_{\mathbf{v},\mathbf{w}}$ with $\mathbf{v}$ or $\mathbf{w}$ longer than $V(H)$ is already determined by the other matrix entries. However, while every $\mathcal{C}_n$ for fixed $n$ is defined by a semidefinite program of polynomial size, the semidefinite program defining $\mathcal{C}_{|V(H)|}$ is of exponential size. We have not implemented any of this since for any particular class of problems, one can construct specialized (hierarchies of) semidefinite relaxations of smaller size~\cite{Anjos,Laurent}.

\subsection{Towards an inverse sandwich theorem?}
\label{invsandwichsec}

Now that we know the complexity of \texttt{ALLOWS\_GENERAL} and \texttt{ALLOWS\_CLASSICAL}, we move on to consider the analogous question for the quantum case, which may have some surprises to offer.

\decprob{ALLOWS\_QUANTUM}{A contextuality scenario $H$}{$\mathcal{Q}(H)\neq\emptyset$}

This is equivalent to asking whether there exists an assignment of projections $P_v\in\mathcal{B}(\H)$ to each $v\in V(H)$ such that $\sum_{v\in e} P_v=\mathbbm{1}$ for all $e\in E(H)$, since any quantum model requires such an assignment by definition, and conversely any such assignment can be turned into a quantum model by choosing an arbitrary state. The Hilbert space $\H$ can be taken to be separable infinite-dimensional without loss of generality, i.e.~$\H=\ell^2(\N)$: if one starts with a finite-dimensional $\H$ with a given assignments of projections $P_v$, one can simply replace $\H$ by the infinite-dimensional $\H\otimes\ell^2(\N)$ and each $P_v$ by $P_v\otimes\mathbbm{1}$. On the other hand, if the given $\H$ is infinite-dimensional but not separable, then one can consider the $C^*$-algebra generated by all the $P_v$, which is separable, and apply the GNS construction with respect to any state in order to obtain a new representation on a separable Hilbert space.

Is it possible to solve $\mathtt{ALLOWS\_QUANTUM}$ by using the semidefinite hierarchy from Section~\ref{npahierarchy}? After all, by definition of the hierarchy, every set $\mathcal{Q}_n(H)$ is given by a semidefinite program of polynomial size, so that determining whether $\mathcal{Q}_n(H)\neq \emptyset$ can be done efficiently. One might suspect that this should give an algorithm for \texttt{ALLOWS\_QUANTUM} thanks to the following observation:

\begin{lem}
$\mathcal{Q}(H)=\emptyset$ if and only if $\mathcal{Q}_n(H)=\emptyset$ for some $n\in\N$.
\end{lem}

\begin{proof}
If $\mathcal{Q}_n(H)=\emptyset$ for some $n$, then clearly $\mathcal{Q}(H)=\emptyset$ as well. To show the converse, assume $\mathcal{Q}(H)=\emptyset$, so that $\bigcap_n \mathcal{Q}_n(H)=\emptyset$. Since this is an intersection of closed subspaces of the compact space $\mathcal{Q}_1(H)$, we conclude by compactness that already finitely many of the $\mathcal{Q}_n(H)$ have empty intersection. Because the $\mathcal{Q}_n(H)$ form a decreasing sequence of sets, there has to be some $n\in\N$ with $\mathcal{Q}_n(H)=\emptyset$.
\end{proof}

The problem with this is that checking whether $\mathcal{Q}_n(H)=\emptyset$ for each $n$ at a time is a procedure that never terminates in case that $\mathcal{Q}(H)\neq \emptyset$. Hence, in order to find an algorithm for \texttt{ALLOWS\_QUANTUM}, we also need a procedure for witnessing that $\mathcal{Q}(H)\neq\emptyset$ if this happens to be the case!

One way to go about this is to try and look in every finite Hilbert space dimension $\H:=\C^d$ at a time and see if there exists a quantum model in this dimension. For each given $d$, this boils down to determining whether a certain system of polynomial equations and inequalities has a solution in $\R$. Thanks to real quantifier elimination~\cite{Tarski}, there are known algorithms for doing this. Therefore, if a quantum model over some finite-dimensional Hilbert space exists, this procedure will eventually find it---even if this may take an exceedingly long time.

By running these two procedures in parallel, we have an algorithm for deciding \texttt{ALLOWS\_QUANTUM} that works in all cases---except when $H$ allows quantum models, but only on infinite-dimensional Hilbert spaces! In this case, both procedures will keep running forever: the semidefinite hierarchy, which tries to show that $\mathcal{Q}_n=\emptyset$ for some $n$, will not terminate since $\mathcal{Q}(H)\neq\emptyset$, but also the look-in-all-finite-dimensions procedure will not be successful since there is no quantum model in finite dimension. Thus we are faced with a mathematical problem:

\begin{prob}
\label{Hinfdim}
Are there contextuality scenarios $H$ which allow quantum models, but only in infinite dimensions?
\end{prob}

We now explain why this is an important problem. In the language of~\cite{FNT}, it can be rephrased as follows: we construct the universal unital $C^*$-algebra associated to a contextuality scenario $H$ in terms of generators and relations,
$$
C^*(H) \defin \left\langle \left\{ P_v \::\: v\in V(H) \right\} \:\bigg|\: P_v = P_v^2 = P_v^*\quad\forall v\in V(H),\quad \sum_{v\in e} P_v = \mathbbm{1} \quad\forall e\in E(H) \right\rangle.
$$
If this $C^*$-algebra is residually finite-dimensional\footnote{Residual finite-dimensionality means that for any nonzero element $x\in C^*(H)$, there is a finite-dimensional representation $\pi$ of $C^*(H)$ with $\pi(x) \ne 0$.} for any $H$, then Problem~\ref{Hinfdim} has a negative answer and the above algorithm solves \texttt{ALLOWS\_QUANTUM}, even if with very high complexity.

Now it is known that Kirchberg's QWEP conjecture and Connes' embedding problem are equivalent to the residual finite-dimensionality of some of these $C^*$-algebras, for example for $C^*(B_{2,3,2})$~\cite{TF,FNT}. Since these are notoriously difficult open problems in the theory of operator algebras, and moreover are generally expected to have a negative answer, we suspect that it is too much to hope for that all $C^*(H)$ are residually finite-dimensional. One way to show this---and thereby also make considerable progress on Connes' embedding problem---would be to solve Problem~\ref{Hinfdim} in the positive.

To conclude, our attempt at constructing an algorithm for deciding $\mathtt{ALLOWS\_QUANTUM}$ has not succeeded, but we have found that $\mathtt{ALLOWS\_QUANTUM}$ is related to Connes' embedding problem and posed the interesting Problem~\ref{Hinfdim}. But may an algorithm for $\mathtt{ALLOWS\_QUANTUM}$ be constructed in a different way? A different approach to $\mathtt{ALLOWS\_QUANTUM}$ lies in recognizing that any instance of it can be reformulated as an $\exists_1$ formula in quantum logic with signature $(\lor,\perp,\mathbbm{1}_{\H})$ on an infinite-dimensional separable Hilbert space. However, since the decidability status of quantum logic is also not known~\cite[p.69]{Svozil}, this approach does not produce a terminating algorithm either and we will not discuss it further.

In conclusion, we do not know of \emphalt{any} terminating algorithm that would solve \texttt{ALLOWS\_QUANTUM}. In fact, we suspect the following:

\begin{conj}
\label{invsandwich}
$\mathtt{ALLOWS\_QUANTUM}$ is undecidable.
\end{conj}

Here is how we think of this conjecture. Recall that if one writes $\chi$ for the chromatic number of a graph, then Lov{\'a}sz's \emph{sandwich theorem}~\cite{Knuth} consists of the inequality 
\beq
\label{lovsandwich}
\alpha(G) \leq \vartheta(G) \leq \chi(\overline{G}),
\eeq
together with the observation that the outer two quantities, the independence and the chromatic number, are $\mathbf{NP}$-hard to compute, while $\vartheta(G)$ can be computed in polynomial time to arbitrary precision. The simple-to-compute quantity $\vartheta$ is `sandwiched' between two hard-to-compute graph invariants.

In analogy with this, we call~\ref{invsandwich} the \emph{inverse sandwich conjecture} since the hypothetically uncomputable $\mathcal{Q}(H)$ is sandwiched between two computable sets,
$$
\mathcal{C}(H) \subseteq \mathcal{Q}(H) \subseteq \mathcal{G}(H) .
$$
So in contrast to the case of~(\ref{lovsandwich}), here the real meat indeed lies in the middle of the sandwich! See also~\cite{WCP}, where it has previously been hypothesized that the set of quantum models in a Bell scenario cannot be characterized algorithmically.

A proof of Conjecture~\ref{invsandwich} would also yield a positive answer to Problem~\ref{Hinfdim}, since undecidability means that our above algorithm cannot terminate on all $H$. In this way, a proof of Conjecture~\ref{invsandwich} would have some interesting consequences for $C^*$-algebra theory. Moreover, it would also prove the undecidability of quantum logic\footnote{More precisely, it would imply that the theory of Hilbert lattices in the signature $(\lor,\perp,\mathbbm{1})$ is not decidable.}. Since these are all very difficult problems in themselves, proving Conjecture~\ref{invsandwich}---if it is correct---will also be very challenging.

\subsection{Other decision problems}

There is a myriad of other interesting decision problems on contextuality scenarios that one can come up with. We now briefly mention several further ones.

\decprob{IS\_CLASSICAL}{A contextuality scenario $H$ and $p\in\mathcal{G}(H)$ with $p(v)\in\Q$}{$p\in\mathcal{C}(H)$}

It is not difficult to see that this is in $\mathbf{NP}$. Furthermore, it is actually $\mathbf{NP}$-complete, since this is the case already for Bell scenarios~\cite{AII}.

Similarly, one can consider decision problems like $\mathtt{IS\_QUANTUM}$ and $\mathtt{IS\_CE}^\infty$. So far, we have not considered any of these any further. Another natural decision problem is the question whether a given scenario allows nonclassical models or not:

\decprob{NONCONTEXTUAL}{A contextuality scenario $H$}{$\mathcal{C}(H)=\mathcal{G}(H)$} 

We also do not know what the complexity of this problem is. We suspect that Theorem~\ref{extchar} together with the techniques of~\cite{Eiter} will be helpful for answering this question.

\newpage
\section{\textbf{Examples}}
\label{examples}

In the previous sections, we have developed the general theory of contextuality scenarios in quite some detail. We have exemplified some of the concepts and results for the case of Bell scenarios. In particular, this illustrates how our formalism makes precise the intuition that nonlocality is a special case of contextuality. 

Now we would like to present some other more or less concrete examples of contextuality scenarios and show how our methods can be applied to these. We note that compiling a detailed list of the examples that have already been considered in the quantum foundations literature would be a gargantuan task beyond the reach of this paper.

\subsection{Modeling subnormalization by no-detection events}\label{CSWtransfer}

We start by discussing the relationship between our approach and that of Cabello, Severini and Winter~\cite{CSW}. CSW base their approach also on hypergraphs $H$ in a very similar spirit as we have done, and this is where we drew our inspiration from. The main difference between the CSW approach and ours is that CSW do not require measurements to be normalized,
\beq
\label{norm}
\sum_{v\in e} p(v) = 1 \quad\forall e\in E(H) ,
\eeq
but only impose the subnormalization constraint
\beq
\label{subnorm}
\sum_{v\in e} p(v) \leq 1 \quad \forall e\in E(H) .
\eeq
A similar requirement $\sum_{v\in e} P_v\leq \mathbbm{1}_\H$ is applied for the projections giving rise to quantum models. We now explain how our approach comprises CSW's. To this end, we construct a new contextuality scenario $H'$ which contains an additional \emph{no-detection event} $w_e$ for each $e\in E(H)$,
\[
V(H')\defin V(H) \cup \{ w_e \::\: e\in E(H) \} 
\]
which turns the `old' edge $e$ into the `new' edge $e\cup\{w_e\}$, so that the new set of edges is given by
\[
\qquad E(H') \defin \left\{ e \cup \{w_e\} \::\: e\in E(H) \right\} .
\]
The r{\^o}le of these no-detection events is to absorb the `missing probability' in the subnormalization equation~\eqref{subnorm}, which also explains the term `no-detection event'. In fact, the normalization~\eqref{norm} for the new edge $e\cup\{w_e\}$ can be rewritten as
\[
p(w_e) = 1 - \sum_{v\in e} p(v).
\]
The non-negativity of this probability is precisely equivalent to the subnormalization~\eqref{subnorm}. It is now straightforward to show:

\begin{prop}
\label{toCSW}
Under this correspondence,
\begin{enumerate}
\item $\mathcal{C}(H')$ equals the set of `classical noncontextual models',
\item $\mathcal{Q}_1(H')$ equals the set of `quantum models', and
\item $\CE^1(H')$ equals the set of `generalized models'
\end{enumerate}
of CSW~\cite{CSW}.
\end{prop}

CSW have also considered `quantum models' on $H$ which do satisfy the normalization of probabilities~\eqref{norm} and therefore lie in our $\mathcal{G}(H)$, without adding no-detection events. Again, it is straightforward to show that their definition of `quantum model' is equivalent to $p\in\mathcal{Q}_1(H)$, which means that $p$ does \emphalt{not} have to be a quantum model in our sense.

Proposition~\ref{toCSW} shows how our approach comprises the one of CSW. However, the converse is not true: it was already noticed in~\cite{CSW} that upon applying the CSW approach to a Bell scenario, the resulting set of `quantum models' in the CSW sense is usually strictly greater than the set of quantum correlations in the Bell scenario, in the standard meaning of the term~\cite{nonlocalreview}. This should be seen in contrast to our Corollary~\ref{Bellquantum}, which shows that our quantum set does indeed recover the usual quantum correlations (with the commutativity paradigm for composite systems). This is why we consider the scenarios $H'$ as nothing but one particular class of examples for our approach.

In fact, the CSW approach completely fails to see the distinction between the different levels of the semidefinite hierarchy developed in Section~\ref{npahierarchy}:

\begin{prop}
\label{collapse}
Let $H'$ be a contextuality scenario with no-detection events as above. Then $\mathcal{Q}(H')=\mathcal{Q}_1(H')$.
\end{prop}

\begin{proof}
Starting from $p\in\mathcal{Q}_1(H')$, we need to show that $p\in\mathcal{Q}(H')$; the converse inclusion is trivial. To do so, we use Proposition~\ref{CSWeq}\ref{port} as a criterion for membership in $\mathcal{Q}_1(H')$. By definition of $H'$, this means that we have a projection $P_v$ for every $v\in V(H)$ and $P_{w_e}$ for every $e\in E(H)$ such that
$$
u\perp v \:\implies \: P_u \perp P_v,\qquad v\in e \: \implies \: P_v\perp P_{w_e} ,
$$
and $p(v) = \langle\Psi|P_v|\Psi\rangle$ as well as $p(w_e) = \langle\Psi|P_{w_e}|\Psi\rangle$. We now define
$$
P'_{w_e} \defin \mathbbm{1}_\H - \sum_{v\in e} P_v
$$
and claim that these, together with the $P_v$ and the state $|\Psi\rangle$, realize the given $p$ as a quantum model. First, due to $\sum_{v\in e} P_v\leq \mathbbm{1}_\H$, the operator $P'_{w_e}$ is also a projection. Second, the projection-valued normalization relation for edges in $E(H')$ holds by definition of $P'_{w_e}$. Third,
$$
\langle\Psi|P'_{w_e}|\Psi\rangle = \langle\Psi|\Psi\rangle - \sum_{v\in e} \langle\Psi|P_v|\Psi\rangle = 1 - \sum_{v\in e} p(v) = p(w_e) .
$$
as claimed. Hence, $p\in\mathcal{Q}(H')$.
\end{proof}

In this sense, the set of quantum models of a scenario which arises in this way is particularly simple: the whole semidefinite hierarchy collapses to the first level! So, scenarios constructed in this way form a very special and well-behaved subclass of all contextuality scenarios. The $n$-circular hypergraphs that we consider next arise in this way. However, many of the more interesting contextuality scenarios---like Bell scenarios---are not of this form and therefore cannot be treated correctly in the CSW approach, as already noticed by CSW~\cite{CSW}.

\subsection{$n$-circular hypergraphs} The $n$-circular hypergraphs generalize the `pentagon' idea of Klyachko-Can-Binicio{\v{g}}lu-Shumovsky (KCBS)~\cite{KCBS}.
\label{circularex}

\begin{defn}
For $n\geq 3$, the \emph{$n$-circular hypergraph} $\Delta_{n}$ is given by
\begin{align*}
V(\Delta_n) & \defin \{v_1,\ldots,v_n,w_1,\ldots,w_n\} , \\
E(\Delta_n) & \defin \left\{ \{v_1,w_1,v_2\},\ldots,\{v_n,w_n,v_1\} \right\} .
\end{align*}
\end{defn}

So, $\Delta_n$ has $2n$ vertices and $n$ edges as follows: if one draws all vertices on a circle in the order $v_1,w_1,\ldots,v_n,w_n,v_1$, then every second triple of adjacent vertices, namely those of the form $\{v_j,w_j,v_{j+1}\}$, is an edge (we write $v_{n+1}=v_1$). The $w_i$ can be interpreted as no-detection events as explained in the previous subsection. In particular, Proposition~\ref{collapse} applies, and we see that $\mathcal{Q}(\Delta_n)=\mathcal{Q}_1(\Delta_n)$.

Figure~\ref{3meas6out} displays $\Delta_3$, which can be metaphorically illustrated as a \emph{firefly box}~\cite{Wilce2}. It corresponds to the \emph{Wright triangle} of~\cite[Ex.~2.13]{FGR} under the relabeling
$$
v_1\mapsto a,\quad w_1\mapsto b,\quad v_2\mapsto c,\quad w_2\mapsto d,\quad v_3\mapsto e,\quad w_3\mapsto f .
$$

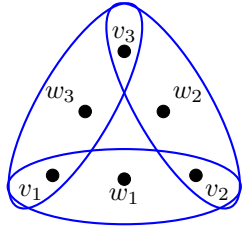
\begin{figure}
\begin{center}
\begin{tikzpicture}
\node[draw,shape=circle,fill,scale=.5] (a) at (30:.6) {} ;
\node [above right] at (a) {$w_2$};
\node[draw,shape=circle,fill,scale=.5] (b) at (90:1.1) {} ;
\node [above] at (b) {$v_3$};
\node[draw,shape=circle,fill,scale=.5] (c) at (150:.6) {} ;
\node [above left] at (c) {$w_3$};
\node[draw,shape=circle,fill,scale=.5] (d) at (210:1.1) {} ;
\node [below left] at (d) {$v_1$};
\node[draw,shape=circle,fill,scale=.5] (e) at (270:.6) {} ;
\node [below] at (e) {$w_1$};
\node[draw,shape=circle,fill,scale=.5] (f) at (330:1.1) {} ;
\node [below right] at (f) {$v_2$};
\draw[thick,blue,rotate=270] (0:.7) ellipse (.5cm and 1.55cm) ;
\draw[thick,blue,rotate=150] (0:.75) ellipse (.5cm and 1.55cm) ;
\draw[thick,blue,rotate=30] (0:.75) ellipse (.5cm and 1.55cm) ;
\end{tikzpicture}
\end{center}
\caption{The $3$-circular hypergraph $\Delta_3$. The labeling of the vertices corresponds to~\cite[Ex.~2.13]{FGR}.}
\label{3meas6out}
\end{figure}

$\Delta_5$ is the `pentagon' scenario on which the KCBS inequality~\cite{KCBS} is defined. It was first considered by Wright in 1978~\cite{Wright}. We now extend some of these results to arbitrary $n$.

\begin{prop}
\label{ch}
Let $n\geq 3$.
\begin{enumerate}
\item\label{chdim} $\dim(\mathcal{C}(\Delta_n))=\dim(\mathcal{G}(\Delta_n))=n$.
\item\label{chparta} If $n$ is even, then $\mathcal{C}(\Delta_n) = \mathcal{G}(\Delta_n)$.
\item\label{chpartb} If $n$ is odd, then $\mathcal{C}(\Delta_n) \subsetneq \mathcal{G}(\Delta_n)$ is determined by the inequality
\beq
\label{kcbsineq}
\sum_i p(v_i) \leq \frac{n-1}{2} .
\eeq
which, for $n=5$, is the KCBS inequality. There is one extreme point of $\mathcal{G}(\Delta_n)$ which violates this inequality. It is the probabilistic model $p_x\in\mathcal{G}(\Delta_n)$ with 
\beq
\label{contextualbox}
p_x(v_i) = \tfrac{1}{2} \quad\forall i,\qquad p_x(w_i) = 0 \quad\forall i.
\eeq
In particular, $\mathcal{G}(\Delta_n)$ has one vertex more than $\mathcal{C}(\Delta_n)$.
\end{enumerate}
\end{prop}

\begin{proof}
We consider all vertex indices modulo $n$, so that $v_{n+1}=v_1$ etc.
\begin{enumerate}
\item[\ref{chdim}] The equations imposed on the probabilities $p(v_i)$ and $p(w_i)$ by the normalization constraints are just
\beq
\label{nodecprob}
p(w_i) = 1 - p(v_i) - p(v_{i+1}) ,
\eeq
which implies $\dim(\mathcal{G}(\Delta_n))\leq n$. The conclusion for both $\mathcal{C}(\Delta_n)$ and $\mathcal{G}(\Delta_n)$ follows from this if we can produce $n+1$ linearly independent deterministic models, which together with normalization would imply that $\dim(\mathcal{C}(\Delta_n))\geq n$. This is simple: the set of models
$$
p_j(v_i) \defin \begin{cases} 1 & \textrm{if } i=j, \\ 0 & \textrm{otherwise}, \end{cases} 
$$
where $j\in\{1,\ldots,n\}$ and the $p_j(w_i)$ are uniquely determined thanks to~(\ref{nodecprob}), is linearly independent. Furthermore, adding to this set the model $p_0$ with $p_0(v_i)=0$ for all $i$ preserves linear independence. This is the desired collection of $n+1$ linearly independent deterministic models.

\item[\ref{chparta}] $\mathcal{C}(\Delta_n) =\CE^1(\Delta_n)$ follows from Corollary~\ref{contextualoddcycle}, and $\CE^1(\Delta_n) = \mathcal{G}(\Delta_n)$ holds because the maximal independent sets of $\mathrm{NO}(\Delta_n)$ are precisely the edges of $\Delta_n$. We obtain the claim by combining these two statements.

In particular, while~(\ref{contextualbox}) is also a probabilistic model for even $n$, in this case it has to be a convex combination of deterministic models. Also, note that our reasoning has not made use of~\ref{chdim}.
\item[\ref{chpartb}] 
We apply Theorem~\ref{extchar} in combination with Corollary~\ref{contextualoddcycle}. Any induced subscenario $H_W$ with $\mathcal{C}(H_W)\neq\mathcal{G}(H_W)$ needs to contain an induced (anti-)cycle of length $\geq 5$ in $\mathrm{NO}(H_W)$. This is possible only if $W$ contains all $v_i$. If $W$ also contains one or more of the $w_i$'s, then $H_W$ does not have a unique probabilistic model. Therefore, there can be at most one nonclassical extreme point of $\mathcal{G}(H)$, namely the one associated to the induced subscenario on $W\defin\{v_1,\ldots,v_n\}$. Now this $H_W$ does indeed have a unique probabilistic model given by $p_x(v_i) = \tfrac{1}{2}$, which yields~(\ref{contextualbox}) upon extension to $\Delta_n$. This proves that $\mathcal{G}(\Delta_n)$ has $p_x$ as its sole nonclassical extreme point without ever using any inequalities.

We now give an independent proof showing that~(\ref{kcbsineq}) characterizes $\mathcal{C}(\Delta_n)$. Thanks to~(\ref{nodecprob}), it is enough to consider the values $p(v_i)$ only. Now the deterministic models correspond to the independent sets in the $n$-cycle graph $C_n$; upon identifying each vertex with the edge adjacent on its left, an independent set in $C_n$ gets identified with a set of edges in $C_n$ no two of which are adjacent at the same vertex, i.e.~with a \emph{matching} on $C_n$. Now it is known~\cite{Schrijver} that the polytope of all matchings on $C_n$ is given by
$$
p(v_i) \geq 0,\qquad p(v_i) + p(v_{i+1}) \leq 1,\qquad \sum_{i=1}^n p(v_i) \leq \frac{n-1}{2} .
$$
This is precisely the description of $\mathcal{C}(\Delta_n)$ that was to be proven.\qedhere
\end{enumerate}
\end{proof}

Compare~\cite{AQTC} for the characterization of classical models in a related scenario.

Concerning the Consistent Exclusivity principle on $\Delta_n$, we have found:

\begin{prop}
$\mathcal{C}(\Delta_{3}) = \CE^1(\Delta_{3}) \subsetneq \mathcal{G}(\Delta_{3})$. For all other $n$, $\CE^1(\Delta_n) = \mathcal{G}(\Delta_n)$.
\end{prop}

\begin{proof}
Since $\{v_1,v_2,v_3\}$ is the only independent set in $\mathrm{NO}(\Delta_3)$ which is not an edge of $\Delta_3$, we find that $\CE^1(\Delta_{3})$ as a subset of $\mathcal{G}(\Delta_3)$ is given by imposing the inequality $p(v_1) + p(v_2) + p(v_3) \leq 1$. This is precisely the inequality~\eqref{kcbsineq} that determines $\mathcal{C}(\Delta_{3})$. For $n\geq 4$, however, every independent set in $\mathrm{NO}(\Delta_n)$ is of the form $\{v_i,w_i,v_{i+1}\}$, i.e.~is itself an edge.
\end{proof}

\subsection{Antiprism scenarios}
\label{antiprism}

\begin{figure}
\begin{center}
\begin{tikzpicture}[scale=0.9]
\definecolor{darkgreen}{rgb}{0,.5,0}
\definecolor{darkyellow}{rgb}{.8,.75,0}
\node[draw,shape=circle,fill,scale=.5] (a1) at (0:2) {} ;
\node at (0:2.8) {$w_2$};
\node[draw,shape=circle,fill,scale=.5] (b1) at (90:2) {} ;
\node at (90:2.8) {$w_1$};
\node[draw,shape=circle,fill,scale=.5] (c1) at (180:2) {} ;
\node at (180:2.8) {$w_4$};
\node[draw,shape=circle,fill,scale=.5] (d1) at (270:2) {} ;
\node at (270:2.8) {$w_3$};
\node[draw,shape=circle,fill,scale=.5] (a2) at (45:5) {} ;
\node at (45:5.8) {$v_2$};
\node[draw,shape=circle,fill,scale=.5] (b2) at (135:5) {} ;
\node at (135:5.8) {$v_1$};
\node[draw,shape=circle,fill,scale=.5] (c2) at (225:5) {} ;
\node at (225:5.8) {$v_4$};
\node[draw,shape=circle,fill,scale=.5] (d2) at (315:5) {} ;
\node at (315:5.8) {$v_3$};
\draw[thick,blue] plot [smooth cycle,tension=.3] coordinates { (-10:2.1) (45:5.5) (100:2.1) } ;
\draw[thick,darkyellow,rotate=90] plot [smooth cycle,tension=.3] coordinates { (-10:2.1) (45:5.5) (100:2.1) } ;
\draw[thick,blue,rotate=180] plot [smooth cycle,tension=.3] coordinates { (-10:2.1) (45:5.5) (100:2.1) } ;
\draw[thick,darkyellow,rotate=270] plot [smooth cycle,tension=.3] coordinates { (-10:2.1) (45:5.5) (100:2.1) } ;
\draw[thick,red] plot [smooth cycle,tension=.3] coordinates { (0:1.7) (312:4.5) (318:5.15) (42:5.15) (48:4.5) } ;
\draw[thick,darkgreen,rotate=90] plot [smooth cycle,tension=.3] coordinates { (0:1.7) (312:4.5) (318:5.15) (42:5.15) (48:4.5) } ;
\draw[thick,red,rotate=180] plot [smooth cycle,tension=.3] coordinates { (0:1.7) (312:4.5) (318:5.15) (42:5.15) (48:4.5) } ;
\draw[thick,darkgreen,rotate=270] plot [smooth cycle,tension=.3] coordinates { (0:1.7) (312:4.5) (318:5.15) (42:5.15) (48:4.5) } ;
\end{tikzpicture}
\end{center}
\caption{The contextuality scenario $AP_4$.}
\label{AP4}
\end{figure}
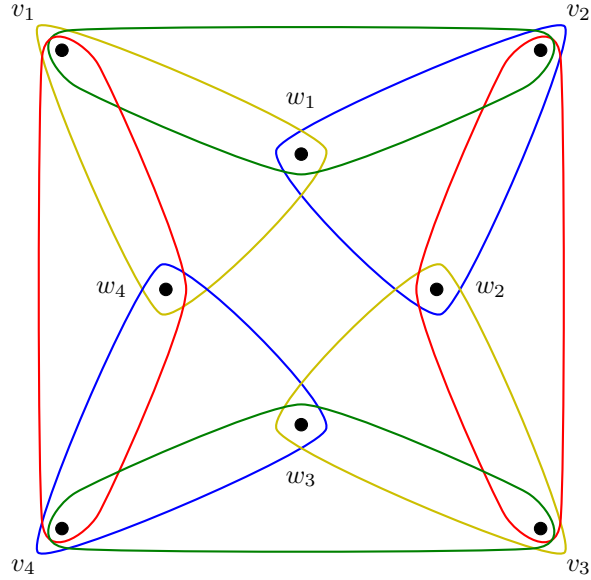

The antiprism scenarios are a variant of the circular hypergraph scenarios with some additional edges thrown in such that there is a symmetry exchanging the $v_i$ with the $w_i$. Again, we consider all vertex indices modulo $n$. The antiprism scenarios are supposed to illustrate that an interesting looking hypergraph is not necessarily an interesting contextuality scenario.

\begin{defn}
Let $n\geq 3$. The \emph{$n$-antiprism scenario} $AP_n$ is
\begin{align*}
V(AP_n) \defin & \hspace{1.0mm} \{v_1,\ldots,v_n,w_1,\ldots,w_n\} , \\
E(AP_n) \defin & \hspace{3.2mm} \left\{ \{v_1,w_1,v_2\},\ldots,\{v_n,w_n,v_1\} \right\}  \\
  & \cup \left\{ \{w_1,v_2,w_2\},\ldots,\{w_n,v_1,w_1\} \right\} .
\end{align*}
\end{defn}

The idea behind the term `antiprism' is that one gets $AP_n$ by considering the antiprism polytope over an $n$-gon and defines a hypergraph $AP_n$ as given by the band of triangles winding itself around the polytope.

\begin{prop}
If $n$ is divisible by $3$, then $\mathcal{C}(AP_n)=\mathcal{G}(AP_n)$ is a $2$-dimensional triangle. Otherwise, $AP_n$ has a unique probabilistic model which is not classical.
\end{prop}

\begin{proof}
We show that $p(v_1)$ and $p(w_1)$ determine all other probabilities $p(v_i)$ and $p(w_i)$ by induction on $i$:
$$
p(v_{i+1}) = 1 - p(v_i) - p(w_i) ,\qquad p(w_{i+1}) = 1 - p(w_i) - p(v_{i+1}) .
$$
In fact, this shows that for all $j$,
$$
p(v_{3j+1}) = p(w_{3j+2}) = p(v_1),\quad p(v_{3j+2}) = p(w_{3j}) = 1-p(v_1)-p(w_1), \quad p(v_{3j}) = p(w_{3j+1}) = p(w_1) .
$$
Now if $n$ is divisible by $3$, then this is consistent upon `going around the cycle', so that $\mathcal{G}(AP_n)$ can be identified with the triangle
$$
p(v_1) \geq 0,\qquad p(v_2) \geq 0,\qquad p(v_1) + p(v_2) \leq 1.
$$
Clearly, the extreme points of this triangle are deterministic, and therefore $\mathcal{C}(AP_n)=\mathcal{G}(AP_n)$.

If $n$ is not divisible by $3$, then the above recurrence relations imply that $p(v_1) = p(v_2) = \tfrac{1}{3}$, so that $\mathcal{G}(AP_n)$ degenerates to a single point. $\mathcal{C}(AP_n)=\emptyset$ holds since there is no deterministic model.
\end{proof}

This may make clear that as a contextuality scenario, $AP_n$ is not very interesting. Nevertheless, it serves well for illustrating our methods once more:

\begin{prop}
\label{APQ1}
$\mathcal{Q}_1(AP_4) = \emptyset$, although $\CE^1(AP_4) = \mathcal{G}(AP_4)$.
\end{prop}

\begin{proof}
Direct inspection shows that every maximal independent set in $\mathrm{NO}(AP_4)$ is an edge, so that the unique probabilistic model given by $p(v_i)=p(w_i)=\tfrac{1}{3}$ is in $\CE^1(AP_4)$. 

It remains to show that this unique probabilistic model is not in $\mathcal{Q}_1(AP_4)$. By Proposition~\ref{Q1vsLov}, this boils down to showing that $\tfrac{1}{3}\vartheta(\mathrm{NO}(AP_4))>1$. Now $\mathrm{NO}(AP_n)$ is the complement of the $4$-antiprism graph $\textarc{m}_4$. Since $\textarc{m}_4$ is vertex-transitive, we deduce~\cite[Thm.~25]{Knuth} that $\vartheta(\textarc{m}_4)\vartheta(\mathrm{NO}(AP_4))=8$. Now $\vartheta(\textarc{m}_4)$ is known~\cite{BPT} to equal $8-4\sqrt{2}$, so that
$$
\vartheta(\mathrm{NO}(AP_4))=\frac{8}{8-4\sqrt{2}} = \frac{2}{2-\sqrt{2}} = 2+\sqrt{2} > 3,
$$
as was to be shown.
\end{proof}

Also, note that the antiprism graph $\textarc{m}_4$ which appears in this proof has also arisen as the non-orthogonality graph of possible events for the PR-box~\cite{Cabello,LOfp}.

\subsection{Matching scenarios}

We now study briefly a very interesting and relevant family of contextuality scenarios. Let $K_m$ be the complete graph on $m$ vertices. In order not to confuse the vertices and edges of $K_m$ with the vertices and edges of the contextuality scenario that we will construct, we will talk about \emphalt{nodes} and \emphalt{arcs} when referring to $K_m$.

We define a contextuality scenario $\mathrm{Mat}_m$ as follows. $V(\mathrm{Mat}_m)$ is defined to be the set of arcs of $K_m$, so that $|V(\mathrm{Mat}_m)|=\tfrac{m(m-1)}{2}$. The set of edges of $\mathrm{Mat}_m$ is $E(\mathrm{Mat}_m)=\{e_1,\ldots,e_m\}$, where $e_j$ is indexed by a node $j\in K_m$ and is defined to be the set of all arcs in $K_m$ adjacent to the node $j$. In the language of hypergraph theory~\cite{Vol}, $\mathrm{Mat}_m$ is the \emph{dual} of $K_m$. For reasons that will become clear, we call it a \emph{matching scenario}.

Matching scenarios have been studied previously: for example, $\mathrm{Mat}_5$ coincides with Figure~2(b) from~\cite{P3M}. Moreover, using the CSW formalism~\cite{CSW}, it has also been studied in~\cite{Cabtwin}. These latter results can be transferred to our setting using the construction of Section~\ref{CSWtransfer}, but they will live in the contextuality scenario $\mathrm{Mat}'_5$ which contains additional vertices representing no-detection events. Studying the scenario $\mathrm{Mat}_m$ itself is more interesting than that; and in fact, after a first version of this paper was made public, it was found~\cite{simplest} that there are quantum models of $\mathrm{Mat}_7$ in $\mathcal{H}=\C^6$. Since we show below that $\mathcal{C}(\mathrm{Mat}_7)=\emptyset$, this constitutes a new state-independent proof of the Kochen--Specker theorem in the sense of Section~\ref{siKS}.

There are certain probabilistic models on $\mathrm{Mat_m}$ which have a special form. By a \emph{half-integer matching}, we mean a probabilistic model on $\mathrm{Mat}_m$ in which each probability lies in $\{0,\tfrac{1}{2},1\}$ in such a way that the arcs with positive probability define a decomposition of $K_m$ into cycles of odd length, where we regard an arc of probability $1$ as a cycle of length $1$. In particular, every \emph{perfect matching} on $K_m$ can be regarded as a half-integer matching.

\begin{prop}
\label{Mat}
\begin{enumerate}
\item\label{MatD} The deterministic models on $\mathrm{Mat}_m$ are precisely the perfect matchings on $K_m$.
\item\label{MatC} $\mathcal{C}(\mathrm{Mat}_m)$ is the perfect matching polytope~\cite{Schrijver} on $K_m$. In particular, $\mathcal{C}(\mathrm{Mat_m})\neq \emptyset$ if and only if $m$ is even.
\item\label{MatG} $\mathcal{G}(\mathrm{Mat}_m)$ is the fractional matching polytope. Its extreme points are precisely the half-integer matchings.
\item\label{MatLO} $\CE^1(\mathrm{Mat}_m)$ is a polytope strictly intermediate between $\mathcal{C}(\mathrm{Mat_m})$ and $\mathcal{G}(\mathrm{Mat_m})$ for $m\geq 5$.
\end{enumerate}
\end{prop}

\begin{proof}
\begin{enumerate}
\item Using Remark~\ref{V1}, a deterministic model corresponds to a collection of arcs in $K_m$ such that there is exactly one arc incident to each node. This is the definition of perfect matching.
\item $\mathcal{C}(\mathrm{Mat}_m)$ is defined to be the convex hull of the deterministic models, and likewise the perfect matching polytope is defined to be the convex hull of the perfect matchings, in the same ambient space. Therefore this follows immediately from~\ref{MatD}.
\item The inequalities defining $\mathcal{G}(\mathrm{Mat}_m)$ are precisely those defining the standard linear relaxation of the perfect matching polytope. Its extreme points are known to be the half-integer matchings~\cite{Schrijver}. This can also be proven using Theorem~\ref{extchar}.
\item For $m\geq 5$, there are two kinds of maximal independent sets in $\mathrm{NO}(\mathrm{Mat}_m)$: first, the edges of $\mathrm{Mat}_m$ themselves; second, all triples of arcs in $K_m$ that form a triangle. In $\CE^1(\mathrm{Mat}_m)$, the latter impose the additional constraint that the sum of the edge weights in a triangle should not exceed $1$. Hence the half-integer matchings with cycles of length $3$ do not belong to $\CE^1(\mathrm{Mat}_m)$, which is therefore a polytope strictly contained in $\mathcal{G}(\mathrm{Mat}_m)$. On the other hand, $\CE^1(\mathrm{Mat}_m)$ still contains half-integer matchings with odd cycles of length $\geq 5$, which are not in $\mathcal{C}(\mathrm{Mat}_m)$.\qedhere
\end{enumerate}
\end{proof}

By definition, $\mathrm{NO}(\mathrm{Mat}_m)$ is the Kneser graph $KG_{m,2}$~\cite{LovKneser}. In particular, $\mathrm{NO}(\mathrm{Mat}_5)$ is the well-known Petersen graph, one of the most widely studied graphs in graph theory. For those interested in combinatorial optimization, the curiosities do not end here:

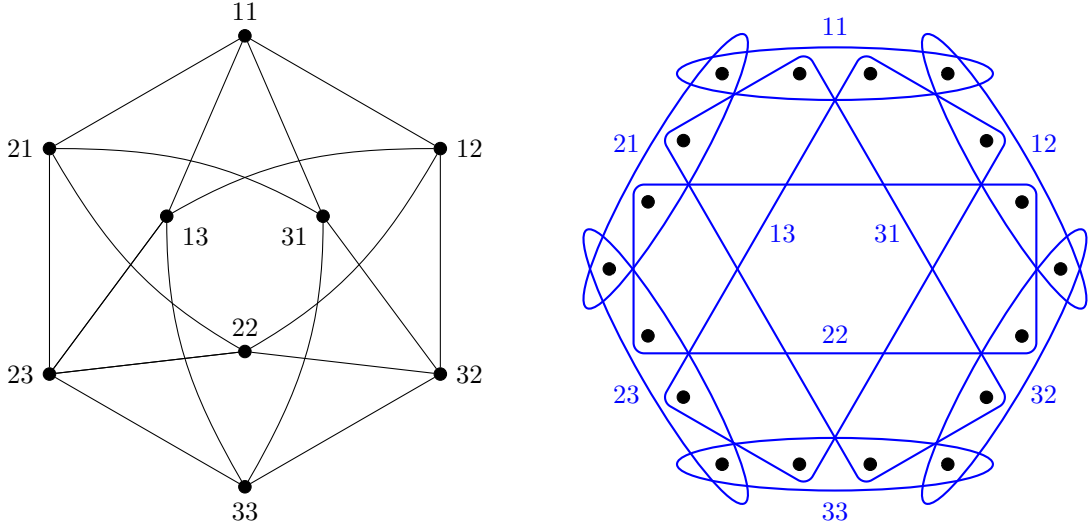
\begin{figure}
\subfigure[$GQ(2,1)$ with nodes labeled such that two nodes share an arc if and only if their labels differ in one position.]{
\label{GQ}
\begin{centering}
\begin{tikzpicture}
\node[draw,shape=circle,fill,scale=.5,label=right:$12$] (12) at (30:3) {} ;
\node[draw,shape=circle,fill,scale=.5,label=above:$11$] (11) at (90:3) {} ;
\node[draw,shape=circle,fill,scale=.5,label=left:$21$] (21) at (150:3) {} ;
\node[draw,shape=circle,fill,scale=.5,label=left:$23$] (23) at (210:3) {} ;
\node[draw,shape=circle,fill,scale=.5,label=below:$33$] (33) at (270:3) {} ;
\node[draw,shape=circle,fill,scale=.5,label=right:$32$] (32) at (330:3) {} ;
\node[draw,shape=circle,fill,scale=.5,label=200:$31$] (31) at (30:1.2) {} ;
\node[draw,shape=circle,fill,scale=.5,label=-20:$13$] (13) at (150:1.2) {} ;
\node[draw,shape=circle,fill,scale=.5,label=90:$22$] (22) at (270:1.2) {} ;
\draw (11) -- (12) -- (32) -- (33) -- (23) -- (21) -- (11) ;
\draw (13) -- (11) -- (31) ;
\draw (31) -- (32) -- (22) ;
\draw (22) -- (23) -- (13) ;
\draw (22) -- (23) -- (13) ;
\draw (33) edge[out=120,in=270] (13) edge[out=60,in=270] (31) ;
\draw (21) edge[out=0,in=150] (31) edge[out=300,in=150] (22) ;
\draw (12) edge[out=240,in=30] (22) edge[out=180,in=30] (13) ;
\end{tikzpicture}
\end{centering}}
\hspace{.5cm}
\subfigure[A redrawing of Figure~\ref{CKSfig} with edge labels corresponding to the node labels of~\subref{GQ}.]{
\label{CKSfig2}
\begin{centering}
\begin{tikzpicture}
\node[draw,shape=circle,fill,scale=.5] (a) at (0:3) {} ;
\node[draw,shape=circle,fill,scale=.5] (b) at (60:3) {} ;
\node[draw,shape=circle,fill,scale=.5] (c) at (120:3) {} ;
\node[draw,shape=circle,fill,scale=.5] (d) at (180:3) {} ;
\node[draw,shape=circle,fill,scale=.5] (e) at (240:3) {} ;
\node[draw,shape=circle,fill,scale=.5] (f) at (300:3) {} ;
\draw[white] (a) -- (b) node[pos=.333,draw,shape=circle,fill,scale=.5,black] (a1) {} node[pos=.5] (am) {} node[pos=.666,draw,shape=circle,fill,scale=.5,black] (a2) {} ;
\draw[white] (b) -- (c) node[pos=.333,draw,shape=circle,fill,scale=.5,black] (b1) {} node[pos=.5] (bm) {} node[pos=.666,draw,shape=circle,fill,scale=.5,black] (b2) {} ;
\draw[white] (c) -- (d) node[pos=.333,draw,shape=circle,fill,scale=.5,black] (c1) {} node[pos=.5] (cm) {} node[pos=.666,draw,shape=circle,fill,scale=.5,black] (c2) {} ;
\draw[white] (d) -- (e) node[pos=.333,draw,shape=circle,fill,scale=.5,black] (d1) {} node[pos=.5] (dm) {} node[pos=.666,draw,shape=circle,fill,scale=.5,black] (d2) {} ;
\draw[white] (e) -- (f) node[pos=.333,draw,shape=circle,fill,scale=.5,black] (e1) {} node[pos=.5] (em) {} node[pos=.666,draw,shape=circle,fill,scale=.5,black] (e2) {} ;
\draw[white] (f) -- (a) node[pos=.333,draw,shape=circle,fill,scale=.5,black] (f1) {} node[pos=.5] (fm) {} node[pos=.666,draw,shape=circle,fill,scale=.5,black] (f2) {} ;
\draw[thick,blue,rotate=120] (am) ellipse (2.1cm and .35cm) node [label={[label distance=.1cm]30:$12$}] {} ;
\draw[thick,blue,rotate=180] (bm) ellipse (2.1cm and .35cm) node [label={[label distance=.25cm]90:$11$}] {} ;
\draw[thick,blue,rotate=240] (cm) ellipse (2.1cm and .35cm) node [label={[label distance=.1cm]150:$21$}] {} ;
\draw[thick,blue,rotate=300] (dm) ellipse (2.1cm and .35cm) node [label={[label distance=.1cm]210:$23$}] {} ;
\draw[thick,blue,rotate=0] (em) ellipse (2.1cm and .35cm) node [label={[label distance=.25cm]270:$33$}] {} ;
\draw[thick,blue,rotate=60] (fm) ellipse (2.1cm and .35cm) node [label={[label distance=.1cm]330:$32$}] {} ;
\node (a3) at ($ (a2) + (10:.3)$) {} ;
\node (b3) at ($ (b2) + (70:.3)$) {} ;
\node (c3) at ($ (c2) + (130:.3)$) {} ;
\node (d3) at ($ (d2) + (190:.3)$) {} ;
\node (e3) at ($ (e2) + (250:.3)$) {} ;
\node (f3) at ($ (f2) + (310:.3)$) {} ;
\draw[thick,blue,rounded corners,rotate=150] (a3) rectangle (d3) node [pos=.5,label={[label distance=.3cm]150:$13$}] {} ;
\draw[thick,blue,rounded corners,rotate=30] (e3) rectangle (b3) node [pos=.5,label={[label distance=.3cm]30:$31$}] {} ;
\draw[thick,blue,rounded corners,rotate=270] (c3) rectangle (f3) node [pos=.5,label={[label distance=.5cm]270:$22$}] {} ;
\end{tikzpicture}
\end{centering}}
\caption{The generalized quadrangle graph $GQ(2,1)$ together with an illustration of the correspondence to Figure~\ref{CKSfig}: the arcs of~\subref{GQ} represent the vertices of~\subref{CKSfig2}, while all the arcs adjacent to a given node in~\subref{GQ} determine an edge in~\subref{CKSfig2}.}
\label{CKSGQ}
\end{figure}

\begin{cor}
$\CE^1(\mathrm{Mat}_5)$, when scaled by a factor of $2$, is the symmetric traveling salesman polytope $\mathrm{STSP}(5)$~\cite{GP,NP}.
\end{cor}

\begin{proof}
Since $5$ is odd, $K_5$ has no perfect matchings. Therefore, every half-integer matching on $K_m$ is a disjoint union of cycles of edges with weight $\tfrac{1}{2}$. Now it follows from~\ref{MatLO} that every extremal vertex of $\CE^1(\mathrm{Mat}_m)$ is a cycle of length $5$ with weight $\tfrac{1}{2}$ on each edge, or in other words the incidence vector of a traveling salesman tour scaled by a factor of $\tfrac{1}{2}$.
\end{proof}

Nothing in the definition of $\mathrm{Mat}_m$ is specific to $K_m$ and can likewise be done starting with any other graph. For example, if one starts with the generalized quadrangle graph $GQ(2,1)$~\cite{quadrangle} depicted in Figure~\ref{GQ}, then one obtains the Cabello--Kochen--Specker scenario shown in Figures~\ref{CKSfig} and~\ref{CKSfig2} by defining a vertex to be an arcs in Figure~\ref{GQ} and an edge to be determined by the four arcs incident at one node.

\appendix
\newpage

\newpage
\section{\textbf{Background on graph theory}}
\label{appcap}

\setcounter{theo}{0}

This section starts by reviewing standard material on the invariants of graphs which are of relevance to the main text, first for unweighted and then for weighted graphs, mostly without proofs. 

For us, a \emph{graph} is an undirected simple graph without isolated vertices. When $G$ is a graph, we denote its set of vertices by $V(G)$. For $u,v\in V(G)$, we write $u\sim_G v$ whenever $u$ and $v$ share an edge (are \emph{adjacent}) in $G$. Usually the graph $G$ is clear from the context, and then we simply write $u\sim v$.

There are many ways to take products of graphs~\cite{IK}. For us, the relevant one is this:

\begin{defn}
Let $G_1$ and $G_2$ be graphs. Their \emph{strong product} is the graph $G_1\boxtimes G_2$ with
$$
V(G_1\boxtimes G_2) \defin V(G_1) \times V(G_2) 
$$
and $(u_1,u_2)\sim (v_1,v_2)$ whenever 
$$
\left( u_1\sim v_1 \land u_2\sim v_2 \right) \lor \left( u_1\sim v_1 \land u_2=v_2 \right) \lor \left( u_1=v_1 \land u_2\sim v_2 \right) .
$$
\end{defn}

This rule for when $(u_1,u_2)\sim (v_2,v_2)$ can be intuitively understood if one thinks of $G_1$ and $G_2$ as \emph{confusability graphs} whose vertices describe items which can be confused with each other whenever they share an edge~\cite{Shannon}. Then, a pair of items $(u_1,u_2)$ can be confused with a pair of items $(v_1,v_2)$ if $u_1$ can be confused with $v_1$ or $u_1=v_1$, and $u_2$ can be confused with $v_2$ or $u_2=v_2$.

For $n\in\N$, we write $G^{\boxtimes n}$ for the $n$-fold strong product of $G$ with itself. Just as in Shannon's application to error-free communication over noisy communication channels~\cite{Shannon}, these strong powers are also of importance to us.

\subsection{Relevant invariants of unweighted graphs}\label{ap:invunw} Since we will later consider graphs equipped with vertex weights, we also use the term `unweighted graph' when working with ordinary graphs in order to emphasize the distinction.

An \emph{independent set} in a graph $G$ is a subset $I\subseteq V(G)$ such that no two vertices in $I$ share an edge. $I$ is an independent set in $G$ if and only if it is a \emph{clique} in the complement graph $\overline{G}$. An independent set $I$ is \emph{maximal} if there is no other independent set $I'$ with $I\subsetneq I'$. The \emph{independence number} $\alpha(G)$ is the largest number of elements in any independent set of $G$; while the independence number can be attained only by a maximal independent set, there may also exist maximal independent sets of smaller cardinality. The independence number is sometimes also called the \emph{stability number}.

\begin{lem}
\label{prodin}
Let $I_1\subseteq G_1$ and $I_2\subseteq G_2$ be maximal independent sets. Then $I_1\times I_2\subseteq G_1\boxtimes G_2$ is also a maximal independent set.
\end{lem}

\begin{proof}
The definition of adjacency in $G_1\boxtimes G_2$ implies immediately that $I_1 \times I_2$ is also an independent set in $G_1\boxtimes G_2$.

We now show maximality of $I=I_1\boxtimes I_2$. For any $v=(v_1,v_2)\in V(G_1\boxtimes G_2)\setminus I$, the following cases are possible:
\begin{enumerate}
\item Case $v_1 \notin I_1$ and $v_2\notin I_2$: by maximality of $I_1$ and $I_2$, there are $u_1\in I_1$ with $u_1\sim v_1$ and $u_2\in I_2$ with $u_2\sim v_2$. Hence $(u_1,u_2)\sim(v_1,v_2)$.
\item Case $v_1 \notin I_1$ and $v_2 \in I_2$: by maximality of $I_1$, there is $u_1 \in I_1$ with $u_1\sim v_1$. Hence $(u_1,v_2)\in I$ and $(u_1,v_2)\sim (v_1,v_2)$.
\item Case $v_1 \in I_1$ and $v_2 \notin I_2$: Similar to the previous case.
\end{enumerate}
In either case, the conclusion is that $v$ is adjacent to some vertex in $I$. Since $v$ was arbitrary, this means that $I$ is a maximal independent set.
\end{proof}

Concerning the independence number, the fact that a product of independent sets is again an independent set immediately shows:

\begin{lemma}
\label{in-ineq}
$$
\alpha(G_1\boxtimes G_2) \geq \alpha(G_1) \alpha(G_2)
$$
\end{lemma}

In particular, this implies for the independence number of the strong powers of a graph,
\beq
\label{fekete}
\alpha(G^{\boxtimes(n+m)}) \geq \alpha(G^{\boxtimes n}) \alpha(G^{\boxtimes m}) \qquad \forall m,n\in\N .
\eeq

\begin{rem}
\label{notmonotone}
Despite this inequality, the sequence $\left(\sqrt[n]{\alpha(G^{\boxtimes n})}\right)_{n\in \N}$ is not monotonically increasing in general; this happens, for example, for the pentagon graph (or $5$-cycle) $\pentagon$, for which~\cite{Shannon}
$$
\alpha(\pentagon) = 2,\qquad \alpha(\pentagon^{\boxtimes 2}) = 5,\qquad \alpha(\pentagon^{\boxtimes 3}) = 10 . 
$$
See~\cite{AL} for more results on the sometimes counterintuitive behavior of $\left(\sqrt[n]{\alpha(G^{\boxtimes n})}\right)_{n\in \N}$. 
\end{rem}

In combination with Fekete's Lemma~\cite{Fekete}, the inequality~(\ref{fekete}) guarantees the existence of the following limit:

\begin{defn}[\cite{Shannon}]
The \emph{(unweighted) Shannon capacity} $\Theta(G)$ is
\beq
\label{uncapdef}
\Theta(G) \defin \lim_{n\to\infty} \sqrt[n]{\alpha(G^{\boxtimes n})} .
\eeq
\end{defn}

Intuitively, $\Theta(G)$ is an asymptotic version of the independence number $\alpha(G)$. It was originally introduced by Shannon in an information-theoretic context as an effective number of symbols which can be transmitted across a noisy communication channel without the possibility of error, using the confusability graph interpretation explained above.

Not much is known about the values of $\Theta$ for particular graphs, not even $\Theta(C_7)$, where $C_7$ is the $7$-cycle~\cite{CGR}.

For graphs $G_1$ and $G_2$, we write $G_1+G_2$ for their disjoint union, which is again a graph.

\begin{lem}[\cite{Shannon}]
\label{thetasumprod}
\begin{enumerate}
\item $\Theta(G_1\boxtimes G_2) \geq \Theta(G_1) \Theta(G_2)$.
\item $\Theta(G_1 + G_2) \geq \Theta(G_1) + \Theta(G_2)$.
\end{enumerate}
\end{lem}

Finding examples in which these inequalities are not tight is surprisingly difficult. The following results are due to Haemers and Alon.

\begin{thm}[\cites{Haemers,Alon}]
\label{haemersalon}
There exist graphs $G_1$ and $G_2$ such that
\begin{enumerate}
\item $\Theta(G_1 \boxtimes G_2) > \Theta(G_1) \Theta(G_2)$,
\item $\Theta(G_1 + G_2) > \Theta(G_1) + \Theta(G_2)$.
\end{enumerate}
\end{thm}

This ends our short summary of the Shannon capacity of graphs, and we now move on to another intimately related graph invariant.

\begin{defn}[\cite{Lovasz}]
\begin{enumerate}
\item\label{orthlab} An \emph{orthonormal labeling} of $G$ is an assignment $v\mapsto |\psi_v\rangle$ of a unit vector $|\psi_v\rangle\in\R^{|V(G)|}$ to every $v\in V(G)$ such that $u\not\sim v$ and $u\neq v$ implies $|\psi_u\rangle\perp|\psi_v\rangle$.
\item The \emph{Lov{\'a}sz number} $\vartheta(G)$ is
$$
\vartheta(G) \defin \min_{|\Psi\rangle,\:|\psi_v\rangle} \max_{v\in V} \frac{1}{|\langle\Psi|\psi_v\rangle|^2}
$$
where $|\Psi\rangle\in\R^{|V(G)|}$ ranges over all unit vectors and $(|\psi_v\rangle)_{v\in V(G)}$ over all orthonormal labelings.
\end{enumerate}
\end{defn}

There are several other equivalent definitions of $\vartheta(G)$ commonly used~\cite{Lovasz}, one of which we will meet in Proposition~\ref{thetaalt} for the weighted Lov{\'a}sz number. 

Multiplicativity of $\vartheta$ is one of its many useful properties:

\begin{prop}[\cite{Lovasz}]
\label{lovaszmult}
$$
\vartheta(G_1\boxtimes G_2) = \vartheta(G_1)\vartheta(G_2) .
$$
\end{prop}

As the fourth and final graph invariant of relevance to us, we now introduce:

\begin{defn}
The \emph{fractional packing number} $\alpha^*(G)$ is
$$
\alpha^*(G) \defin \max_q \sum_v q(v)
$$
where $q : V(G) \to [0,1]$ ranges over all vertex weighings satisfying $\sum_{v\in C} q(v) \leq 1$ for all cliques $C\subseteq V(G)$.
\end{defn}

The fractional packing number can be regarded as the linear relaxation of the independence number. For this reason, it is sometimes also called \emph{fractional independence number}.

\begin{prop}[\cite{Lovasz}]
\label{sandwich}
$$
\alpha(G) \leq \Theta(G) \leq \vartheta(G) \leq \alpha^*(G) .
$$
\end{prop}
In general, none of these inequalities is an equality. This is most difficult to see for $\Theta(G)\leq\vartheta(G)$, for which it was shown by Haemers~\cite{Haemers} after having been posed as an open problem by Lov{\'a}sz~\cite{Lovasz}.

\subsection{Relevant invariants of weighted graphs}\label{ap:invw}

We now generalize these definitions to graphs equipped with vertex weights, i.e.~to graphs $G$ equipped with a \emph{weight function} $p:V(G)\to \R_+$. We omit a proof whenever it is completely analogous to the unweighted case. Two weight functions $p_1:V(G_1)\to \R_+$ and $p_2:V(G_2)\to \R_+$ can be tensored to a new weight function on the strong product graph,
$$
p_1\otimes p_2 : V(G_1\boxtimes G_2)\to \R_+,\quad (v_1,v_2)\mapsto p_1(v_1)p_2(v_2).
$$
In this way, the $n$th power $p^{\otimes n}$ is a weight function on $G^{\boxtimes n}$. Similarly, there is an obvious weight function $p_1+p_2$ defined on the disjoint union $G_1+G_2$. When $p_1$ and $p_2$ are defined on the same graph, we use the same notation $p_1+p_2$ for the pointwise sum; despite this ambiguous notation, the meaning will always be clear from the context.

For each item in the following definition, it should be clear how it generalizes the concepts from the unweighted setting in the sense that the previous definitions are recovered if the weights are $p(v)=1$ for all vertices $v$.

\begin{defn}
\label{wdefs}
Let $G$ be a graph equipped with vertex weights $p$.
\begin{enumerate}
\item The \emph{weighted independence number} $\alpha(G,p)$ is the largest total weight of an independent set in $G$, that is the largest sum of weights of elements of an independent set.  
\item The \emph{weighted Lov{\'a}sz number} $\vartheta(G,p)$ is
\beq
\label{wlovdef}
\vartheta(G,p) \defin \min_{|\Psi\rangle,\:|\psi_v\rangle} \max_{v\in V} \frac{p(v)}{|\langle\Psi|\psi_v\rangle|^2}
\eeq
where $|\Psi\rangle\in\R^{|V(G)|}$ ranges over all unit vectors and $(|\psi_v\rangle)_{v\in V(G)}$ over all orthonormal labelings.
\item The \emph{weighted Shannon capacity} $\Theta(G,p)$ is
\beq
\label{wcapdef}
\Theta(G,p) = \lim_{n\to\infty} \sqrt[n]{\alpha(G^{\boxtimes n},p^{\otimes n})} .
\eeq
\item The \emph{weighted fractional packing number} $\alpha^\ast(G,p)$ is
$$
\alpha^*(G,p) \defin \max_q \sum_{v\in V} p(v)\, q(v) .
$$
where $q:V(G)\to\R_+$ ranges over all vertex weights satisfying $\sum_{v\in C} q(v) \leq 1$ for all cliques $C\subseteq V(G)$.
\end{enumerate}
\end{defn}

The fraction in~(\ref{wlovdef}) uses the convention $\tfrac{0}{0}=0$. See~\cite{Knuth} for several equivalent definitions of $\vartheta(G,p)$, in particular this one:

\begin{prop}[{\cite[Sec.~10]{Knuth}}]
\label{thetaalt}
The Lov{\'a}sz number is also given by
\[
\vartheta(G,p) = \max_{|\Psi\rangle,\:|\psi_v\rangle} \sum_{v\in V(G)} p(v)\:|\langle\Psi|\psi_v\rangle|^2
\]
where $|\Psi\rangle\in\R^{|V(G)|}$ still ranges over all unit vectors, but $(|\psi_v\rangle)_{v\in V(G)}$ ranges over all orthonormal labelings of the complementary graph $\overline{G}$.
\end{prop}

Among the useful consequences of this result is the following:

\begin{lem}
\label{thetaconv}
$\vartheta(G,p_1 + p_2) \leq \vartheta(G,p_1) + \vartheta(G,p_2).$
\end{lem}

The fractional packing number can alternatively be characterized as follows:

\begin{prop}
\label{fpndualprop}
Let $\mathrm{Cl}(G)$ denote the set of all cliques on $G$. Then
\beq
\label{fpndual}
\alpha^*(G,p) = \min_{x} \sum_{C\in\mathrm{Cl}(G)} x(C)
\eeq
where $x$ ranges over all functions $x:\mathrm{Cl}(G)\to\R_+$ with $p(v)\leq \sum_{C\ni v} x(C) \:\forall v$.
\end{prop}

\begin{proof}
Linear programming duality.
\end{proof}

We also have a generalization of Lemma~\ref{thetasumprod}:

\begin{lem}\label{proptheta}
\begin{enumerate}
\item 
\beq
\label{superadd}
\Theta(G_1 + G_2, p_1 + p_2) \geq \Theta(G_1,p_1) + \Theta(G_2,p_2) .
\eeq
\item 
\beq
\Theta(G_1\boxtimes G_2,p_1\otimes p_2) \geq \Theta(G_1,p_1) \Theta(G_2,p_2) .
\eeq
\end{enumerate}
\end{lem}

\begin{proof}
As in the unweighted case~\cite{Shannon}.
\end{proof}

Since, as remarked earlier, these inequalities are not tight in general even in the unweighted case, neither can they be tight in the weighted case. One might expect simpler counterexamples to exist in the weighted case, but we have not been successful in finding any.

When $p_1$, $p_2$ are weight functions on the same graph $G$, a superadditivity inequality no longer holds for trivial reasons. For example for $G=K_2$, the graph on two adjacent vertices $\{u,v\}$, equipped with indicator functions $p_1 = \mathbbm{1}_{u}$ and $p_2 = \mathbbm{1}_{v}$, we have
$$
1 = \Theta(G,p_1+p_2) < \Theta(G,p_1) + \Theta(G,p_2) = 2 .
$$

Many statements about these weighted invariants can be reduced to statements about their unweighted counterparts using a technique which we call \emph{blow-up}. Applying this technique requires the vertex weights to be rational. Therefore, we begin by proving a continuity result which allows us to reduce many problems to the case of rational weights.

\begin{lem}
\label{addempty}
Let $(G,p)$ be a weighted graph and $\overline{K}_m$ the empty graph on $m$ vertices with arbitrary weights $q$. Then,
\beq
X(G + \overline{K}_m, p + q) = X(G,p) + \sum_{v\in V(\overline{K}_m)} q(v)
\eeq
for all four invariants $X\in\{\alpha,\Theta,\vartheta,\alpha^*\}$.
\end{lem}

\begin{proof}
This is trivial for $X=\alpha$, since the maximal independent sets in $G+\overline{K}_m$ are those of the form $I+\overline{K}_m$, where $I$ is a maximal independent set in $G$. For $X=\vartheta$, the claim is a special case of~\cite[eq.~(18.2)]{Knuth}. For $X=\alpha^*$, it follows from an application of Proposition~\ref{fpndualprop}, since a clique in $G+\overline{K}_m$ is a clique in $G$ or a single vertex in $\overline{K}_m$. So it remains to treat the case $X=\Theta$.

Since $\Theta(\overline{K}_m,q) = \sum_v q(v)$, the inequality `$\geq$' is an instance of superadditivity~(\ref{superadd}) of $\Theta$. To also show `$\leq$', we choose any independent set $I$ in $(G+\overline{K}_m)^{\boxtimes n}$ and partition it into a disjoint union
$$
I = \bigcup_{\vec{s}\in\{0,1\}^n} I_{\vec{s}}
$$
where each $I_{\vec{s}}$ contains only vertices $(v_1,\ldots,v_n)$ with $v_i\in V(G)$ if $s_i=0$ and $v_i\in V(\overline{K}_m)$ if $s_i=1$. Then upon dropping all components $i$ with $s_i=1$, such an $I_{\vec{s}}$ becomes an independent set in some $G^{\boxtimes k}$. In this way, we obtain the estimate
\begin{align*}
\alpha\left( (G+\overline{K}_m)^{\boxtimes n}, (p+q)^{\otimes n}\right) & \leq \sum_{k=0}^n \binom{n}{k} \alpha(G^{\boxtimes k},p^{\otimes k}) \left(\sum_i q_i\right)^{n-k} \\
& \leq \sum_{k=0}^n \binom{n}{k} \Theta(G,p)^k \left(\sum_i q_i\right)^{n-k} = \left(\Theta(G,p) + \sum_i q_i\right)^n ,
\end{align*}
which implies the desired inequality upon taking the $n$-th root and then $n\to\infty$.
\end{proof}

\begin{lem}
\label{incrweight}
Let $(G,p)$ be a weighted graph, $v\in G$ a vertex, $\mu\in\R_+$ and $X\in\{\alpha,\Theta,\vartheta,\alpha^*\}$. Then 
\beq
X(G,p) \leq X(G,p + \mu\mathbbm{1}_v) \leq X(G,p) + \mu .
\eeq
\end{lem}

\begin{proof}
The first inequality is clear since $X(G,p)$ is a non-decreasing function of $p$.

Since adding additional edges cannot increase the value of $X$ and two vertices with exactly the same neighbors can be identified to one vertex by adding the weights (for $X=\vartheta$, see~\cite[Lemma~16]{Knuth}), we have $X(G,p+q\mathbbm{1}_v) \leq  X(G + \overline{K}_1,p+q)$. Now the second inequality follows from the previous lemma with $m=1$.
\end{proof}

This lemma directly gives the desired continuity result:

\begin{cor}
\label{thetacont}
For any graph $G$ and any $X\in\{\alpha,\Theta,\vartheta,\alpha^*\}$, the function $p\mapsto X(G,p)$ is continuous.
\end{cor}

We can now introduce the blow-up technique which can be used to translate problems from the weighted case to the unweighted setting. 

\begin{defn}
Let $(G,p)$ be a weighted graph with $p(v)\in\N$ for all $v$. Then the \emph{blow-up} $\mathrm{Blup}(G,p)$ is the unweighted graph with vertex set
$$
\big\{(v,k) \::\: v\in G,\, k\in \{1,\ldots,p(v)\} \, \big\} ,
$$
where we take $(v,k)$ and $(v',k')$ to be adjacent if and only if $v\sim v'$ in $G$.
\end{defn}

Intuitively speaking, $\mathrm{Blup}(G,p)$ is constructed by replacing every vertex $v$ in $G$ by $p(v)$ many non-adjacent vertices. In particular, if $p(v)=0$, the vertex $v$ simply gets removed from the graph. Blow-ups have also been considered in~\cite[Sec.~16]{Knuth}, although not under that name.

\begin{lem}
For vertex weights in $\N$,
\label{blup}
\begin{enumerate}
\item $\mathrm{Blup}(G_1 + G_2, p_1 + p_2) = \mathrm{Blup}(G_1,p_1) + \mathrm{Blup}(G_2,p_2)$.
\item $\mathrm{Blup}(G_1\boxtimes G_2,p_1\otimes p_2) = \mathrm{Blup}(G_1,p_1)\boxtimes \mathrm{Blup}(G_2,p_2)$;
\item $X(\mathrm{Blup}(G,p)) = X(G,p)$ for every $X\in\{\alpha,\Theta,\vartheta,\alpha^*\}$.
\end{enumerate}
\end{lem}

\begin{proof}
Straightforward.
\end{proof}

We can now already reap some of the simpler benefits of these results. By the continuity statement of Corollary~\ref{thetacont}, it is sufficient for the proof of many statements to consider rational weights. In this case, one can often restrict to natural number weights without loss of generality by rescaling all weights by the smallest common denominator, and then the blow-up technique applies. This yields the following list of consequences of results already derived for unweighted graphs:

\begin{cor}
\label{monotone}
$$
\alpha(G,p) \leq \Theta(G,p) \leq \vartheta(G,p) \leq \alpha^*(G,p) .
$$
\end{cor}

\begin{proof}
Proposition~\ref{sandwich}.
\end{proof}

\begin{cor}[{\cite[(20.5)]{Knuth}}]
\label{lovmult}
$$
\vartheta(G_1\boxtimes G_2,p_1\otimes p_2) = \vartheta(G_1,p_1)\vartheta(G_2,p_2)
$$
\end{cor}

\begin{proof}
Proposition~\ref{lovaszmult}.
\end{proof}

\begin{cor}
\label{win-ineq}
$$
\alpha(G_1\boxtimes G_2,p_1\otimes p_2) \geq \alpha(G_1,p_1) \alpha(G_2,p_2)
$$
\end{cor}

\begin{proof}
Lemma~\ref{in-ineq}.
\end{proof}

\newpage
\section{\textbf{Results on graphs with Shannon capacity $=$ independence number}}
\label{uwmainconjs}

\setcounter{theo}{0}

Here we study graphs whose Shannon capacity coincides with their independence number, which arise in Section~\ref{LOsec}.

\begin{defn}
An unweighted graph $G$ is \emph{single-shot} if $\alpha(G) = \Theta(G)$.
\end{defn}

In this section, we will study single-shot graphs in some detail and exploit some of the results from the main text for deriving a stronger version of Theorem~\ref{haemersalon}, namely that the phenomena $\Theta(G_1\boxtimes G_2) > \Theta(G_1)\Theta(G_2)$ and $\Theta(G_1+G_2) > \Theta(G_1) + \Theta(G_2)$ do occur even within the class of single-shot graphs. As far as we know, these results are new.

Single-shot graphs are the \emphalt{Class 1} graphs of Berge~\cite{Berge}\footnote{We thank Andr{\'a}s Salamon for pointing out this reference.}. $G$ is single-shot precisely when the sequence $\left(\sqrt[n]{\alpha(G^{\boxtimes n})}\right)_{n\in\N}$ is constant. Our terminology is motivated by the information-theoretic interpretation alluded to in Appendix~\ref{appcap}: if a communication channel has a confusability graph which is single-shot, then there exists a zero-error code for this channel which operates on the single-shot level.

A well-known class of single-shot graphs are the \emph{perfect graphs} (see Section~\ref{perfectionsection}), which have the property that $\alpha(G)=\alpha^*(G)$~\cite{Knuth}. On the other hand, the Petersen graph is not perfect, but nevertheless single-shot since its Lov{\'a}sz number coincides with its independence number~\cite[p.~31]{Knuth}.

In order to use our theorem from the main text to deduce results about single-shot graphs, we need to consider weighted graphs first:

\begin{defn}
A weighted graph $(G,p)$ is \emph{single-shot} if $\alpha(G,p) = \Theta(G,p)$.
\end{defn}

For weighted graphs, our main result is this:

\begin{thm}
\label{wthm}
There exist weighted single-shot graphs $(G_1,p_1)$ and $(G_2,p_2)$ such that
\[
\Theta(G_1\boxtimes G_2, p_1\otimes p_2) > \Theta(G_1,p_1) \Theta(G_2,p_2) 
\]
and
\[
\Theta(G_1 + G_2, p_1 + p_2) > \Theta(G_1,p_1) + \Theta(G_2,p_2).
\]
It is even possible to have $\alpha(G_2,p_2) = \vartheta(G_2,p_2)$.
\end{thm}

\begin{proof}
The first part follows from Theorem~\ref{CEactivate}: we put $G_1\defin\mathrm{NO}(H_A)$ and $G_2\defin\mathrm{NO}(H_B)$ and equip these graphs with vertex weights coming from probabilistic models $p_1\in\CE^\infty(H_A)$ and $p_2\in\CE^\infty(H_B)$ with $p_1\otimes p_2\not\in\CE^1(H_A\otimes H_B)$. Then we have
\[
\Theta(G_1,p_1) = \alpha(G_1,p_1) = 1, \qquad \Theta(G_2,p_2) = \alpha(G_2,p_2) = 1,
\]
by Lemma~\ref{LOchar}, so that both $(G_1,p_1)$ and $(G_2,p_2)$ are single-shot. On the other hand, the assumption $p_1\otimes p_2\not\in\CE^1(H_A\otimes H_B)$ implies that
\[
\alpha(G_1\boxtimes G_2,p_1\otimes p_2) > 1 = \Theta(G_1,p_1)\Theta(G_2,p_2).
\]
This is stronger than what we have to prove.

The second claim follows from analogous reasoning, applied to the construction in the proof of Theorem~\ref{nonconvexthm}.

Finally, since we had $p_2\in\mathcal{Q}(H_B)\subseteq\mathcal{Q}_1(H_B)$, we actually also have $\vartheta(G_2,p_2)=1$ by Proposition~\ref{Q1vsLov}.
\end{proof}

This proof is completely constructive: as a particular example for the phenomenon of Theorem~\ref{CEactivate}, we may take the scenarios and probabilistic models discussed after its proof. So $G_1$ will have $220$ vertices corresponding to the $3$-element subsets of $\{1,\ldots,12\}$, and two such $v,w\in V(G_1)$ are adjacent if and only if $|v\cap w|=1$. We can rescale the weights so that they are all $1$ and $G_1$ is effectively unweighted. We have
\[
\alpha(G_1) = \Theta(G_1) = 12 < \frac{260}{11} = \vartheta(G_1).
\]
The other graph $G_2$ turns out to have $\numprint{12100}$ vertices of two kinds:
\begin{enumerate}
\item again the $3$-element subsets of $\{1,\ldots,12\}$, of which there are $220$.
\item unordered pairs of such $3$-element subsets having one element in common, of which there are $\numprint{11880}$.
\end{enumerate}
Two vertices $v$ and $w$ of the first kind are adjacent if and only if $|v\cap w| \in\{0,2\}$. A vertex $v$ of the first kind and a vertex $\{w,w'\}$ of the second kind are adjacent if and only if $v\neq w$ and $v\neq w'$. Finally, any two vertices $\{v,v'\}$ and $\{w,w'\}$ of the second kind are adjacent. Concerning the weights, all vertices of the first have a weight of $\tfrac{13}{121}$, while of the second kind carry the higher weight $\tfrac{95}{121}$.

Another result which is easily proved along the same lines as Theorem~\ref{wthm} is this:

\begin{thm}
\label{graphconvex}
There is a graph $G$ with weight functions $p_1$ and $p_2$ such that both $(G,p_1)$ and $(G,p_2)$ are single-shot, and
$$
\Theta(G,p_1 + p_2) >\Theta(G,p_1) + \Theta(G,p_2) ,
$$
meaning that $(G,p_1 + p_2)$ is not single-shot.
\end{thm}

This time, the proof builds on Theorem~\ref{nonconvexthm} and its proof. Alternatively, one can also translate the proof of Theorem~\ref{nonconvexthm} into a derivation of the present result from Theorem~\ref{wthm}, so that taking $G:=G_1+G_2$ will work. In terms of our explicit example, this will be a graph on $220+\numprint{12100}=\numprint{12320}$ vertices.

In order to transfer Theorem~\ref{wthm} to the case of unweighted graphs, we need to apply the blow-up technique of Appendix~\ref{ap:invw}. However, doing this requires showing that the vertex weights of a single-shot graph can be turned into rational numbers by a small perturbation in such a way that the perturbed graph is still single-shot:

\begin{lem}
\label{rationalize}
Let $(G,p)$ be a weighted single-shot graph. Then for every $\eps>0$ there exist weights $p'(v)\in\Q$ with $|p(v) - p'(v)| < \eps$ and such that $(G,p')$ is still single-shot with $\alpha(G,p') = \alpha(G,p)$.

The same perturbation guarantees that if $\alpha(G,p)=\vartheta(G,p)$, then $\alpha(G,p')=\vartheta(G,p')=\alpha(G,p)$.
\end{lem}

\begin{proof}
Let $p_\mathrm{max}$ be the largest weight of a vertex in $G$, and fix $\delta>0$ such that $2\delta \cdot p_\mathrm{max} \leq \eps$.
Fix any independent set $v_1,\ldots,v_n$ of maximal weight and choose rational numbers $p'(v_i)\in \big( (1-\delta )p(v_i), (1+\delta) p(v_i)\big)$ such that $\sum_i p'(v_i) = \sum_i p_i = \alpha(G,p)$. Furthermore, for vertices $w$ not in that set, choose arbitrary rational numbers $p'(w)\in\big(  (1-2\delta)  p(w), (1-\delta) p(w) \big)$. Then $2\delta\cdot p_\mathrm{max} \leq \eps$ guarantees $|p(v)-p'(v)|<\eps$ for all $v\in V(G)$.

Now we claim that $\alpha(G,p') = \Theta(G,p') = \alpha(G,p)$. Upon setting $q_i \defin p'(v_i) - (1-\delta) p(v_i)$, we estimate
$$
\alpha(G,p') \leq \Theta(G,p') \leq \Theta\left(G,(1-\delta) p\right) + \sum_i q_i ,
$$
where the last inequality follows from Lemma~\ref{addempty} and the fact that transporting some weight from some vertex to a new isolated vertex cannot decrease the capacity. Since $\sum_i q_i = \alpha(G,p) - (1-\delta) \alpha(G,p)$, we can further evaluate this to
$$
\alpha(G,p') \leq \Theta(G,p') \leq (1-\delta) \Theta(G,p) + \delta \alpha(G,p) = \alpha(G,p) .
$$
On the other hand, we have constructed $p'$ in such a way that there is an independent set of weight $\alpha(G,p)$, and hence all these inequalities are actually equalities.

The same argument can be applied with $\vartheta$ in place of $\Theta$.
\end{proof}

Now we can use our techniques to turn Theorem~\ref{wthm} into a statement about unweighted graphs:

\begin{thm}
\label{uwthm}
There exist single-shot graphs $G_1$ and $G_2$ such that
\[
\Theta(G_1\boxtimes G_2) > \Theta(G_1) \Theta(G_2) 
\]
and
\[
\Theta(G_1 + G_2) > \Theta(G_1) + \Theta(G_2).
\]
It is even possible to have $\alpha(G_2) = \vartheta(G_2)$.
\end{thm}

\begin{proof}
Take weighted graphs $(G_1,p_1)$ and $(G_2,p_2)$ as in Theorem~\ref{wthm} and apply Lemma~\ref{rationalize} to both $(G_j,p_j)$ with a certain $\eps>0$ and obtain $(G_j,p'_j)$. Then, the differences
\[
(p'_1\otimes p'_2)(v_1,v_2) - (p_1\otimes p_2)(v_1,v_2)
\]
can be bounded by a certain function of $\eps$ and the $\alpha(G_j,p_j)$'s which converges to $0$ as $\eps\to 0$. In particular, the continuity statement of Corollary~\ref{thetacont} guarantees that one can choose $\eps$ so small that the inequalities
\[
\Theta(G_1\boxtimes G_2,p_1\otimes p_2) > \Theta(G_1,p_1)\Theta(G_2,p_2),\qquad \Theta(G_1 + G_2, p_1 + p_2) > \Theta(G_1,p_1) + \Theta(G_2,p_2)
\]
are preserved in the sense that
\[
\Theta(G_1\boxtimes G_2,p'_1\otimes p'_2) > \Theta(G_1,p_1)\Theta(G_2,p_2),\qquad \Theta(G_1 + G_2, p'_1 + p'_2) > \Theta(G_1,p_1) + \Theta(G_2,p_2).
\]
The statement of Lemma~\ref{rationalize} guarantees that $\Theta(G_j,p_j)=\Theta(G_j,p'_j)$, and therefore we actually have
\[
\Theta(G_1\boxtimes G_2,p'_1\otimes p'_2) > \Theta(G_1,p'_1)\Theta(G_2,p'_2),\qquad \Theta(G_1 + G_2, p'_1 + p'_2) > \Theta(G_1,p'_1) + \Theta(G_2,p'_2).
\]
After multiplying each weight functions $p'_j$ by the respective common denominator, it becomes integer-valued, and the claim then follows from the blow-up Lemma~\ref{blup}.
\end{proof}

Again, this proof can be turned into an explicit construction, building upon the previous explicit example of $(G_1,p_1)$ and $(G_2,p_2)$. Since in this case, all weights are rational, the continuity considerations are redundant, and an application of the blow-up Lemma~\ref{blup} is sufficient. So the previous weights $p_1$ already were all $1$, no blow-up needs to be applied to $G_1$, and we have $\alpha(G_1)=\Theta(G_1)=12$. For $G_2$, every vertex of the first kind should be replaced by $13$ copies, while each vertex of the second kind turns into $95$ new vertices. Hence the resulting $G_2$ has
\[
13\cdot 220 + 95\cdot \numprint{11880} = \numprint{1131460}
\]
many vertices and satisfies $\alpha(G_2)=\vartheta(G_2)=121$. Translating the previous proofs into this language shows indeed that
\[
\Theta(G_1\boxtimes G_2) \geq \alpha(G_1\boxtimes G_2) \geq 121\cdot\vartheta(G_1) = 121\cdot \frac{260}{11} = 2860 > 12\cdot 121.
\]
It might be an interesting challenge to find smaller examples of the phenomenon of Theorem~\ref{uwthm}.

\newpage
\section{\textbf{Virtual edges, completeness, and product scenarios}}
\label{multiproducts}

\subsection{Virtual edges}
\label{virtualcomp}

Here, we introduce some background material relevant for the discussion of products of three or more contextuality scenarios in Section~\ref{higherprodssub}. The following material is based on the notion of \emph{perspectivity} of Foulis and Randall~\cite{FR81,FPR}, although the details are different. Readers not interested in the subtleties of products of more than two scenarios can safely ignore this section.

We start with an arbitrary contextuality scenario $H$. We write $2^{V(H)}$ for the power set of $V(H)$, i.e.~the set of all subsets of $V(H)$. The symbol `$\dot{\cup}$' stands for disjoint union, i.e.~$A=B \disjcup C$ means that $B\cap C = \emptyset$ and $B\cup C = A$.

\renewcommand{\labelenumi}{(\roman{enumi})} 
\renewcommand{\theenumi}{(\roman{enumi})} 

\begin{defn}
\label{simdef}
Let $\simeq$ be the smallest equivalence relation closed under the following rules:
\begin{enumerate}
\item\label{edgesim} If $A,B\in E(H)$, then $A\simeq B$.
\item\label{unionsim} If 
\[
A = A_1\disjcup\ldots\disjcup A_n, \qquad B = B_1\disjcup\ldots\disjcup B_n
\]
are partitions such that $A_i\simeq B_i$ for all $i$, then $A\simeq B$.
\item\label{decompsim} Conversely, if
\[
A=A'\disjcup C,\qquad B = B'\disjcup C,
\]
are partitions with $A\simeq B$, then also $A'\simeq B'$.
\end{enumerate}
\end{defn}

\renewcommand{\labelenumi}{(\alph{enumi})} 
\renewcommand{\theenumi}{(\alph{enumi})} 

In other words, we have $A\simeq B$ if and only if this can be derived from the rules~\ref{edgesim}--\ref{decompsim}, together with reflexivity, symmetry and transitivity of $\simeq$, in a finite number of steps. If $A\simeq B$, then we also say that $A$ and $B$ are \emph{equivalent}.

Instead of postulating rule~\ref{unionsim} for all $n$, it would be sufficient to do so for $n=2$, from which the general case follows.

\begin{ex}
Foulis and Randall~\cite{FR81,FPR} call two sets of vertices $A,B\subseteq V(H)$ \emph{in perspective} if there is $C\subseteq V(H)$ with $A\cap C = B\cap C = \emptyset$ such that $A\cup C\in E(H)$ and $B\cup C\in E(H)$. Rules~\ref{edgesim} and~\ref{decompsim} show that if $A$ and $B$ are in perspective, then $A\simeq B$. 
\end{ex}

\begin{ex}
In Figure~\ref{funnyvirt}, with the solid lines indicating the edges, we have $\{v_1,v_2,v_3\}\simeq\{w_1,w_2\}$. 
\end{ex}

The idea behind the previous definition is that $A\simeq B$ guarantees that $A$ and $B$ have equal probability under any probabilistic model:

\begin{prop}
\label{equalprobprop}
If $A\simeq B$ and $p\in\mathcal{G}(H)$, then
\beq
\label{equalprobeq}
\sum_{v\in A} p(v) = \sum_{v\in B} p(v).
\eeq
\end{prop}

\begin{proof}
We use induction on the number of steps in which the given equivalence $A\simeq B$ is derived. The last step involved may be one of the rules~\ref{edgesim}--\ref{decompsim} or an application of the reflexivity, symmetry or transitivity of $\simeq$. In all these cases, the proof of~\eqref{equalprobeq} is a straightforward consequence of the induction assumption.

We exemplify this for rule~\ref{unionsim}. By the induction hypothesis, we assume that $\sum_{v\in A_i} p(v) = \sum_{v\in B_i} p(v)$ for all $i$. But then,
\[
\sum_{v\in A} p(v) = \sum_{i=1}^n \sum_{v\in A_i} p(v) = \sum_{i=1}^n \sum_{v\in B_i} p(v) = \sum_{v\in B} p(v),
\]
as was to be shown.
\end{proof}

\begin{figure}
\subfigure[]{
\label{funnyvirt}
\begin{centering}
\begin{tikzpicture}
\node[draw,shape=circle,fill,scale=.5] at (0,0) {} ;
\node at (-0.6,0) {$v_2$};
\node[draw,shape=circle,fill,scale=.5] at (2,0) {} ;
\node at (-1.4,0) {$v_1$};
\node[draw,shape=circle,fill,scale=.5] at (-2,0) {} ;
\node at (2.6,-0) {$v_3$};
\node[draw,shape=circle,fill,scale=.5] at (2,2) {} ;
\node[draw,shape=circle,fill,scale=.5] at (-1,2) {} ;
\node[draw,shape=circle,fill,scale=.5] at (-2,-2) {} ;
\node at (-1.5,-2.5) {$w_1$};
\node[draw,shape=circle,fill,scale=.5] at (1,-2) {} ;
\node at (1.5,-2.5) {$w_2$};
\draw[thick,blue] (2,1) ellipse (.4cm and 1.6cm) ;
\draw[thick,blue] (.5,2) ellipse (1.8cm and .4cm) ;
\draw[thick,blue] (-2,-1) ellipse (.4cm and 1.6cm) ;
\draw[thick,orange,dashed] (0,0) ellipse (2.8cm and 1cm) ;
\draw[thick,orange,dashed] (-.5,-2) ellipse (1.8cm and .4cm) ;
\draw[thick,blue] plot [smooth cycle,tension=.8] coordinates { (0.3,-.1) (-2.3,-.1) (-1,2.3) } ;
\draw[thick,blue] plot [smooth cycle,tension=.8] coordinates { (-0.3,.1) (2.3,.1) (1,-2.3) } ;
\end{tikzpicture}
\end{centering}}
\hspace{1cm}
\subfigure[]{
\label{pentagonvirt}
\begin{centering}
\begin{tikzpicture}
\node[draw,shape=circle,fill,scale=.5] (a) at (0:2cm) {} ;
\node at (72:2.6cm) {$u_1$};
\node at (288:2.6cm) {$u_2$};
\node[draw,shape=circle,fill,scale=.5] (b) at (72:2cm) {} ;
\node[draw,shape=circle,fill,scale=.5] (c) at (144:2cm) {} ;
\node[draw,shape=circle,fill,scale=.5] (d) at (216:2cm) {} ;
\node[draw,shape=circle,fill,scale=.5] (e) at (288:2cm) {} ;
\draw[thick,blue,rotate=36] (0:1.63cm) ellipse (.4cm and 1.6cm) ;
\draw[thick,blue,rotate=108] (0:1.63cm) ellipse (.4cm and 1.6cm) ;
\draw[thick,blue,rotate=180] (0:1.63cm) ellipse (.4cm and 1.6cm) ;
\draw[thick,blue,rotate=252] (0:1.63cm) ellipse (.4cm and 1.6cm) ;
\draw[thick,blue,rotate=324] (0:1.63cm) ellipse (.4cm and 1.6cm) ;
\draw[thick,orange,dashed] (.6cm,0) ellipse (.6cm and 2.2cm) ;
\end{tikzpicture}
\end{centering}}
\caption{Examples of contextuality scenarios with virtual edges (dashed).}
\label{virtualmeasfig}
\end{figure}
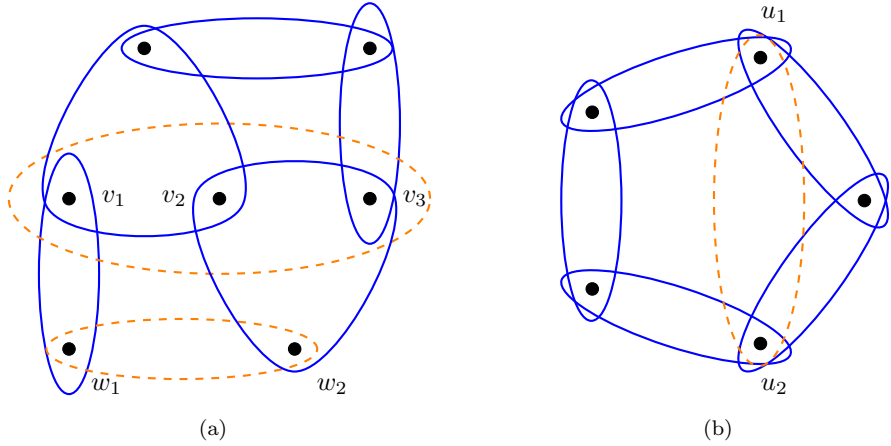

\begin{defn}
A \emph{virtual edge} is a set of outcomes $e'\subseteq V(H)$ equivalent, $e'\simeq e$, to some edge $e\in E(H)$.
\end{defn}

In fact, if $e'$ is equivalent to some edge, then it is automatically equivalent to any edge.

For example, any edge is also a virtual edge. Less trivially, the dashed regions in Figure~\ref{virtualmeasfig} form virtual edges.

The crucial point is that any virtual edge also satisfies the normalization of probability: by Proposition~\ref{equalprobprop}, being equivalent to an edge implies that the total probability under any probabilistic model is the same as that of an edge, which is $1$. In terms of the interpretation of vertices as outcomes and edges as measurements, a virtual edge is a set of outcomes which `wants to be' a measurement.

The virtual edges on $H$ form themselves a contextuality scenario which we denote by $\bar{H}$ and call the \emph{completion} of $H$. By construction, we have $V(H)=V(\bar{H})$ and $E(H)\subseteq E(\bar{H})$. We call $H$ \emph{complete} if $\bar{H}=H$. For any $H$, the completion $\bar{H}$ is complete.

These definitions imply that if $p$ is a probabilistic model on $H$, then it is also a probabilistic model on $\bar{H}$, and vice versa. In other words, we have $\mathcal{G}(\bar{H})=\mathcal{G}(H)$. Similar equalities hold for the sets of classical and quantum models, $\mathcal{C}(H)$ and $\mathcal{Q}(H)$, but not for $\CE^1(H)$:

\begin{prop}
\label{HHbar}
\begin{enumerate}
\item For any scenario $H$, we have
\[
\mathcal{C}(\bar{H}) = \mathcal{C}(H),\qquad \mathcal{Q}(\bar{H}) = \mathcal{Q}(H),\qquad \mathcal{G}(\bar{H}) = \mathcal{G}(H),
\]
and
\[
\mathcal{Q}_{n+1}(H)\subseteq \mathcal{Q}_n(\bar{H}) \subseteq \mathcal{Q}_n(H).
\]
\item\label{NOhierarchy} 
If $\mathrm{NO}(\bar{H}) = \mathrm{NO}(H)$, then
\[
\mathcal{Q}_n(\bar{H}) = \mathcal{Q}_n(H),\qquad \CE^n(\bar{H}) = \CE^n(H).
\]
\item\label{NOcompl} However, there is a scenario $H$ for which
\[
 \CE^1(\bar{H}) \subsetneq \CE^1(H).
\]
\end{enumerate}
\end{prop}

The last statement also shows that if one assigns projections $P_v$ to the vertices satisfying the normalization condition $\sum_{v\in e}P_v=\mathbbm{1}$, then there may be orthogonality relations between these projections which cannot be read off from the non-orthogonality graph $\mathrm{NO}(H)$ alone.

We suspect that there also exist $\bar{H}$ with $\mathcal{Q}_1(\bar{H}) \subsetneq \mathcal{Q}_1(H)$, but we have not yet constructed any examples of this.

\begin{proof}
\begin{enumerate}
\item We already saw above that $\mathcal{G}(\bar{H}) = \mathcal{G}(H)$. Since one obtains the set of classical models by restricting to the deterministic models on each side and taking their convex hull, this immediately implies $\mathcal{C}(\bar{H})=\mathcal{C}(H)$.

Concerning quantum models, the equality follows from a quantum analogue of Proposition~\ref{equalprobprop}: for any assignment of a projection $P_v$ to each $v\in V(H)$ satisfying the normalization equation $\sum_{v\in e} P_v = \mathbbm{1}$ for all $e\in E(H)$, an equivalence $A\simeq B$ for $A,B\subseteq V(H)$ implies that
\beq
\label{projAB}
\sum_{v\in A} P_v = \sum_{v\in B} P_v .
\eeq
Again, this can be proven in the same way as Proposition~\ref{equalprobprop} by using induction on the proof of $A\simeq B$ and showing that each of the steps~\ref{simdef}\ref{edgesim}--\ref{decompsim} preserves equations of type~\eqref{projAB}. As a special case, we therefore have that $\sum_{v\in e'} P_v = \mathbbm{1}$ for every virtual edge $e'$, which implies the claim.

We now consider the sets $\mathcal{Q}_n$ of the semidefinite hierarchy, which display a more subtle behavior. The second inclusion $\mathcal{Q}_n(\bar{H}) \subseteq \mathcal{Q}_n(H)$ is clear since the definition of $\mathcal{Q}_n(H)$ can be regarded as a relaxation of $\mathcal{Q}_n(\bar{H})$ in the sense that there are (possibly) less constraints on the matrices $M$; the first inclusion $\mathcal{Q}_{2n}(H)\subseteq\mathcal{Q}_n(\bar{H})$ is much trickier to prove, and we need to get our hands dirty by working with the explicit properties of $M$ of Lemma~\ref{npan} and Remark~\ref{npaproplist}.

So assume that $p\in\mathcal{Q}_{n+1}(H)$, so that there exists a positive semidefinite matrix $M$ with entries $M_{\mathbf{v},\mathbf{w}}$ indexed by $\mathbf{v},\mathbf{w}\in V(H)^{*(n+1)}$ satisfying $p(v)=M_{\emptyset,v}$ and having the other properties discussed in Section~\ref{dfnhierarchy}. We then claim that $p\in\mathcal{Q}_n(\bar{H})$ is witnessed by the very same matrix, restricted to entries indexed by $\mathbf{v},\mathbf{w}\in V(H)^{*n}$. To this end, we need to show that this matrix has the required properties of Lemma~\ref{npan}: first, for $\mathbf{v}\in V(H)^{*n}$ and $\mathbf{w}\in V(H)^{*n}$ and any virtual edge $e'\simeq e\in E(H)$, we will prove that
\beq
\label{Mvirtual}
\sum_{x\in e'} M_{\mathbf{v}x,\mathbf{w}} = M_{\mathbf{v},\mathbf{w}}.
\eeq
With $e$ in place of $e'$, we know that this equation holds; the current one then follows upon using the fact that for any $A\simeq B$, we have
\[
\sum_{x\in A} M_{\mathbf{v}x,\mathbf{w}} = \sum_{x\in B} M_{\mathbf{v}x,\mathbf{w}}.
\]
This fact in turn can be shown using the exact same kind of computation as for $\mathcal{G}$ and $\mathcal{Q}$.

Second, we need to show that if $\mathbf{v}=v_1\ldots v_k$ and $\mathbf{w}=w_1\ldots w_m$ for $k,m\leq n$, then $v_k\perp w_m$ in $\bar{H}$ implies that $M_{\mathbf{v},\mathbf{w}}=0$. By~\eqref{npahankel}, this is equivalent to showing $M_{\mathbf{v}w_m,w_1\ldots w_{m-1}}=0$, which by~\eqref{bounddiag} follows from
\[
M_{\mathbf{v}w_m,\mathbf{v}w_m} = 0.
\]
In order to see this, we evaluate
\beq
\label{Mcomp}
M_{\mathbf{v},\mathbf{v}} \stackrel{\eqref{Mvirtual}}{=} \sum_{x\in e'} M_{\mathbf{v}x,\mathbf{v}} = \sum_{x\in e'} M_{\mathbf{v}x,\mathbf{v}x},
\eeq
where $e'$ is a virtual edge containing both $v_k$ and $w_m$; such an $e'$ exists since $v_k\perp w_m$. In the second step, we have used $M_{\mathbf{v}x,\mathbf{v}}=M_{\mathbf{v}x,\mathbf{v}x}$, which can be seen by choosing some edge $e\ni x$ in $H$, computing $M_{\mathbf{v}x,\mathbf{v}}=\sum_{y\in e} M_{\mathbf{v}x,\mathbf{v}y}$, and noticing that all terms drop out, except for $M_{\mathbf{v}x,\mathbf{v}x}$, due to the assumption~\eqref{npaort}. The sum on the right-hand side of~\eqref{Mcomp} contains the term $M_{\mathbf{v}v_k,\mathbf{v}v_k}$, which is actually equal to the left-hand side because of a similar argument. Since all other summands are non-negative, it follows that all these other summands must vanish, and so does $M_{\mathbf{v}w_m,\mathbf{v}w_m}$ in particular.

In this argument, we have never needed to consider any matrix entry indexed by a string of length greater than $n+1$. Hence the claim follows.

\item We already know $\mathcal{Q}_n(\bar{H})\subseteq\mathcal{Q}_n(H)$, so it remains to show the converse inclusion, which we do in the same way as in the previous paragraphs. The equation~\eqref{Mvirtual} follows as before. The equation $M_{\mathbf{v},\mathbf{w}}=0$ for $v_k\perp w_m$ in $\mathrm{NO}(\bar{H})$ follows from the assumption $\mathrm{NO}(\bar{H})=\mathrm{NO}(H)$ together with the assumptions on $M$. The claim for $\CE^n$ is also straightforward.

\item It is clear that $\CE^1(\bar{H})\subseteq\CE^1(H)$ for all $H$, since $\mathrm{NO}(\bar{H})$ is a subgraph of $\mathrm{NO}(H)$. To show that this inclusion is strict in general, it is sufficient to find a scenario $H$ together with a probabilistic model $p\in\CE^1(H)$ and a virtual edge $e'$ such that $\sum_{v\in e'} p(v') > 1$. The possibly simplest example is illustrated in Figure~\ref{pentagonvirt}: the only probabilistic model $p$ is the one which assigns a weight of $\tfrac{1}{2}$ to each vertex, and it clearly satisfies Consistent Exclusivity, so that $p\in\CE^1(H)$. On the other hand, if one takes the dashed virtual edge into account, Consistent Exclusivity is violated, so that $p\not\in\CE^1(\bar{H})$.
\qedhere
\end{enumerate}
\end{proof}

In this sense, every contextuality scenario is observationally equivalent to its completion as far as $\mathcal{C}$, $\mathcal{Q}$ and $\mathcal{G}$ are concerned, while the completion may put additional constraints with respect to Consistent Exclusivity and possibly also with respect to the semidefinite hierarchy. It follows that for many purposes, we can consider complete scenarios only without loss of generality. If we have two scenarios on the same underlying set of outcomes such that every edge in the first is a virtual edge in the second and vice versa, then the completions of these two scenarios coincide, and we consider these two scenarios observationally equivalent.

\subsection{Higher Foulis--Randall products}
\label{higherprodssub}

We now return to the products of more than two contextuality scenarios introduced in Section~\ref{higherprodsmain}, where we had encountered the unpleasant situation that there are different choices of products of more than two scenarios. Our goal here is to show that all these choices are equivalent in the sense that their completion does not depend on the particular choice of product. By the results of the previous section, this means that these different choices of product do not need to be distinguished for most of the purposes of this paper.

We fix scenarios $H_1,\ldots,H_n$ of which we want to study the product and start by deriving some auxiliary results about equivalence on $^{\min}\otimes_{i=1}^n H_i$.

\begin{lem}
\label{marginals}
For any party $k$, if $v\in V(H_k)$ is an arbitrary vertex and $e_i,e'_i\in E(H_i)$ are arbitrary edges for all $i\neq k$, then
\[
(\{v\}\times \vec{e}) \simeq (\{v\}\times \vec{e}'),
\]
where we write $\vec{e}:=\prod_{i\neq k} e_i$ and $\{v\}$ refers to the new component at party $k$.
\end{lem}

Together with Proposition~\ref{equalprobprop}, this statement has the particular consequence that the marginal probabilities of party $i$ are well-defined, i.e.~independent of the measurements which are jointly conducted by the other parties, for any probabilistic model.

\begin{proof}
We assume that $e_i = e'_i$ for all $i\neq k$, with one possible exception at some party $j$; proving the claim in this case then implies the general case by transitivity of $\simeq$.

But then, since we consider $^{\min}\otimes_{i=1}^n H_i$, there is an edge representing the joint measurement in which party $j$ chooses their measurement as a function of $k$'s outcome such that $j$ measures $e'_j$ if this outcome is $v$, and measures $e_j$ otherwise. On the other hand, there is another edge in which $j$'s measurement is always $e_j$, independent of what the others' outcomes may be. The claim then follows from applying rule~\ref{edgesim} of Definition~\ref{simdef} to these two edges, and then rule~\ref{decompsim}.
\end{proof}

\begin{lem}
\label{timesv}
Let $S$ be any subset of the parties and $k\in S$. If $A\simeq B$ for $A,B\subseteq V\left(^{\min}\otimes_{i\in S\setminus\{k\}} H_i\right)$, then also
\beq
\label{timesveq}
(\{v\}\times A)\simeq (\{v\}\times B)
\eeq
with respect to $^{\min}\otimes_{i\in S} H_i$ for any $v\in H_k$. 
\end{lem}

\begin{proof}
Again, we use the same technique as in the proof of Proposition~\ref{equalprobprop}: induction over the length of the proof of the equivalence $A\simeq B$.

The base cases are when the proof of the equivalence is reflexivity, i.e.~simply $A\simeq A$, or if both $A$ and $B$ are edges. In the former case,~\eqref{timesveq} trivially follows also by reflexivity; the latter case requires a bit more work. $A$ must be of the form~\eqref{nFRedge2}, meaning that
\[
A = \bigcup_{\vec{v}'} \{\vec{v}'\} \times f(\vec{v}')
\]
where some party $j\in S\setminus\{k\}$ measures last and chooses their measurement as a function $f(\vec{v}')\in E(H_j)$ of the others' outcomes $\vec{v}'$, where $\vec{v}'$ ranges over $\prod_{i\in S\setminus\{j,k\}} e_i$. Now it is enough to show that
\beq
\label{toprove}
\left(\{v\}\times \left(\bigcup_{\vec{v}'} \{\vec{v}'\}\times f(\vec{v}') \right)\right)\simeq \left(\{v\}\times \left(\bigcup_{\vec{v}'} \{\vec{v}'\}\times \widehat{e}_j\right)\right),
\eeq
where $\widehat{e}_j$ is some fixed edge: the same can then be done for $B$, and then we can apply symmetry and transitivity of $\simeq$ to obtain the claim upon noting that the right-hand side is a product of the singleton set $\{v\}$ with a product of edges, as in Lemma~\ref{marginals}. But the equivalence~\eqref{toprove} follows from a similar trick as in the proof of Lemma~\ref{marginals}: choosing an arbitrary $e_k\in E(H_k)$ with $v\in e_k$ and adding the set
\[
(e_k\setminus\{v\}) \times \left(\bigcup_{\vec{v}'} \{\vec{v}'\} \times \widehat{e}_j\right),
\]
to both sides of the putative equivalence~\eqref{toprove} as a disjoint union proves this equivalence by rule~\ref{decompsim}, since both sides of~\eqref{toprove} then become themselves edges of $^{\min}\otimes_{i\in S} H_i$ in which party $j$ conducts their measurement as a function of the others' outcome, this time including party $k$.

The different induction steps comprising applications of rules~\ref{unionsim} and~\ref{decompsim} as well as symmetry and transitivity of $\simeq$ are again straightforward.
\end{proof}

\begin{thm}
Any edge in $^{\max}\otimes_{i=1}^n H_i$ is a virtual edge in $^{\min}\otimes_{i=1}^n H_i$.
\end{thm}

\begin{proof}
We prove the slightly more general statement that for any non-empty subset of parties $S\subseteq\{1,\ldots,n\}$, any edge in $^{\max}\otimes_{i\in S} H_i$ is a virtual edge in $^{\min}\otimes_{i\in S} H_i$, using induction on $|S|$. For $|S|=1$, there is nothing to prove.

For $|S|>1$, we pick any edge in $^{\max}\otimes_{i\in S} H_i$ given as the outcome set $O(\mathcal{P})$ of a measurement protocol $\mathcal{P}=(k,e,f)$ for $S$. By the induction hypothesis, the edges associated to all subprotocols $f(v)$ for $S\setminus\{k\}$ are virtual edges in $^{\min}\otimes_{i\in S\setminus\{k\}} H_i$, so that $f(v)\simeq e'$ for any fixed edge $e'\in E\left(^{\min}\otimes_{i\in S\setminus\{k\}} H_i\right)$. By Lemma~\ref{timesv}, this implies that also
\[
(\{v\}\times f(v))\simeq (\{v\}\times e').
\]
Since this holds for all $v\in e$, we can apply $\bigcup_{v\in e}$ on both sides, and the claim then follows from rule~\ref{unionsim} and the fact that $e\times e'\in E\left(^{\min}\otimes_{i\in S} H_i\right)$.
\end{proof}

Together with Lemma~\ref{allintermed}, this immediately implies the main result of this section: 

\begin{thm}
\label{equalcomp}
\begin{enumerate}
\item The completions of $^{\min}\otimes_{i=1}^n H_i$ and $^{\max}\otimes_{i=1}^n H_i$ and of any iterated binary product $H_1\otimes\ldots\otimes H_n$ with any bracketing all coincide.
\item The sets of probabilistic models $\mathcal{C}$, $\mathcal{Q}$ and $\mathcal{G}$ for a product scenario do not depend on the particular choice of the product.
\end{enumerate}
\end{thm}

Unfortunately, the analogous result does not hold for non-orthogonality graphs: 

\begin{prop}
\label{NOnoLO}
\begin{enumerate}
\item Two vertices $\vec{u} = (u_1,\ldots,u_n)$ and $\vec{v}=(v_1,\ldots,v_n)$ are orthogonal in any iterated binary product $H_1\otimes\ldots\otimes H_n$ if and only if they are orthogonal in $^{\max}\otimes_{i=1}^n H_i$ and if only if they are `locally orthogonal', i.e.~there exists an index $i$ such that $u_i \perp v_i$.
\item However, this does not apply to $^{\min}\otimes_{i=1}^n H_i$: there are scenarios $H_1$, $H_2$ and $H_3$ for which $\mathrm{NO}(^{\min}\otimes_{i=1}^3 H_i) \neq \mathrm{NO}(^{\max}\otimes_{i=1}^3 H_i)$.
\end{enumerate}
\end{prop}

\begin{proof}
\begin{enumerate}
\item We already know that $E\left(H_1\otimes\ldots\otimes H_n\right)\subseteq E(^{\max}\otimes_{i=1}^n H_i)$, so that $\vec{u}\perp\vec{v}$ with respect to $H_1\otimes\ldots\otimes H_n$ implies the same with respect to $^{\max}\otimes_{i=1}^n H_i$.

First, we show that $u_i\perp v_i$ implies that $\vec{u}\perp\vec{v}$ with respect to $H_1\otimes\ldots\otimes H_n$. After permuting the scenarios if necessary, our iterated binary product is of the form
\[
(H_1\otimes\ldots\otimes H_j)\otimes (H_{j+1}\otimes\ldots\otimes H_n),
\]
where each factor is again a (smaller) iterated binary product. So by induction, the claim can be reduced to the statement that local orthogonality implies orthogonality in the case of a binary Foulis--Randall product $H_A\otimes H_B$. Then if $\vec{u}=(u_A,u_B)$ and $\vec{v}=(v_A,v_B)$, $u_A\perp v_A$ means that there is an $e\in E(H_A)$ with $u_A,v_A\in e$. Now we can choose some $e_u\in E(H_B)$ and $e_v\in E(H_B)$ with $u_B\in e_u$ and $v_B\in e_v$. Then, the protocol in which Alice first measures $e$ and then Bob measures $e_u$ if Alice obtained $u_A$ and $e_v$ otherwise results in the edge
\[
\left(\{u_A\}\times e_u\right) \cup \left( e\setminus\{u_A\} \times e_v \right) \:\in E(H_A\otimes H_B).
\]
Since both $(u_A,u_B)$ and $(v_A,v_B)$ lie in this edge, we conclude that these two vertices are orthogonal in $H_A\otimes H_B$ as claimed.

Conversely, we show that if $\vec{u}\perp\vec{v}$ with respect to $^{\max}\otimes_{i=1}^n H_i$, then there is a party $i$ for which $u_i\perp v_i$. So let $\mathcal{P}$ be a measurement protocol with $\vec{u},\vec{v}\in O(\mathcal{P})$. We prove the claim by induction on $n$. For $n=1$, it is trivial. In general, we know that the protocol is of the form $\mathcal{P}=(k,e,f)$ for $k$ the party that measures first, $e\in E(H_k)$ the measurement which that party conducts, and $f$ assigning to each outcome of $e$ a subsequent protocol for the remaining parties. The assumption $\vec{u},\vec{v}\in O(\mathcal{P})$ implies that $u_k,v_k\in e$; so if $u_k\neq v_k$, then $u_k\perp v_k$ and we are done. Otherwise, if $u_k=v_k$, then both $(u_1,\ldots,\bcancel{u_k},\ldots,u_n)$ and $(v_1,\ldots,\bcancel{v_k},\ldots,v_n)$ are outcomes of the resulting subprotocol $f(u_k)$. In this case, the claim follows from the induction hypothesis.
\item We take the three scenarios $H_A$, $H_B$ and $H_C$ from Figure~\ref{nonassfig}. Then the two triples $(a_1,b_1,c_1)$ and $(a_2,b_3,c_3)$ are locally orthogonal and therefore orthogonal in the maximal product. On the other hand, for these two triples to belong to one common edge in the minimal tensor product, it would have to be the case that for at least two of the parties, there is an edge which contains the corresponding component of each triple. Since this fails for both parties $B$ and $C$, this is not the case, and therefore $(a_1,b_1,c_1)\not\perp (a_2,b_3,c_3)$ in the minimal product. \qedhere
\end{enumerate}
\end{proof}

In the sense of these results, the Foulis--Randall product is associative: while it is not associative `on the nose', as exemplified in Proposition~\ref{nonass}, it is associative for all practical purposes in the sense that the resulting product scenarios are observationally equivalent. In particular, it makes sense to speak of \textit{the} completion of the $n$-fold product scenario, without specifying \textit{which} product one refers to. In fact, we suspect that this completion can be computed directly in terms of the completions of the individual scenarios:

\begin{conj}
\[
\overline{H_1\otimes\ldots\otimes H_n} = {}^{\max}\otimes_{i=1}^n \bar{H}_i.
\]
\end{conj}

\newpage
\section{\textbf{Relation to the observable-based approach}}
\label{reltosheaf}

The observable-based approach to quantum contextuality and nonlocality has first been studied explicitly and in complete generality by Abramsky and Brandenburger~\cite{AB}. It was used much earlier in a different mathematical context by Vorob'ev~\cite{Vorob}. See also~\cite{LSW,FC}, where similar definitions have been used. In this section, our goal is to sketch how the observable-based approach can be embedded into our formalism. A converse construction should be possible upon augmenting the observable-based approach by additional constraints as in~\cite[Sec.~7]{AB}. In this sense, the two formalisms are essentially equivalent. We believe that both approaches have their merits; for example, in both cases, the relation to sophisticated mathematical methods can be exploited. In the observable-based approach, this has been done in~\cite{AMS}; for the hypergraph-based approach, this has been started in~\cite{CSW} and further developed in this paper.

\subsection{Definitions for the observable-based approach} The following definition blends the terminology of~\cite{AB} with the one of~\cite{FC}; the actual content is the same regardless.

\begin{defn}
A \emph{marginal scenario} $(X,O,\mathcal{M})$ is a finite set $X$, the elements of which we call \emph{observables}, together with a finite set $O$ of outcomes and a \emph{measurement cover} $\mathcal{M}$, which is a family of subsets $\mathcal{M}\subseteq 2^X$ such that
\begin{enumerate}
\item every element of $X$ occurs in some $C$, i.e.~$\bigcup_{C\in\mathcal{M}} C = X$.
\item $\mathcal{M}$ is an anti-chain: for any $C,C'\in\mathcal{M}$, if $C\subseteq C'$, then $C=C'$.
\end{enumerate}
The $C\in\mathcal{M}$ are called \emph{measurement contexts}.
\end{defn}

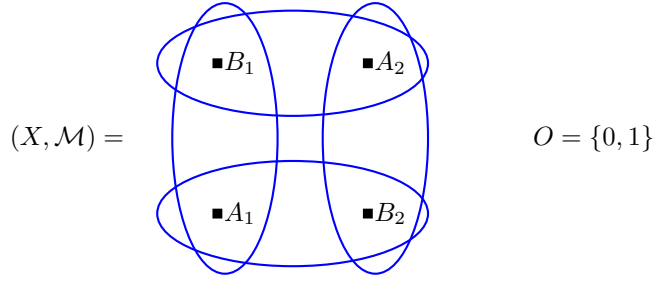
\begin{figure}
\begin{tikzpicture}[scale=1]
\node at (-2,1) {$(X,\mathcal{M})=$};
\node[draw,shape=rectangle,fill,scale=.5] (a1) at (0,0) {} ;
\node[draw,shape=rectangle,fill,scale=.5] (b1) at (0,2) {} ;
\node[draw,shape=rectangle,fill,scale=.5] (b2) at (2,0) {} ;
\node[draw,shape=rectangle,fill,scale=.5] (a2) at (2,2) {} ;
\node[right of=a1,node distance=3mm] {$A_1$} ;
\node[right of=b1,node distance=3mm] {$B_1$} ;
\node[right of=a2,node distance=3mm] {$A_2$} ;
\node[right of=b2,node distance=3mm] {$B_2$} ;
\foreach \x in {0,2} \draw[thick,blue] (\x+.1,1) ellipse (.7cm and 1.8cm) ;
\foreach \y in {0,2} \draw[thick,blue] (1,\y) ellipse (1.8cm and .7cm) ;
\node at (5,1) {$O = \{0,1\}$} ;
\end{tikzpicture}
\caption{The CHSH scenario as a marginal scenario. We now draw the vertices as squares in order to indicate that the interpretation differs from the one of all other illustrations of hypergraphs in this paper.}
\label{obsCHSH}
\end{figure}

From the mathematical point of view, the maximal sets of compatible observables are a hypergraph precisely as in Definition~\ref{defncs}, but the physical interpretation is quite different. A subset $C\subseteq X$ with $C\in\mathcal{M}$ is to be thought of as a maximal set of jointly measurable observables. See Figure~\ref{obsCHSH} for an example, which shows the CHSH scenario as a marginal scenario with observables $A_1,A_2,B_1,B_2$ where the four pairs
$$
\{A_1,B_1\},\quad \{A_1,B_2\},\quad \{A_2,B_1\},\quad \{A_2,B_2\}
$$
are jointly measurable, but no other pairs or triples of observables are jointly measurable. In particular, these four pairs also are the maximal sets of jointly measurable observables and thereby form the measurement cover
\[
\mathcal{M} = \{\{A_1,B_1\},\{A_1,B_2\},\{A_2,B_1\},\{A_2,B_2\}\}.
\]

As is common practice with many other mathematical structures, we denote a marginal scenario $(X,O,\mathcal{M})$ simply by $X$, at least when $O$ and $\mathcal{M}$ are clear from the context.

As noted in~\cite{AB}, it is not a substantial restriction to assume that all observables take values in the same set of outcomes $O$. We assume this mainly for convenience of notation and note that all of our considerations and results can easily be extended to the general case in which each observable $A\in X$ takes values in an associated finite set of outcomes $O_A$ depending on $A$.

In the following, we would like to associate a contextuality scenario $H[X]$ to a marginal scenario $X$. In order to do so, we need to consider measurements of compatible observables which are conducted in a certain temporal order. Assume that we have already measured some observable $A\in X$; then is it possible to define a marginal scenario which encodes all the possibilities for subsequent measurements compatible with $A$? The following notion achieves this:

\begin{defn}
Given an observable $A\in X$, the \emph{induced marginal scenario} $X\{A\}$ is the marginal scenario having observables
$$
X\{A\} \defin \left\{ A' \in X \:\big|\: A'\neq A,\: \exists C\in\mathcal{M} \textrm{ s.t. } \{A,A'\}\subseteq C \right\} 
$$
and measurement contexts all the $C\setminus\{A\}$ for those $C\in\mathcal{M}$ with $A\in C$.
\end{defn}

The idea in considering only subsequent measurements $A'$ that are compatible with $A$ is that measuring one of these refines the information obtained via the first measurement. By definition, any $X\{A\}$ has a smaller number of observables than the original $X$. In particular, iterating this construction by taking an induced marginal scenario of an induced marginal scenario etc., one eventually ends up with an empty scenario, and the process terminates.

This termination property allows us to make the following recursive definition:

\begin{defn}
A \emph{measurement protocol} $T$ on a marginal scenario $X$ is
\begin{enumerate}
\item $T=\emptyset$ if $X=\emptyset$;
\item otherwise, $T=(A,f)$, where $A\in X$ is an observable and $f:O\to \mathrm{MP}(X\{A\})$ is a function, where $ \mathrm{MP}(X\{A\})$ is the set of all measurement protocols on the scenario $X\{A\}$.
\end{enumerate}
\end{defn}

Intuitively, a measurement protocol consists of a choice of observable and an assignment of a new measurement protocol to each outcome of the observable, where the new measurement protocol lives on the induced marginal scenario. 

Upon unraveling the recursive structure of this definition, one finds that a measurement protocol specifies sequences of measurements which can be applied to the system, where the choices of subsequent measurements $f$ are allowed to depend on the outcomes of all earlier ones. These measurement sequences have the additional property that all measurements in a sequence are compatible and that no measurement can occur twice in the same sequence. Due to the allowed dependence of later measurements on earlier outcomes, a measurement protocol has a tree-like structure\footnote{The term `\href{https://en.wikipedia.org/wiki/Decision_tree}{decision tree}' comes to mind.}, and we denote a measurement protocol by the letter `$T$' in order to indicate this. Note that every measurement sequence is automatically maximal in the sense that it contains all observables of a certain measurement context, since the measurement protocol can end only at a stage at which the induced marginal scenario is empty.

The set of outcomes $\mathrm{Out}(T)$ of a measurement protocol $T$ is also defined recursively: if $T=\emptyset$, then there is only a single outcome which we denote by `$\ast$', so that $\mathrm{Out}(\emptyset) = \{\ast\}$. Otherwise, we have $T = (A,f)$, and then we put
$$
\mathrm{Out}(T) \defin \left\{\, (A,a,\alpha') \::\: a\in O,\: \alpha'\in \mathrm{Out}(f(a)) \,\right\} .
$$
In words: an outcome of a measurement protocol $T=(A,f)$, where $A$ is the initial measurement and $f:O\to\mathrm{MP}(X\{A\})$ assigns to each of its outcomes a subsequent measurement protocol, is a triple $\alpha=(A,a,\alpha')$ consisting of a record remembering the initial observable $A$, an outcome $a$ of $A$, and an outcome $\alpha'$ of the subsequent protocol $f(a)$. Upon unraveling this recursive definition, one finds that an outcome $\alpha\in\mathrm{Out}(T)$ corresponds to a measurement sequence in $T$ together with an associated sequence of outcomes for these measurements, such that applying the protocol to any outcome in the sequence results in the following measurement (except for the last outcome in the sequence, where the protocol ends).

One can use recursion as follows to assign to every protocol outcome $\alpha\in\mathrm{Out}(T)$ a unique measurement context $C(\alpha)\in\mathcal{M}$ in which the outcome lives: if $\alpha=(A,a,\alpha')$ as above, and $\alpha'$ lives in the context $C(\alpha')$ on $X\{A\}$, then $\alpha$ lives in the context $C(\alpha):=\{A\}\cup C(\alpha')$ on $X$. Furthermore, one can associate in the obvious way to every such $\alpha$ an assignment of outcomes $s(\alpha)\in O^{C(\alpha)}$. In total, the protocol outcome $\alpha$ has associated to it a context $C(\alpha)$ together with an assignment of outcomes $s(\alpha)\in O^{C(\alpha)}$ in this context. Constructing $C(\alpha)$ and $s(\alpha)$ remembers which observables were measured during the protocol realization and which outcomes were obtained, but it forgets the temporal order in which these measurements were conducted.

\begin{defn}
\label{HXcorr}
The contextuality scenario $H[X]$ associated to a marginal scenario $X$ has vertices
$$
V(H[X]) \defin \left\{ \: (C,s) \::\: C\in\mathcal{M},\: s\in O^C  \: \right\} 
$$
and every measurement protocol $T$ on $X$ defines an edge given as the set of vertices
$$
e_T \defin \left\{ \: (C(\alpha),s(\alpha)) \::\: \alpha\in\mathrm{Out}(T) \: \right\},
$$
so that $E(H[X]) \defin \{\: e_T \: :\: T\in \mathrm{MP}(X) \: \}$.
\end{defn}

In particular, every vertex of $H[X]$ belongs to exactly one measurement context $C$. This parallels and generalizes the situation for Bell scenarios $B_{n,k,m}$ from Section~\ref{bellscen}, since a context in a Bell scenario is a choice of setting for every party.

We write $P$ for an \emph{empirical model} on $X$~\cite{AB}. This means that for each measurement context $C\in\mathcal{M}$, we have a probability distribution $P_C$ over $O^C$, such that the \emph{sheaf condition}\footnote{Although used as such in~\cite{AB} and all follow-up works, the term `sheaf condition' is actually a misnomer, since~\eqref{sheafcond} is a \emph{presheaf} condition, which constitutes a significant difference~\cite{MM}.} holds:
\beq
\label{sheafcond}
P_{C|C\cap C'} = P_{C'|C\cap C'} \quad \forall C,C'\in\mathcal{M} ,
\eeq
where $P_{C|C\cap C'}$ stands for the marginal distribution of $P_C$ associated to the observables in $C\cap C'$. For an assignment of outcomes $s\in O^C$, the probability $P_C(s)$ is to be thought of as the probability of obtaining the joint outcome $s$ when jointly measuring all observables in $C$. The sheaf condition is a generalization of the no-signaling condition.

\subsection{Correspondence to our approach}

To an empirical model $P$ we associate a probabilistic model on the contextuality scenario $H[X]$ by setting, for each $C\in\mathcal{M}$ and each $s\in O^C$,
\beq
\label{pP}
p(C,s) \defin P_C(s) .
\eeq
It needs to be verified that this actually is a probabilistic model, i.e.~that these probabilities are suitably normalized for every edge in $E[X]$.

Conversely, given a probabilistic model $p$ on $H[X]$, we can read~(\ref{pP}) the other way around in order to define an empirical model $P$ on $X$ in terms of $p$.

\begin{thm}
This defines a linear bijection between empirical models on $X$ and probabilistic models on $H[X]$.
\end{thm}

This bijective correspondence generalizes Proposition~\ref{Bnosig}: every Bell scenario is a marginal scenario in the obvious way~\cite{AB}, and it can be checked that applying Definition~\ref{HXcorr} in this case recovers the contextuality scenarios $B_{n,k,m}$ of Section~\ref{bellscen}.

\begin{proof}
We first verify that~(\ref{pP}) turns an empirical model $P$ into a probabilistic model $p$. It needs to be shown that
\beq
\label{pPnorm}
\sum_{\alpha\in\mathrm{Out}(T)} P_{C(\alpha)}(s(\alpha)) = 1
\eeq
for any measurement protocol $T$. In order to prove this, we introduce the notion of \emph{post-measurement} empirical model. Suppose that a measurement of an observable $A\in X$ has resulted in an outcome $a\in O$. Then all subsequent measurements compatible with $A$ live in the scenario $X\{A\}$. For such subsequent measurements, we expect the posterior probabilities
$$
P^{\mathrm{post}(a)}_{C\setminus\{A\}}(s) = \frac{P_C(s\cup\{a\})}{P_{\{A\}}(a)}.
$$
for all outcome assignments $s\in O^{C\setminus\{A\}}$, and we write $s\cup\{a\}$ for the outcome assignment in $O^C$ which maps $A\mapsto a$ and behaves like $s$ otherwise. It is straightforward to check that this defines an empirical model on $X\{A\}$.

We now use induction on the size of $X$ in order to prove~(\ref{pPnorm}). The base case is $X=\emptyset$, in which there is nothing to prove. For the induction step, we decompose $T = (A,f)$ and use the induction hypothesis on each $P^{\mathrm{post}(a)}$ for those $a\in O$ with $P_{\{A\}}(a) \neq 0$. Then
$$
\sum_{\alpha\in\mathrm{Out}(T)} P_{C(\alpha)}(s(\alpha)) = \sum_{a} \sum_{\alpha'\in\mathrm{Out}(f(a))} P_{\{A\}}(a)\, P^{\mathrm{post}(a)}_{C(\alpha')}(s(\alpha')) = \sum_{a} P_{\{A\}}(a) = 1,
$$
where the second to last step uses the induction hypothesis.

Conversely, we need to prove that if $p$ is a probabilistic model on $H[X]$, then the associated $P$ is an empirical model, i.e.~that it satisfies~(\ref{sheafcond}). In the case that $C\cap C' = \emptyset$, this follows from the normalization of probability $\sum_{s\in O^C} P_C(s)$, which in turn is an easy consequence of the normalization relation $\sum_{(C',s)\in e_T} p(C',s)=1$, where the measurement protocol $T$ consists of measuring all observables in the original context $C$ in a fixed but arbitrary order.

Now consider the case $C\cap C'\neq \emptyset$, and let $s_0\in O^{C\cap C'}$ be an arbitrary assignment of outcomes to the observables in $C\cap C'$. Then we consider a measurement protocol $T$ given by conducting the measurements in $C\cap C'$ in an arbitrary order, and then conducting the measurements $C\setminus C'$ if the joint outcome was $s_0$, and conducting the measurements $C'\setminus C$ otherwise. Then the normalization equation associated to this measurement protocol reads
$$
\sum_{t\in O^{C\setminus C'}} p(C,s_0\cup t) + \sum_{s_0\neq s\in O^{C\cap C'}} \sum_{t\in O^{C'\setminus C}} p(C',s\cup t) = 1 .
$$
Comparing this with the normalization equation associated to the measurement protocol which simply measures all observables in $C'$ and outputs their joint outcome,
$$
\sum_{s\in O^{C\cap C'}} \sum_{t\in O^{C'\setminus C}} p(C',s\cup t) = 1 ,
$$
gives, upon splitting the latter equation into an $s=s_0$ part and an $s\neq s_0$ part,
$$
\sum_{t\in O^{C\setminus C'}} p(C,s_0\cup t) = \sum_{t\in O^{C'\setminus C}} p(C',s_0\cup t) ,
$$
which is the equation that was to be shown.
\end{proof}

It is instructive to analyze how this generalizes the proof of Proposition~\ref{Bnosig}.

There are analogous correspondence theorems for quantum models and classical models. Since these are perfectly analogous both in the statement and in the proof, we do not discuss them further.

\newpage

\clearpage

\bibliography{AFLS}

\end{document}